\documentclass[a4paper, reqno]{amsart}
\usepackage{graphicx}
\usepackage{color,xcolor,graphicx}
\usepackage{amsmath,amsfonts,amsthm,amssymb}
\usepackage{latexsym,mathtools,csquotes}
\usepackage{amstext}
\usepackage{enumerate}
\usepackage{cancel}

\newcommand\sing{\mathrm{sing}}

\newcommand\st{\mathrm{st}}

\renewcommand{\le}{\mathrm{l}}
\newcommand{\re}{\mathrm{r}}
\newcommand{\kb}{{\mathbf k}}

\newcommand\cH{\mathcal{H}}

\renewcommand{\bar}{\overline}
      \newcommand{\Bog}{{\operatorname{bg}}}
      
      \renewcommand{\i}{{\operatorname{i}}}
     \newcommand{\e}{\operatorname{e}}
     
     \newcommand{\hc}{{\operatorname{h.c.}}}

     \newcommand{\s}{{\operatorname{s}}}
     \renewcommand{\d}{{\operatorname{d}}}

     \newcommand{\R}{{\mathbb{R}}}

\renewcommand{\Re}{{\operatorname{Re}}}
\renewcommand{\Im}{{\operatorname{Im}}}
\newcommand{\zz}{{\mathbb{Z}}}

\newcommand{\rr}{{\mathbb{R}}}
\newcommand{\x}{\mathbf{x}}

\newcommand{\0}{{\bf 0}}
\newcommand{\kk}{{\bf k}}
\newcommand{\xx}{{\bf x}}
\newcommand{\q}{{\bf q}}
\newcommand{\pp}{{\bf p}}
\newcommand{\qq}{{\bf q}}
\newcommand{\p}{\mathbf{p}}

\newcommand\Tr{\operatorname{Tr}}

\newcommand{\BB}{{\operatorname{B}}}
\newcommand{\LL}{{\operatorname{L}}}

     \theoremstyle{plain}
     \newtheorem{thm}{Theorem}[section]
     \newtheorem{prop}[thm]{Proposition}
     \newtheorem{lemma}[thm]{Lemma}

     \theoremstyle{definition}

     \newtheorem{example}[thm]{Example}

     \newtheorem{remark}[thm]{Remark}

     \newcommand\beq{\begin{equation}}
       \newcommand\eeq{\end{equation}}
     \renewcommand\k{\mathbf{k}}
     \renewcommand\p{\mathbf{p}}
          \renewcommand\q{\mathbf{q}}
\newcommand{\ben}{\begin{arabicenumerate}}  
  \newcommand{\een}{\end{arabicenumerate}}
\newcounter{smallarabics}
\newenvironment{arabicenumerate}
{\begin{list}{{\normalfont\textrm{(\arabic{smallarabics})}}}
  {\usecounter{smallarabics}\setlength{\itemindent}{0cm}
   \setlength{\leftmargin}{5ex}\setlength{\labelwidth}{4ex}
   \setlength{\topsep}{0.75\parsep}\setlength{\partopsep}{0ex}
   \setlength{\itemsep}{0ex}}}
{\end{list}}

   \numberwithin{equation}{section}
    \numberwithin{figure}{section}

    \title[Zero]{
      Damping of phonons in Bose gas\\ at low temperatures}

\author{J. Derezi\'{n}ski}
\author{L. Pettinari}
\address[J. Derezi\'{n}ski]
{Dept. of Math. Methods in Phys., University of Warsaw\\ 
Pasteura 5, 02-093 Warszawa, Poland} 
\email{Jan.Derezinski@fuw.edu.pl}
\address[L. Pettinari]{Dipartimento di Matematica, Universit\`a di Trento and INFN-TIFPA and INdAM, Via Sommarive 14, I-38123 Povo, Italy}
\email{lorenzo.pettinari@unitn.it}

\date{\today}

\begin{document}

\begin{abstract} We consider  a  homogeneous 
  Bose gas in a large cubic box with periodic boundary conditions interacting with a small potential with a positive Fourier transform. We compute the imaginary part of the phononic excitation spectrum in the lowest order of perturbation theory in the thermodynamic limit at low temperatures and low momentum. Our analysis is based on perturbation theory of the standard Liouvillean. We use two approaches:  the first, motivated by the standard representation of operator algebras, examines resonances  near zero; the second analyzes the 2-point correlation function in the energy-momentum space.
\end{abstract}
\maketitle
\tableofcontents

\section{Introduction}

 We consider a homogeneous Bose gas with a  two-body potential $v(\xx)$
having a positive Fourier transform at positive density and a (low) temperature
$\frac1{k_\mathrm{B}\beta}$.
We assume that the density is positive, and is fixed by  a positive
parameter $\nu$, which can be interpreted as the chemical potential.
We assume at least partial
Bose-Einstein  condensation of the quantum gas, which is believed to be
appropriate for low temperatures.
We use the Hamiltonian with the  zeroth mode 
replaced by a $c$-number. We work  in a large periodic box of size $L\to\infty$.

We treat
the quadratic part of our Hamiltonian as a good approximation of the
system. This quadratic part
can be diagonalized with the help of a Bogoliubov
transformation \cite{B_47}. Therefore, we call it the \textit{Bogoliubov Hamiltonian}.
As is well-known, this yields the quasiparticle
dispersion relation
\begin{equation}
\omega_{\Bog}(\kk)=
\sqrt{\frac14 |\kk|^4+\frac{\nu\hat v(\kk)}{\hat v(\0)}|\kk|^2}. 
\label{omega1.}\end{equation}
This relation for low momenta behaves as
$\omega_{\Bog}(\kk)=\sqrt\nu|\k|$, and the corresponding
quasiparticles are called \textit{phonons}.

In the Bogoliubov approximation the interacting Bose  gas is approximately
described  as the   free Bose gas with a peculiar dispersion
relation. It is remarkable that this picture to a large degree is validated by 
experiments with ${}^4\mathrm{He}$\cite{GBSKDFO_2021} and alkali gases.
Neutron scattering (for larger momenta) and sound
propagation (for small momenta) in ${}^4\mathrm{He}$ allow us to measure its  quasiparticle
dispersion relation, which turns out to be a rather well-defined curve
similar to its theoretical prediction
\eqref{omega1.}. Experiments show that this relation is remarkably sharp, even though we
 expect it to be broadened
 by  interaction with higher order
terms of the Hamiltonians.

This broadening is the main topic of our paper. It
can be described by the
imaginary part of the dispersion 
relation computed in perturbation theory.
To define the perturbation expansion we replace the potential $v$ by
$\kappa v$, which yields $\sqrt\kappa$ in front of the three body
interaction and $\kappa$ in front of the four-body interaction.

For positive temperatures
it is convenient
to use the
Liouvillean instead of the Hamiltonian.
This involves doubling the original Hilbert space, which now possesses
two 
kinds  of quasiparticles: ``excitations over the thermal
equilibrium'',
and ``holes in thermal equilibrium'',
 which we will call the {\em left}, resp. {\em right quasiparticles}.

The main result of our paper is a computation of
the imaginary part of the dispersion relation
in the lowest nontrivial order in
$\kappa$. This amounts to the Fermi Golden Rule for the 
three-body term. We do this both at zero temperature, and at positive
temperature, keeping the  chemical potential $\nu$ and the temperature
$\frac1{k_\mathrm{B}\beta}$ fixed, and taking thermodynamic limit $L\to\infty$. 

We show that the imaginary part of the quasiparticle
dispersion relation in the order $\kappa^1$ is the
sum of two terms:
\beq  -\gamma_\BB(\k,\beta,\nu) - \gamma_\LL(\k,\beta,\nu).\label{gamma}\eeq

In the literature the term $\gamma_\BB$ goes under the name of the
{\em Beliaev damping}.
 The Beliaev damping is caused by a left quasiparticle decaying into 
two left quasiparticles. This process persists down to the zero temperature.
This decay rate was first computed by Beliaev
in a series of two remarkable papers \cite{B_58_1,B_58}. This result has been
re-derived following different theoretical approaches in   several other
works, e.g. \cite{DLN_2023,GN_64, G_98, L_97, MS_59}.
In particular,  the present paper can be viewed as a continuation of the
recently published  detailed derivation of the Beliaev damping at zero  
temperature by Dereziński, Li and Napiórkowski \cite{DLN_2023}.
The damping rate
observed in experiments appears to be consistent with its theoretical predictions \cite{DOSK_2002,FHHM_2001}.

The term $\gamma_\LL$ is usually called the {\em Landau damping}. It was first computed by Hohenberg and Martin in \cite{HM_65} and later by other authors \cite{G_98, L_97,PS_97}. Also in this case, experimental results \cite{CEJMW_97} seem to validate the theoretical predictions; see \cite{L_97} for a comparison between experimental and theoretical results. This process is absent at zero temperature, as it involves an interaction between left and right quasiparticles. Indeed, it is caused by a left quasiparticle
decaying into one left and one right  quasiparticle.

Before describing the formulas for
\eqref{gamma},  let us introduce the auxiliary
functions
\begin{align}\label{eq: bogo coefficients 1}
s_\kk&:=\operatorname{sign}(\hat{v}(\k))\frac{\sqrt{\sqrt{\omega_\Bog(\kk)^2+\nu^2\frac{\hat 
       v(\kk)^2}{\hat
       v(\0)^2}}-\omega_\Bog(\kk)}}{\sqrt{2\omega_\Bog(\kk)}},\\
  c_\kk& \label{eq: bogo coefficients 2}:=\frac{\sqrt{\sqrt{\omega_\Bog(\kk)^2+\nu^2\frac{\hat 
       v(\kk)^2}{\hat v(\0)^2}}+\omega_\Bog(\kk)}}{\sqrt{2\omega_\Bog(\kk)}}.
  \end{align}
They are the  coefficients of the Bogoliubov transformation diagonalizing
the quadratic part of the Hamiltonian: note that $c_\kk^2-s_\kk^2=1$.
For three momenta $\kk$, $\pp$, $\qq$, we set
\begin{align}\label{eq: jj}
		j(\k;\p,\q) := \sqrt{\frac{\nu}{\hat{v}(\0)}}&\Big{[}\hat{v}(\k)(s_\k-c_\k)(c_\p s_{\q} +c_{\q}s_{\p}) \notag\\ +& \hat{v}(\p)(c_\p- s_\p)(c_\k c_{\q} +s_{\k}s_{\q}) \notag \\ +&\hat{v}(\q)(c_{\q}-s_{\q})(c_\p c_{\k} +s_{\p}s_{\k})   \Big{]},
	\end{align} and 
	\begin{align}\label{eq: kappa}
		\kappa(\k,\p,\q) := \sqrt{\frac{\nu}{\hat{v}(\0)}}&\Big{[}\hat{v}(\k)(s_\k-c_\k)(c_\p s_{\q} +c_{\q}s_{\p}) \notag\\ +& \hat{v}(\p)(s_\p- c_\p)(c_\k s_{\q} +s_{\k}c_{\q}) \notag \\ +&\hat{v}(\q)(s_{\q}-c_{\q})(c_\p s_{\k} +s_{\p}c_{\k})   \Big{]},
	\end{align}
The functions $j(\k;\p,\q)=j(\k;\q,\p)$ and $\kappa(\k,\p,\q) = \kappa(\p, \k,\q) = \kappa(\p,\q,\k)$ describe the amplitude of a
three-quasiparticle vertex. We prove that
	\begin{align}\notag	\gamma_\BB=\frac{\pi}{2}\int& \frac{\d\p}{(2\pi)^3}
                                    j(\k;\p,\p-\k)^2
                                    \delta(\omega_\Bog(\k)-\omega_\Bog(\p)-\omega_\Bog(\k-\p))\\
          &\times\frac{(1-
            \e^{\frac{\beta}{2}(\omega_\Bog(\k)+\omega_\Bog(\p)+\omega_\Bog(\k-\p))})^2
            }{(\e^{\beta\omega_\Bog(\k)}-1)(\e^{\beta\omega_\Bog(\p)}-1)(\e^{\beta\omega_\Bog(\k-\p)}-1)}, \label{gamma_B}
          \\ \notag\gamma_\LL=\pi\int& \frac{\d\p}{(2\pi)^3} j(\p-\k;\k,\p)^2
                         \delta(\omega_\Bog(\k-\p)-\omega_\Bog(\p)-\omega_\Bog(\k))\\
          &\times\frac{( \e^{\frac{\beta}{2}(\omega_\Bog(\k)+\omega_\Bog(\k-\p))} - \e^{\frac{\beta}{2}\omega_\Bog(\p) })^2 }{(\e^{\beta\omega_\Bog(\k)}-1)(\e^{\beta\omega_\Bog(\p)}-1)(\e^{\beta\omega_\Bog(\k-\p)}-1)}. \label{gamma_L}
		\end{align} The function $\kappa(\k,\p,\q)$ will be relevant in the computation of the quasi-particle two-point function.

We also analyze the behavior of $\gamma_\BB$ for small momenta, and
of $\gamma_\LL$ for small momenta and temperatures.
To understand these results one should first remark that a natural
dimensionless measure of the momentum is $\frac{|\k|}{\sqrt\nu}$, and
of temperature is $\frac{1}{\beta\nu}$. Thus the ratio of momentum to
temperature is described by $\beta|\k|\sqrt\nu$.

In both theorems we make some additional assumptions on the potential, which
essentially say that the Fourier transform of $v$ is not too negative and is
sufficiently differentiable. Some additional hypotheses for the characterization of the shape of the Bogoliubov dispersion relation will be needed. They will be explained in detail in
Subsect. \ref{Assumptions}.

The following theorem describes the Beliaev damping rate for small
momenta.

\begin{thm}\label{thm: beliaev damping.}
  The following estimates hold: \ben
	\item For small momenta and   temperature/momentum ratios we have \begin{align}\label{eq: reduced expression for Beliaev.}
\gamma_\BB(\k,\beta,\nu)=  \frac{3\hat{v}(\0)\nu^{3/2}}{640\pi}\frac{|\k|^5}{\nu^{5/2}}\left(1 + O\Big( 
\frac{1}{(\beta\sqrt{\nu}|\k|)^3}\Big) +  O\Big(\frac{|\k|^2}{\nu}\Big)\right) \notag \\ \text{as }\frac{|\k|}{\sqrt{\nu}},\; \;\frac{1}{\beta\sqrt{\nu}|\k|}\to 0 .
	\end{align}
      \item For small momenta and  momentum/temperature ratios we have	\begin{align}\label{eq: Beliaev for high temperatures.}\notag
		\gamma_\BB(\k; \beta,\nu) = \frac{3\hat{v}(\0)\nu^{3/2}}{128\pi } \frac{|\k|^4}{\nu^2}\frac{1}{\beta\nu}\left( 1  + O\Big(\beta\sqrt{\nu}|\k|\Big) +O\Big(\frac{|\k|^2}{\nu}\Big)\right)\\ \text{as } \frac{|\k|}{\sqrt{\nu}}, \; \beta\sqrt{\nu}|\k|\to 0.
	\end{align} 
\een
\end{thm}

Next we describe the asymptotics of the Landau damping rate for small
momenta and temperatures.

\begin{thm}\label{thm: Landau damping.}
  The following estimates hold:
	\ben
	\item For small momenta, temperatures and the
          momentum/temperature ratios we have
         \begin{align}\label{eq: very high temperatures.}
		\gamma_\LL(\k;\beta,\nu) =  \frac{3\pi^3\hat{v}(\0)\nu^{3/2}}{40}\frac{|\k|}{\sqrt{\nu}}\frac{1}{(\beta\nu)^4}	\left(1 +  O\left( \beta\sqrt{\nu}|\k|\right)+  O\Big(\frac{1}{(\beta\nu)^2}\Big)\right)\notag \\  \text{as }\frac{|\k|}{\sqrt{\nu}},\;\frac{1}{\beta\nu},\;\beta\sqrt{\nu}|\k|\to 0 .
	\end{align}
      \item
For small momenta, temperatures and the
          temperature/momentum ratios we have
       	\begin{align}\label{eq: basse temperature.}
			\gamma_\LL(\k;\beta,\nu) = \frac{9\zeta(3)\hat{v}(\0)\nu^{3/2}}{16\pi}\frac{1}{(\beta\nu)^3}\frac{|\k|^2}{\nu}\left(1 + O\Big(\frac{1}{\beta\sqrt{\nu}|\k|}\Big) +O\Big(\frac{|\k|^2}{\nu}\Big)\right) \notag \\ \text{as }\frac{|\k|}{\sqrt{\nu}},\;\frac{1}{\beta\nu},\;\frac{1}{\beta\sqrt{\nu}|\k|}\to 0.
	\end{align}
	\een
\end{thm}

	For  $\beta\sqrt{\nu}|\k|\to 0$, the function $\gamma_\LL(\k;\beta,\nu)$ approaches the known low-temperature value of the Landau damping rate  as originally  derived by Hohenberg and Martin in \cite{HM_65}. To the best of our knowledge, the temperature dependent correction \eqref{eq: Beliaev for high temperatures.} for $\gamma_\BB(\k;\beta,\nu)$ as $\beta\sqrt{\nu}|\k|\to 0$ is first obtained here.  In this regime, the ratio between the Landau and the Beliaev damping is of order \begin{equation}
		\frac{\gamma_\BB(\k;\beta,\nu)}{\gamma_\LL(\k;\beta,\nu)} = O\left( (\beta\sqrt{\nu}|\k|)^3 \right) \quad \text{as } \beta\sqrt{\nu}|\k| \to 0,
	\end{equation} that is, the dominant effect is given by the Landau damping.
	
	 Conversely, in the limit $\frac{1}{ \beta\sqrt{\nu}|\k|}\to0$, $\gamma_\BB(\k;\beta,\nu)$ reproduces the standard expression originally computed in \cite{B_58}. In the same asymptotic regime Equation \eqref{eq: basse temperature.} for the Landau damping appears to be a novel contribution of the present paper. The ratio of the two rates is of order \begin{equation}
	 \frac{\gamma_\LL(\k;\beta,\nu)}{\gamma_\BB(\k;\beta,\nu)} = O\left( \frac{1}{(\beta\sqrt{\nu}|\k|)^3} \right) \quad \text{ as } \frac{1}{\beta\sqrt{\nu}|\k|}\to 0,
	 \end{equation} and the dominant contribution is the Beliaev damping one.
	
	It was further pointed out in \cite{L_97} that there exist
        experimental setups where the approximation
        $\frac{1}{\beta\nu} \to 0$ is not well justified. In that
        case, one should employ the full integral expressions
        \eqref{gamma_B} and \eqref{gamma_L}. We give a brief
        analysis of the \textit{high temperature regime} $\beta\nu\to 0$ in Proposistion \ref{prop: high temperature}, although its
        physical relevance is somewhat dubious.

So far we have tried to describe the physical meaning of our results
without discussing the formalism used to derive them. Actually, we
will use two distinct formalisms, both yielding
 the formulas
\eqref{gamma_B} and \eqref{gamma_L} for the damping rates.
Both of them, described 
in the first part of our paper, are based on the Hamiltonian with a $c$-number
condensate, and the corresponding Liouvillean.
Let us call them
\ben
\item the {\em standard representation approach},
\item the {\em Green function approach}.
  \een

The first formalism
is directly inspired by the modular theory of $W^*$-algebras
\cite{BR1,BR2,DJP_2003}. We
first consider the Bogoliubov Liouvillean, that is the Liouvillean
obtained from the purely quadratic part of the Hamiltonian. It is a quadratic
bosonic operator, and therefore its $W^*$-algebraic theory is well understood.
One can consider one-quasiparticle states and their corresponding
standard vector representatives, which belong to the so-called
standard positive cone.
They are embedded in 
continuous spectrum, and under the influence of the perturbation in
thermodynamic limit
they develop an imaginary part of their energy.
These vectors are obtained by acting on the KMS vector with one left
and one right quasiparticle creation operator.

Similar computations, pioneered by Jak\v{s}i\'c-Pillet, were
applied in the literature to spin-boson \cite{JP_96} and Pauli-Fierz-type
Hamiltonians \cite{DJ_03}, where they involved nontrivial (type III)
$W^*$-algebraic systems directly in thermodynamic limit. In these
applications an infinitely extended system interacted with a small
one, and there was no translation invariance. In our case we consider
a sequence of finite volume systems whose size goes to infinity. Our
$W^*$-algebras are trivial from the point of view of the abstract
theory (they are of type I). Nevertheless, in our opinion
the methods developed in the
 operator algebras approach yield  useful guiding principles for physical systems.

The second approach involves time-ordered two-point correlation functions of the
fields in the energy-momentum space. It was originally
inspired by the formalism of quantum electrodynamics, and was used
 by Beliaev in \cite{B_58_1,B_58}. The correlation functions often go
 under the name of  Green functions, since by the Lehmann
 representation they are expectation values of the resolvent of the Liouvillian. This approach belongs 
to standard tools of many body quantum physics \cite{FW_2003}.

The pure Bogoliubov Liouvillean
has a singular shell along the dispersion relation of
quasiparticles, both in the upper halfspace and symmetrically in the
lower halfspace. These singular shells are produced by acting with one
quasiparticle creation operator onto the Bogoliubov KMS vector. Under
the perturbation, the dispersion relation should develop an imaginary
part, so that the shells become broadened. This broadening is visible
in 2-point functions.

The two-point functions are convenient theoretically, however
they seem to be difficult to access  empirically. Typical
experiments measure the so-called van Hove form factors, expressible
in terms of density-density correlation functions. They are special
cases of  4-point correlation functions. One expects
that the energy-momentum picture of density-density correlation
functions is similar to that
of 2-point functions \cite{GN_64,G_2004}.

Anyway, both the standard representation approach and the Green
function approach are consistent with one another, at least to the lowest nontrivial
order in perturbation theory, and yield the same damping rate given in
\eqref{gamma}, \eqref{gamma_B}, and \eqref{gamma_L}. They both rely on
a \textit{weak-coupling approximation}, where the small parameter
governing the expansion of the energy corrections is the Fourier
transform of the interaction potential.
As we mentioned above, this is made explicit
through the
introduction of a coupling constant $\kappa$ in front of the
potential, while keeping the chemical potential $\nu$ fixed.

In most physics literature when considering phonon damping one uses
the so-called low-density approximation, where the potential is not
necessarily small, but very short range
\cite{B_58_1,B_58,GN_64,HG_98,MM_60, T_59}. The small parameter is
then the scattering length $\mathfrak{a}$. The low-density approximation is usually
considered to be more physically meaningful,  relevant especially for
Bose condensates of alkali elements. We prefer to use the weak
coupling limit, since mathematically it is simpler and cleaner.

One can also remark that the most famous Bose condensate, that is
${}^4\mathrm{He}$ below the critical temperature, is neither weakly interacting, nor dilute. Its
quasiparticle dispersion relation has a complicated  shape with {\em rotons,
maxons} and the {\em Pitaevskii plateau}, qualitatively similar to the picture
obtained in the weak coupling limit.

In fact, in the dilute
approximation, instead of \eqref{omega1.}, the dispersion relation reads
\begin{equation}
\omega_{\Bog}(\kk)=
\sqrt{\frac14 |\kk|^4+n_0 \mathfrak{a}|\kk|^2},
\label{omega2.}\end{equation}
where $\mathfrak{a}$ is the scattering length and $n_0$ is the density of the
condensate. Now, \eqref{omega2.} is strictly monotonic, without maxons
or rotons.

When investigating 
the quasiparticle
spectrum of Bose gas  from  ``first principles''
one assumes that the Hamiltonian is in a large box with periodic
boundary conditions, and then goes to thermodynamic limit. One needs
to arrange this Hamiltonian in such a way that one can
extract the quadratic part, which yields the Bogoliubov
approximation.
This can be done in at least two settings:
\begin{enumerate}
\item the canonical approach,
\item the grand-canonical approach,
  \end{enumerate}

In the canonical approach one can start from the $N$-body Hamiltonian with a fixed number of particles. Following the
idea of Arnowitt-Girardeau \cite{AG_59},
one can transform the $N$-body space into a Fock
space with a variable number of particles by treating the zero mode in
a special way. This transformation was used e.g. by
Derezi\'{n}ski-Napi\'{o}rkowski \cite{DN_2014}, and essentially also  by
Lewin-Nam-Serfaty-Solovej \cite{LNSS_2015}. One keeps the density of
particles $\frac{N}{L^d}$ fixed.
If  we retain in the Hamiltonian only the terms which have 
``nice thermodynamic limits'', then we obtain the Hamiltonian with a 
$c$-number substitution.

The second approach follows Hugenholtz-Pines \cite{HP_59}, and is
described e.g. in \cite{DLN_2023}. It starts from the
grand-canonical Hamiltonian, which allows for a variable number of
particles from the very beginning and is obtained by
 subtracting  $\nu N$ from the
Hamiltonian. The chemical potential $\nu$ is kept fixed in
thermodynamic limit. If at the
end we compress our Fock space away from the zeroth mode, we also
obtain the Hamiltonian with a 
$c$-number substitution, the same as before.

We should mention yet  another approach, which is implemented for example in \cite{LY_58, MM_60,MS_59}, where one makes use of a non-self-adjoint Hamiltonian. We will not explore this direction in the present paper.

One of the important elements of our analysis is keeping fixed the
parameter  $\nu$. In a more precise analysis we actually encounter several possibilities
for the interpretation of the parameter $\nu$.
\begin{enumerate} \item In the grand-canonical approach $\nu$ is the true chemical potential (then we
  are forced to use the  formalism with a variable
  number of particles).
\item
  $\nu$  can be the density times $\kappa\int v(\x)\d x$. This is natural if we use the canonical formalism,
  with a fixed number of particles.
  \item  $\nu$  can be the density of the condensate (the
    zero mode) times $\kappa\int v(\x)\d x$.
\end{enumerate}
Anyway, these three interpretations differ (at least formally) by a
quantity of the order $\kappa$, and therefore they do not affect the
imaginary part of the dispersion relation in the lowest order.

As discussed above, within the weak-coupling approximation, our
results reproduce the known decay rates \eqref{eq: reduced expression
  for Beliaev.} and \eqref{eq: very high temperatures.} in the
corresponding temperature and momentum regimes. In principle, the
Green-function approach should determine both the imaginary and the
real parts of the quasiparticle dispersion relation. The
approximations employed in the present work, however, are not
sufficiently accurate to compute the real part correctly. Even worse, the resulting correction to the energy of a quasiparticle with momentum $\k$ exhibits an unphysical infrared divergence of order $\frac{1}{|\k|}$. We now explain the origin of this failure.
Throughout this paper, we assume that the Bose gas is described by a Hamiltonian with a $c$-number condensate. Formally, this Hamiltonian is obtained from the usual first-principles Hamiltonian of the Bose gas by replacing the creation and annihilation operators of the zero mode by $\sqrt{\nu}\in\rr$ and dropping an additive constant. The resulting Hamiltonian takes the form
\beq
H^L=H_{\Bog}^L+\sqrt{\kappa}H_3^L+\kappa H_4^L,
\eeq
and the corresponding Liouvillean has the analogous decomposition
\beq
L^L=L_{\Bog}^L+\sqrt{\kappa}L_3^L+\kappa L_4^L.
\eeq
Here, the superscript $L$ denotes the size of the box, and the limit $L\to\infty$ is taken in our computations.
In perturbation theory for the quasiparticle spectrum, the contribution of order $\sqrt{\kappa}$ vanishes. The leading correction is therefore of order $\kappa$ and has the form
\begin{align}\label{priu}
	\left(\Psi_0\right.\left|L_4^L\Psi_0\right)
	+
	\left(L_3^L\Psi_0\right.\left|
	(L_{\Bog}^L-\i\epsilon)^{-1}L_3^L\Psi_0\right),
\end{align}
where $\Psi_0$ denotes an appropriate unperturbed vector. Thus, \eqref{priu} is the sum of a Hellmann-Feynman term involving $L_4^L$ and a Fermi Golden Rule term involving $L_3^L$.
The imaginary part of the correction comes entirely from the Fermi
Golden Rule term and therefore depends only on $L_3^L$. As shown in
\cite{DP2026}, the $c$-number substitution produces the correct cubic
contribution $L_3^L$. This explains why the present approach correctly
captures the damping rates. By contrast, the real part of the
dispersion relation receives contributions from both terms in
\eqref{priu}. The quartic term $L_4^L$ obtained through the naive
$c$-number substitution misses non-negligible contributions generated
by the commutation relations of the zero-mode creation and
annihilation operators. This issue was already observed by Hugenholtz
and Pines \cite{HP_59} in the grand-canonical formalism.

At zero temperature in the dilute approximation the real 
shift of the dispersion relation was computed already by Beliaev \cite{B_58}.
In the weak coupling regime this real shift
  is the subject of our second paper on this topic \cite{DP2026}.
We discuss both the canonical and grand-canonical
settings at zero temperature.
Careful analysis
leads to additional terms in the Hamiltonian, absent the naive
Hamiltonian with a $c$-number substitution. Some of
them are of order $\kappa^1$, others are of higher order and/or
(formally) vanish in thermodynamic limit. We
discuss them in the separate paper \cite{DP2026}.  They are necessary
when we compute the real part of the corrections to the quasiparticle spectrum.
The real  shift requires an analysis much finer than
 the naive $c$-number substitution used in this paper.

Let us conclude with some additional remarks on the literature. In recent years, there have been numerous rigorous justifications of Bogoliubov’s \textit{c}–number substitution in various physical regimes. In the mean-field setting, this was first established by Seiringer \cite{S_2011}; see also \cite{LNSS_2015,DN_2014,GS2013,NS_2015} for related developments. Corresponding results in the Gross–Pitaevskii regime were subsequently obtained in \cite{BBCS_2019,BSS2022,NT_2023}, and extensions beyond this regime appear in \cite{BCS2022}.

Operator-algebraic methods have also been employed to construct $C^*$-dynamical systems describing interacting Bose gases \cite{B_2020,BY_2025}. A different strategy, based on techniques from quantum field theory, was developed by Galanda and Pinamonti in \cite{GP2025,GP2026}, where equilibrium states for an interacting Bose gas are constructed via a suitable resummation of perturbative expansions.

Further advances concerning the bottom of the excitation spectrum beyond Bogoliubov theory have recently been achieved in the Gross–Pitaevskii regime \cite{COSS_2025}. In the mean-field regime, analogous analyses were carried out in \cite{BPS2021,NN_2017}. Finally, the energy–momentum spectrum of Bose gases was first rigorously investigated in \cite{DN_2014,DLN_2023}.




Apart from the introduction, the paper consists of two main sections. In 
Section \ref{Quasiparticle spectrum of Bose gas}
we discuss the general theory of quasiparticle spectrum at zero and
low positive temperatures. In particular, we introduce the two
approaches to define the phonon decay rate.
In Section \ref{sec: decayrates}, we will formulate precisely and prove
Theorems 
\ref{thm: beliaev damping.} and \ref{thm: Landau damping.}
stated in the introduction about
the asymptotics of the phonon damping for low
momenta.
We will consider the (physical)
dimension 3 and we will impose additional
technical assumptions on the potential.


\section{Quasiparticle spectrum of Bose gas}
\label{Quasiparticle spectrum of Bose gas}
\subsection{The ``first-principles Hamiltonian''}
\label{The ``first principles Hamiltonian''}

Consider a real function
$\mathbb{R}^{d}\ni \x\mapsto v
(\x)$, describing the two-body potential of a quantum gas. Its Fourier
transform will be defined  by \begin{equation}\label{eq: Fourier of potential}
\hat{v}(\p):=\int\d \x\; v(\x)\e^{-\i\p\x} .
\end{equation}

 A large part  of our preliminary analysis will be general and will not depend on the  
dimension nor on the regularity assumptions on the potential. Here we just mention the most basic
assumptions, which we will need to describe 
the general theory.

We assume that $\nu\in ]0,\nu_0]$, (for some $\nu_0>0$) which has the interpretation of the effective
chemical potential. Moreover, we will always assume
\begin{subequations}\label{eq: assumptions on potential1}
	\begin{align}
		&\label{eq: continuity of Fourier} v\in L^1(\mathbb{R}^d,\d \x),\; v(\x)\in\mathbb{R};\\
		&\label{eq: cutoff} \hat{v}\in L^2(\mathbb{R}^d,\d \k);\\
		&\label{eq: positivity of Fourier} \hat{v}(\0)>0,\; \hat{v}(\k) > -\hat{v}(\0)\frac{|\k|^2}{4\nu}; \qquad \nu\in ]0, \nu_0];\\
		&v(\x) = v(-\x),\quad \text{which implies $\hat{v}(\k) = \hat{v}(-\k)$}.\label{eq: symmetry requirement}
	\end{align}
\end{subequations} The first hypotheses, \eqref{eq: continuity of
  Fourier} guarantees that the Fourier transform of $v$, as in
\eqref{eq: Fourier of potential}, is a continuous function and the
Hamiltonian is self-adjoint. The second one \eqref{eq: cutoff}
provides a convenient natural cut-off on the high momenta region. The
third one, \eqref{eq: positivity of Fourier}, is needed to make sense
of the Bogoliubov transformation, \textit{cf}. Equation \eqref{eq: bogo
  roto}, where the linear combination of $b_\k,\; b_\k^*$ can only be
defined for $\frac{1}{2}|\k|^2+
\nu\frac{\hat{v}(\k)}{\hat{v}(\0)}>0$. The last condition, \eqref{eq:
  symmetry requirement}, can be always imposed if we deal with
identical particles (Bose or Fermi).

In our detailed
computations,
described in Section \ref{sec: decayrates}, we fix  $d=3$, and  we will
assume some additional restrictions on the
class of interaction potential $v$ we consider.
 They will be described in Subsect. \ref{Assumptions}.

We consider Bose gas with potential $v$ in large but finite  volume with periodic
boundary conditions.
Following the standard approach, we replace
the infinite space $\rr^d$ by the torus  $\Lambda_L=]-L/2,L/2]^{d}$.
 Let $\Xi_L:= \frac{2\pi }{L}\zz^d$ be the momentum lattice
    corresponding to $\Lambda_L$.  In the momentum representation
    and  the 2nd quantized formalism, the
    Hamiltonian, number operator and  total momentum are
    operators on the Fock space $\Gamma_\s\big(l^2(\Xi_L)\big)$ given by
\begin{eqnarray}\label{2B}
H^L 
&=&\sum_{\p}\frac{1}{2}\p^2a^*_{\p}
a_{\p}\\
&&+\frac{1}{2L^d}\sum_{\p,\q,\k}{\hat v}(\k)
a^*_{\p+\k}a^*_{\q-\k}a_{\p}a_{\q},\nonumber 
\\
N^L&=&\sum_{\p}a^*_{\p}a_\p,\label{2BN}\\
\mathbf{P}^L&=&\sum_{\p}\p a^*_{\p}a_\p,\label{2BP}
\end{eqnarray}
where all momenta are summed over $\Xi_L$. We will sometimes call \eqref{2B} the ``first-principles Hamiltonian
in finite volume''.


\subsection{Hamiltonian  with a $c$-number condensate}
\label{Hamiltonian  with a $c$-number condensate}



In this and in the following sections we will always implicitly assume the validity of \eqref{eq: continuity of Fourier},\eqref{eq: cutoff},\eqref{eq: positivity of Fourier} and \eqref{eq: symmetry requirement}.
Let us introduce a parameter
 $\nu>0$.
Set $\Xi_L^>:=\Xi_L\setminus\{\0\}$.
Our basic Hamiltonian will not be \eqref{2B}, but the \textit{Hamiltonian of
the Bose gas  with a $c$-number
condensate}. It acts on $\Gamma_\s\big(l^2(\Xi_L^>)\big)$ and is
defined by
\begin{align}\label{decom0}
    H^L&=
            H_{\Bog}^L+H^{L}_{3}+H^{L}_4,\\
  H_{\Bog}^L&:=\sum_{\p\in\Xi_L^>}\Big(\frac{\p^2}{2}+\frac{\nu\hat v(\p)}{\hat v(\0)}\Big) a_\p^* 
 a_\p
+
\sum_{\p\in\Xi_L^>}\Big(\frac{\nu\hat v(\p)}{2\hat v(\0)}
  a_\p 
 a_{-\p}
          +\hc\Big),\\\label{put1}
H^{L}_{3}&:=\frac{1}{L^{d/2}}\sum_{\k,\p,\p+\k\in\Xi_L^>}\frac{\sqrt{\nu}\hat 
                   v(\k)}{\sqrt{\hat v(\0)}}\Big( a_{\p+\k}^*  a_\k  a_{\p}
+\hc\Big),\\\label{put2}
H^{L}_4&:=\frac{1}{2L^d}
\sum_{\substack{\p,\q,\\ \p+\k,\q-\k\in\Xi_L^>}}\hat v(\k) a_{\p+\k}^*  a_{\q-\k}^* 
 a_\q  a_\p.  
\end{align}
The Hamiltonian $H^L$ can be obtained from $H^L$ by replacing
 $\frac{\sqrt{\hat v(\0)}}{L^{\frac{d}{2}}}a_0$, $\frac{\sqrt{\hat v(\0)}}{L^{\frac{d}{2}}}a_0^*$
  with  $\sqrt\nu$,  and dropping a constant.

We will usually replace $v$ with $\kappa v$, assuming that $\kappa$
is small.
Thus parameter $\nu$ can be interpreted as
$\kappa\hat v(\0)\frac{N_\0}{L^d}$. It will be kept fixed and independent of $\kappa$, so that
$H_\Bog^L$ does not depend on $\kappa$.
Note that $H_{3}^L$,
resp. $H_4^L$
consists of terms of the
order $\sqrt\kappa$, resp. $\kappa$.
Thus \eqref{decom0} should be replaced by
 \beq    H^L=
            H_{\Bog}^L+\sqrt\kappa H^{L}_{3}+\kappa H^{L}_4.\eeq

\subsection{Thermodynamic  limit}

We would like to study Bose gas
 in thermodynamic limit $L\to\infty$,  keeping $\nu>0$ and  
  $\rr^d\ni\kk\mapsto \hat v(\kk)$ fixed. 
  $\Xi_L^>$ formally converges to $\rr^d$. One should
  replace $\frac{1}{L^d }\sum\limits_{\kb\in\Xi_L^>}$ with $\frac{1}{(2\pi)^d }\int\d
\kb$, and the operators $a_\kb$, $a_\kb^*$ for  $\kb\in\Xi_L^>$ with
$\frac{(2\pi)^{\frac d 2}}{L^{\frac d 2}} a_\kb$,
$\frac{(2\pi)^{\frac d 2}}{L^{\frac d 2}} a_\kb^*$ for
 $\kb\in \rr^d$.

 Thus 
$H_{\Bog}^L$,  $H^{L}_{3}$, $H^{L}_4$ and $\mathbf{P}^L$ formally
converge to operators on $\Gamma_\s\big(L^2(\rr^d)\big)$
\begin{align}\label{therm1}
 H_\Bog&:=\int\d \p\;\Big(\frac{\p^2}{2}+\frac{\nu\hat v(\p)}{\hat v(\0)}\Big) a_\p^*
 a_\p
+
\int\d \p\;\Big(\frac{\nu\hat v(\p)}{2\hat v(\0)}
  a_\p 
 a_{-\p}
+\hc\Big),
\\\label{therm2}
  H_{3}&:=\frac{1}{(2\pi)^{\frac{ d}{2}}}\int\d \p\d\k\;\frac{\sqrt{\nu}\hat 
                   v(\k)}{\sqrt{\hat v(\0)}}\Big( a_{\p+\k}^*  a_\k  a_{\p}
+\hc\Big),\\
H_4&:=\frac{1}{2(2\pi)^d}\int\d\p\d\k\d\q\;\hat v(\k) a_{\p+\k}^*  a_{\q-\k}^* 
          a_\q  a_\p,  \\
  \mathbf{P}&:=\int\d \p\;\mathbf{p} a_\p^*
 a_\p.
\end{align} 
Now
\beq H:=H_\Bog+\sqrt\kappa H_{3}+\kappa H_4 \label{B2a}\eeq
is the Hamiltonian of the Bose gas in infinite volume with a
$c$-number condensate given by $\nu$.

Note that we can also perform thermodynamic limit in the ``first principles
Hamiltonian'' $H^L$ of \eqref{2B}, and then we obtain $H$ of
\eqref{B2a} with $\nu=0$.

In the literature the Hamiltonians $H$ and $H^L$  often appear
in descriptions of
Bose gas at positive density for large $L$.
One should admit that the whole procedure  of the $c$-number substitution is
problematic. In particular, we do not know
whether $H$ in thermodynamic limit is well-defined as a self-adjoint 
operator. In any case, 
 we prefer to ask questions about the limits as
$L\to\infty$ of finite volume
quantities involving $H^L$.

\subsection{Bogoliubov method}
\label{s3}

The Hamiltonian $H_\Bog^L$, and also $H_\Bog$,
will be called the {\em Bogoliubov Hamiltonian} and it will be treated
as the main part of the full Hamiltonian. It  is purely quadratic,
therefore it can be diagonalized using an appropriate {\em Bogoliubov
  transformation}.

In fact, consider  $\theta=(\theta_\kk)_{\kk\in \Xi_L^>}$, a sequence of real numbers. Set
\begin{equation}
c_\kk:=\cosh\theta_\kk,\ \  s_\kk:=-\sinh\theta_\kk.\label{hyper}\end{equation}
For $\k\in\Xi_L^>$ we 
 make the substitution 
  \begin{eqnarray}
& a_\kk^*=c_\kk b_\kk^*- s_\kk b_{-\kk},&
 a_\kk=c_\kk b_\kk-s_\kk b_{-\kk}^*,\label{rota}\end{eqnarray}

 Note  that we have
$s_\kk=s_{-\kk}$ and $c_\kk =c_{-\kk}=\sqrt{1+ s_\kk^2}$.

We choose the Bogoliubov rotation that  kills double creators and
annihilators, which amounts to choosing $s_\k$ and $c_\k$ as 
\begin{align}\label{eq: bogo roto}
	c_\kk&=\frac{1}{\sqrt2}\left(\left(1-\left(\frac{\hat
		v(\kk)\frac{\nu}{\hat
			v(\0)}}{\frac12\kk^2+\hat
		v(\kk)\frac{\nu}{\hat
			v(\0)}}\right)^2\right)^{-1/2}+1\right)^{1/2}. \\ 
		s_\kk&= \frac{\nu\hat{v}(\k)}{2\hat{v}(\0)\omega_\Bog(\k)c_\kk}\label{sk4}
		\end{align} 
Note that \eqref{eq: bogo roto} and \eqref{sk4} are well defined on account of Assumption \eqref{eq: positivity of Fourier}.
We obtain
\begin{align}
H_{\Bog}^{L}&=E_{\Bog}^{L}
              +\sum_{\kk\in\Xi_L^>} \omega_{\Bog}(\kk) b_\kk^*b_\kk,\label{elem}\\
  \mathbf{P}^{L}&=\sum_{\kk\in\Xi_L^>} \kk b_\kk^*b_\kk,
\end{align}
where
 the quasiparticle excitation spectrum is defined in Equation \eqref{omega1.}
and the ground state energy is
\begin{eqnarray}
E_{\Bog}^{L}&=
&  -\sum_{\k\in\Xi_L^>}\frac12
\left(
\frac12\kk^2+\frac{\nu\hat v(\kk)}{\hat v(\0)}-\omega_{\Bog}(\kk)\right).\label{elem1}\end{eqnarray}

Introduce the unitary operator
\begin{equation}
U_{\theta}:=\prod_{\kk\in\Xi_L^>}\e^{-\frac12
\theta_\kk a_\kk^* a_{-\kk}^*+\frac12
    \theta_\kk a_\kk  a_{-\kk}}.\label{uthe}\end{equation}
Note that
\[U_\theta^*
 a_\kk
U_\theta=b_\kk,\ \
\ \ 
U_\theta^*
 a_\kk^*U_\theta
=b_\kk^*.\]
We can rewrite \eqref{elem}
as
\begin{eqnarray}
H_{\Bog}^{L}&=&E_{\Bog}^{L}
+U_\theta^*\sum_{\k\in\Xi_L^>} \omega_{\Bog}(\kk) a_\kk^*a_\kk U_\theta,\label{elem1}
\end{eqnarray}
and the ground state of the Bogoliubov Hamiltonian is
\begin{align}\label{ground}
  \Omega_\Bog:&=U_\theta^*\Omega^>
  = \prod_{\p\in\Xi^>_L}\frac{1}{\cosh^{\frac{1}{2}}\theta_\p}\e^{\frac{1}{2}\tanh\theta_\p a_\p^* a_{-\p}^*} \Omega^>,
\end{align}
where $\Omega^>$ is the vacuum in $\Gamma_\s(l^2(\Xi_L^>))$.

Note that $\omega_{\Bog}(\kk)$ of \eqref{omega1.} is 
well defined for all values $\kk\in\R^d$, even though it is restricted 
to $\kk\in\Xi_L^>$ in (\ref{elem}) and
(\ref{elem1}). Therefore, in thermodynamic limit the Bogoliubov Hamiltonian can be
rewritten as
\begin{align}
H_{\Bog}&
            +\int\d\k\; \omega_{\Bog}(\kk) b_\kk^*b_\kk,\label{elem0}\\
              \mathbf{P}&=\int\d \kk\; \kk b_\kk^*b_\kk,
\end{align}
after subtracting from $H_\Bog$ the infinite ground state energy, whose density is given by
\begin{eqnarray}
\lim_{L\to+\infty }\frac{E_{\Bog}^L}{L^d}&=
&  -\frac1{2(2\pi)^d}\int\d\k\;
\left(
\frac12\kk^2+\frac{\nu\hat v(\kk)}{\hat
                 v(\0)}-\omega_{\Bog}(\kk)\right).\label{elem1a}\end{eqnarray}

           The joint energy-momentum spectrum 
           covers a part of the upper halfspace of $\rr^{1+d}$. If the
           volume is finite, it is discrete and
           \beq\label{sigma}\sigma \big(H_\Bog^L-E_\Bog^L,\mathbf{P}^L\big) 
\subset [0,\infty[\times\frac{2\pi}{L}\zz^d 
           \eeq
           In thermodynamic limit
           we have
                      \beq\label{sigma}\sigma \big(H_\Bog-E_\Bog,\mathbf{P}\big) 
\subset [0,\infty[\times\rr^d ,
           \eeq
           and it
is absolutely continuous with respect to the Lebesgue
measure on $\rr^{1+d}$ except for the ground
state at $(0,\mathbf{0})$ and the {\em elementary
  quasiparticle spectrum} at
\beq \sigma_\sing(H_\Bog-E_\Bog,\mathbf{P})=\big\{\big(\omega_\Bog(\mathbf{k}),\mathbf{k}\big)\ |\
\mathbf{k}\in\rr^d\big\}.\label{destroy1a}\eeq
\eqref{sigma} is the subadditive hull of \eqref{destroy1a}. For more details concerning the construction of the excitation spectrum we refer to \cite{CDZ_2009,DLN_2023,DN_2014}.

The same Bogoliubov transformation \eqref{rota} can be applied also to $H^L_3$ and $H^L_4$, thus obtaining an interaction expressed in terms of $b_\k^*,\; b_\k$.

\subsection{Positive temperatures}\label{Positive temperatures}
Let us recall  basic facts about quantum physics at positive
temperatures.  In the back of our heads we have the algebraic
description of infinitely extended 
systems involving $C^*$- and $W^*$-algebras \cite{BR1,BR2,DJP_2003}, and we will sometimes use terminology developed in these
frameworks. Unfortunately, the $C^*$-algebraic framework is problematic for
interacting bosons. Maybe, one could use the $C^*$-algebras introduced by
Buchholz \cite{B_2020,BY_2025}, but we have not tried them.

Suppose a quantum system is described by a Hamiltonian $H$ and a Hilbert
space $\cH$.
  Thus the evolution of  $A\in B(\cH)$ is given by
  \beq A(t):=\e^{\i tH}A\e^{-\i tH}.\eeq

  It is convenient to pass to the so-called {\em standard  representation} of 
  $B(\cH)$.
Its Hilbert
  space can be identified with
  Hilbert-Schmidt operators,
  denoted
  $B^2(\cH)$. The scalar product on   $B^2(\cH)$ is given by the
  trace:
  \beq(\Phi|\Psi):=\Tr\Phi^*\Psi,\quad \Phi,\Psi\in B^2(\cH).\eeq
    The algebra $B(\cH)$ has two commuting representations, the
    (linear) left
    one and the (antilinear) right one:
  \begin{align}\pi_\le(A)\Phi&:=A\Phi,\\
    \pi_\re(A)\Phi&:=\Phi A^*,\quad \Phi\in B^2(\cH), \quad A\in B(\cH).
  \end{align}

  Every normal state $\phi$ on $B(\cH)$ can be represented by a  {\em vector 
  representative}, that is a vector 
  $  \Phi\in B^2(\cH)$ such that 
  \beq \phi(A)=(\Phi|\pi_\le(A)\Phi).\eeq 
It has a unique {\em standard vector representative} which is given by a 
positive Hilbert-Schmidt matrix. Thus for the state $\phi$ with the density matrix 
$\rho_\phi$
 the standard vector representative  is $\Phi:=\sqrt{\rho_\phi}$.

  At inverse temperature $\beta$  the KMS state is 
  given by the Gibbs density matrix
  \beq\label{gibbs}
  \omega_\beta(A)=\Tr A \frac{\e^{-\beta H}}{\Tr\e^{-\beta H}},\quad
  A\in B(\cH).\eeq
 For positive temperatures,  the thermal state is faithful, therefore
 its GNS representation coincides with the
 standard representation. For zero temperature the GNS representation
 is just the usual irreducible representation on $\cH$. However, in
 order to unify the descriptions, one can use the standard
 representation for both positive and zero temperatures.

 The standard representative of the state $\omega_\beta$ is 
  \beq\Omega_\beta:=\frac {\e^{-\frac\beta2 
    H}}{\big(\Tr\e^{-\beta H}\big)^{\frac12}}.\eeq 
  Thus 
  \begin{align} 
    \omega_\beta(A)&=(\Omega_\beta|\pi_\le(A)\Omega_\beta).\end{align}

  The dynamics is generated by the Liouvillean
  \beq L:=[H,\cdot], \eeq
  so that
\beq  \pi_\le\big(A(t)\big)=\e^{\i tL}\pi_\le(A)\e^{-\i 
      tL},\quad L\Omega_\beta=0.\eeq

In general, a self-adjoint operator $K$ on $\cH$, understood as the generator of a
1-parameter unitary group, in the standard representation
should be replaced by its {\em standard} version \beq
K^\st :=[K,\cdot].\eeq

\subsection{Araki-Woods representation}

Let us go back to the Bose gas.
Its standard representation 
  acts on Hilbert-Schmidt operators on the Fock
space, denoted
\beq B^2\big(\Gamma_\s(l^2(\Xi_L^>)\big),\label{hilsch}\eeq
the Hamiltonian should be replaced by the Liouvillean, and the
momentum  by its standard version:
\begin{align}
  L^L&:=[H^L,\cdot],\\
  \mathbf{P}^{\st,L}&:=[\mathbf{P}^L,\cdot].
\end{align}
We have the Gibbs state of the full Hamiltonian $\omega_\beta$. We
will also use the Gibbs state for the Bogoliubov Hamiltonian
$\omega_{\Bog,\beta}$ with the standard vector representative $\Omega_{\Bog,\beta}$.

We can write
\begin{align}
  L^L&:=L_\Bog^L+\sqrt{\kappa}L_3^{L}+\kappa L_4^{L},\\
  L_\Bog^L&:=[H_\Bog^L,\cdot],\\
  L_3^{L}&:=[H_3^{L},\cdot],\\
  L_4^{L}&:=[H_4^{L},\cdot].
\end{align}

The representation on \eqref{hilsch}
is not very convenient for calculations. We prefer to pass to the 
Araki-Woods representation \cite{AW_63, DG_2013} well adapted to the Bogoliubov Liouvillean $L_\Bog^L$. 

This passage can be done in two steps. First, it is easy to see that
\eqref{hilsch} is naturally isomorphic to the Fock space on the
doubled 1-particle space
\beq
\Gamma_\s\big(l^2(\Sigma_L^>)\oplus l^2(\Sigma_L^>)
\big).\eeq
Then
\begin{align}
  \pi_\le\big(b(\kk)\big)=b_\le(\kk),\quad 
  \pi_\le\big(b^*(\kk)\big)=b_\le^*(\kk),\\
  \pi_\re\big(b(\kk)\big)=b_\re(\kk),\quad 
  \pi_\re\big(b^*(\kk)\big)=b_\re^*(\kk),
  \end{align}
where we  use $b_\le(\kk)$, $b_\le^*(\kk)$, resp.  $b_\re(\kk)$, 
$b_\re^*(\kk)$ for the annihilation, creation operators acting on 
the left, resp. right space $ l^2(\Xi_L^>)$. 
This is still not very convenient, because  $\Omega_{\Bog,\beta}$
corresponds to a relatively complicated squeezed vector.

Then one  applies an appropriate Bogoliubov transformation, obtaining
the so-called Araki-Woods representation.
This is still a representation on
\beq 
\Gamma_\s\Big(l^2(\Xi_L^>)\oplus l^2(\Xi_L^>)\Big).\eeq
However, in this representation the
vector $\Omega_{\Bog,\beta}$ is mapped onto the Fock vacuum $\Omega$, and the
creation/annihilation operators are represented by
$\pi_{\beta,\le}$ and $\pi_{\beta,\re}$ as follows:
\begin{align}
b_{\beta,\le}^*(\kk):=\pi_{\beta,\le}(b^*(\kk))&=\big(1-\e^{-\beta\omega_\Bog(\kk)}\big)^{-\frac12}b_\le^*(\kk) 
                              +\big(\e^{\beta\omega_\Bog(\kk)}-1\big)^{-\frac12}b_\re(\kk),\\
b_{\beta,\le}(\kk):=\pi_{\beta,\le}(b(\kk))  &=\big(1-\e^{-\beta\omega_\Bog(\kk)}\big)^{-\frac12}b_\le(\kk) 
                        +\big(\e^{\beta\omega_\Bog(\kk)}-1\big)^{-\frac12}b_\re^*(\kk),\\
  b_{\beta,\re}^*(\kk):=\pi_{\beta,\re}(b^*(\kk)) &=\big(\e^{\beta\omega_\Bog(\kk)}-1\big)^{-\frac12}b_\le(\kk) 
                              +\big(1-\e^{-\beta\omega_\Bog(\kk)}\big)^{-\frac12}b_\re^*(\kk),\\
 b_{\beta,\re}(\kk):=  \pi_{\beta,\re}(b(\kk))&= \big(\e^{\beta\omega_\Bog(\kk)}-1\big)^{-\frac12}b_\le^*(\kk) 
                              +\big(1-\e^{-\beta\omega_\Bog(\kk)}\big)^{-\frac12}b_\re(\kk).
\end{align}

For the following, we define \begin{equation}
	\rho(\kk)\colon = (\e^{\beta\omega_{\text{bg} }(\kk) } -1 )^{-1},\quad 1+ \rho(\kk) = (1- \e^{-\beta\omega_{\text{bg}}(\kk) })^{-1}.
\end{equation}

The free Liouvillean and the standard momentum
(written in two equivalent ways) become
\begin{align}
L_\Bog^L&=
\sum_{\k\in\Xi^>_L}\omega_\Bog(\kk)\Big( b_{\le}^*(\kk) 
b_{\le}(\kk)-
        b_{\re}^*(\kk) b_{\re}(\kk)\Big)\\&=
  \sum_{\k\in\Xi^>_L}\omega_\Bog(\kk)\Big( b_{\beta,\le}^*(\kk) 
b_{\beta,\le}(\kk)-
  b_{\beta,\re}^*(\kk) b_{\beta,\re}(\kk)\Big),\\
  \mathbf{P}^{\st,L}&=
\sum_{\k\in\Xi^>_L}\kk\Big( b_{\le}^*(\kk) 
b_{\le}(\kk)-
        b_{\re}^*(\kk) b_{\re}(\kk)\Big)\\&=
  \sum_{\k\in\Xi^>_L}\kk\Big( b_{\beta,\le}^*(\kk) 
b_{\beta,\le}(\kk)-
  b_{\beta,\re}^*(\kk) b_{\beta,\re}(\kk)\Big).
\end{align}

In thermodynamic limit
the joint spectrum of the Bogoliubov Liouvillean and the standard
momentum
is the whole space:
\beq\label{sigma1}\sigma \big(L_\Bog,\mathbf{P}^\st\big)=\rr^{1+d}.\eeq
It has a bound state at
$(0,\mathbf{0})$ corresponding to the KMS state. The singular
spectrum has now two branches: the
positive for ``left quasiparticles'' and negative for ``right quasiparticles'':
\beq
\sigma_\sing(L_\Bog,\mathbf{P}^\st)=
\big\{\big(\omega_\Bog(\mathbf{k}),\mathbf{k}\big)\ |\
\mathbf{k}\in\rr^d\big\}
\cup
\big \{\big(-\omega_\Bog(\mathbf{k}),\mathbf{k}\big)\ |\
\mathbf{k}\in\rr^d\big\}
.\label{destroy2}\eeq


Recall that the full Hamiltonian was given in \eqref{decom0}--\eqref{put2}. It is
useful to express it in terms of
 $b(\kk)$, $b^*(\kk)$.

 The full Liouvillean is
  \begin{align} L^L:=&L_\Bog^L
                       + \sqrt{\kappa}L_3^{L}+\kappa L_4^{L}\label{pert},\\
    L_3^{L}:=& H_{3,\le}^{L}-H_{3,\re}^{L},\\
    L_4^{L}:=& H_{4,\le}^{L}-H_{4,\re}^{L}
                     ,\end{align}
 where
$ H_{3,\le}^L,H_{4,\le}^L$, 
resp. $H_{3,\re}^L,H_{4,\re}^L$ are obtained by insertion of 
 $b_{\beta,\le}(\kk)$, $b_{\beta,\le}^*(\kk)$, resp.  $b_{\beta,\re}(\kk)$,
 $b_{\beta,\re}^*(\kk)$ instead of
 $b(\kk)$, $b^*(\kk)$
 in $ H_{3}^L,H_{4}^L$.

The Liouvillean $L^L_3$ can be written as the sum of two contributions $L^L_{3,1} + L^L_{3,0}$. The first is the one relevant for physical processes involving three quasi-particles where two of them are destroyed and one is created or vice versa, two are created and the remaining one is destroyed. It can be written as \begin{align}\label{L31}
	&L_{3,1}^L = \notag \\
	&\sum_{\substack{\p,\q,\\ \p+\q\in\Xi_L^>}} V_{\p,\q} \left(b^*_{\beta,\le}(\q+\p)b_{\beta,\le}(\p)b_{\beta,\le}(\q) - b^*_{\beta,\re}(\q+\p)b_{\beta,\re}(\p)b_{\beta,\re}(\q)  + \text{h.c.} \right),
\end{align} where the effective interaction is \begin{align}\label{eq: eff interaction}
	V_{\p,\q} = &\frac{\sqrt{\nu}\hat{v}(\p)}{L\frac{d}{2}\sqrt{\hat{v}(\0)}}(c_\p-s_\p)(c_\q c_{\q+\p} + s_\q s_{\q+\p})\notag \\
	&+\frac{\sqrt{\nu}\hat{v}(\p+\q)}{L^\frac{d}{2}\sqrt{\hat v(\0)}}(s_{\p+\q} - c_{\p+\q})s_\p c_\q.
\end{align} This effective potential is related to the function $j(\p+\q;\p,\q)$ in Equation \eqref{eq: jj}, as 
 \begin{align}\label{eq: relation vjj}
	V_{\p,\q} + V_{\q,\p} = & \frac{1}{L^\frac{d}{2}}\sqrt{\frac{\nu}{\hat{v}(\0)}}\left( \hat{v}(\p)(c_\p -s_\p)(c_\q c_{\q+\p} + s_\q s_{\q+\p}) \right. \notag \\ 
	&+ \hat{v}(\q)(c_\q -s_\q)(c_\p c_{\q+\p} + s_\p s_{\q+\p})\notag \\ 
	&\left. + \hat{v}(\p+\q)(s_{\p+\q} -c_{\q+\p})(c_\p s_\q +c_\q s_{\p}) \right) \notag \\
	= & \frac{1}{L^\frac{d}{2}} j(\p+\q;\p,\q).
\end{align} 

The second contribution $L_{3,0}^L$ corresponds to processes where three particles are either created or destroyed \begin{align}\label{L32}
	L^L_{3,0} = \sum_{\substack{\p,\q,\\ \p+\q\in\Xi_L^>}} U_{\p,\q} \left(b^*_{\beta,\le}(-\q-\p)b^*_{\beta,\le}(\p)b^*_{\beta,\le}(\q) \notag \right. \\  \left.- b^*_{\beta,\re}(-\q-\p)b^*_{\beta,\re}(\p)b^*_{\beta,\re}(\q)  + \text{h.c.} \right),
\end{align}
with 
\begin{align}
	U_{\p,\q} = \frac{\sqrt{\nu}\hat{v}(\p)}{L^\frac{d}{2} \sqrt{\hat v(\0)}}(c_{\p+\q}s_\p s_\q - s_{\p+\q} c_\p c_\q).
\end{align} This effective potential is related to the function $\kappa(\p+\q,\p,\q)$  in Equation \eqref{eq: kappa} as 
\begin{align}\label{eq: relation ukappa}
	&\left(U_{\k,-\p} + U_{-\p,\k} + U_{\p,\k-\p} + U_{\p-\k,\k} + U_{\k,\p-\k}+U_{\k-\p,\p}\right) = \notag \\
	=&\frac{1}{L^\frac{d}{2}}\sqrt{\frac{\nu}{\hat{v}(\0)}}\left( \hat{v}(\k)(s_\k-c_\k)(c_\p s_{\k-\p} + c_{\k-\p}s_\p )\right. \notag\\ 
	&+  \hat{v}(\p)(s_\p-c_\p)(c_\k s_{\k-\p} + c_{\k-\p}s_\k ) \notag\\ 
	& \left.+  \hat{v}(\k-\p)(s_{\k-\p}-c_{\k-\p})(c_\k s_{\p} + c_{\p}s_\k ) \ \right) = \frac{1}{L^\frac{d}{2}}\kappa(\k,\p,\k-\p).
\end{align} 

The Liouvillean $L^L_4$ can  be written as the sum of five terms.
Two of them are of the second order:
\begin{align}
  	L_{2,1}^L = & \frac{1}{L^d} \sum_{\p,\q\in\Xi_L^>}\left( \hat{v}(\0)(c_\p^2+s_\p^2)s_\q^2 b^*_{\beta,\le}(\p)b_{\beta,\le}(\p)\right. \notag \\ 
	&\left. + \hat{v}(\q-\p)((c_\p^2+s_\p^2)s_\q^2 + 2c_\p c_\q s_\q s_\p) b^*_{\beta,\le}(\p)b_{\beta,\le}(\p) - (\le\leftrightarrow\re)\right),\label{eq: L42}
\\	L^L_{2,0} = & -\frac{1}{2L^d}\sum_{\p,\q\in\Xi_L^>}\left(\hat{v}(\0)2c_\p s_\p s_\q^2 \right. \notag \\ 
	&\left. +\hat{v}(\q-\p)((s_\p^2 + c_\p^2)c_\q s_\q + 2c_\p s_\p s_\q^2)\right)(b_{\beta,\le}^*(\p)b_{\beta,\le}^*(-\p) \notag \\
	&+ \text{h.c.} - (\le\leftrightarrow\re)).
\end{align}
The remaining three are of the fourth order:
\begin{align}
	L^{L}_{4,2} =& \frac{1}{2L^{d}}\sum_{\substack{\p,\q, \\ \p+\k, \q-\k\in\Xi_L^>}} \left[\hat{v}(\k)(c_{\k+\p}c_{\q-\k}c_{\q}c_{\p} +s_{\p+\k}s_{\q-\k}s_\p s_\q \right. \notag \\
	&\left. + 2c_{\k+\p}c_{\p}s_{\q}s_{\q-\k} )+ 2\hat{v}(\p+\q)c_\q c_{\q-\k}s_\p s_{\p+\k}\right] \notag \\ 
	&\times \left(b_{\beta,\le}^*(\k+\p)  b_{\beta,\le}^*(\q-\k)
   b_{\beta,\le}(\q)  b_{\beta,\le}(\p) - (\le\leftrightarrow\re
   )\right), \label{eq: L41}\\
	L^L_{4,1} = &- \frac{1}{L^{d}}\sum_{\substack{\p,\q, \\ \p+\k ,\q-\k\in\Xi_L^>}} \hat{v}(\k) \notag  \\ 
	&\times \left(c_{\k+\p}c_{\q-\k}c_\p s_\q b^*_{\beta,\le}(\p+\k) b^*_{\beta,\le}(\q-\k)b^*_{\beta,\le}(-\q)b_{\beta,\le}(\p)\right. \notag  \\
	&+\left. c_{\p+\k} s_{\q-\k} s_\q s_\p b^*_{\beta,\le}(\p+\k) b^*_{\beta,\le}(-\p)b^*_{\beta,\le}(\q)b_{\beta,\le}(\k-\q) \right. \notag \\
	& \left. + \text{h.c.} - (\le\leftrightarrow\re)\right) \\
	L^L_{4,0} = & \frac{1}{2L^{d}}\sum_{\substack{\p,\q, \\ \p+\k, \q-\k\in\Xi_L^>}} \hat{v}(\k)c_{\k+\p}c_{\q-\k}s_\q s_\p \notag  \\ 
	&\times \left(b^*_{\beta,\le}(\p+\k) b^*_{\beta,\le}(\q-\k)b^*_{\beta,\le}(-\p)b_{\beta,\le}^*(-\q) + \text{h.c.} - (\le\leftrightarrow\re)\right) .
\end{align} In the above we labeled with $(\le\leftrightarrow\re)$ the terms which can be obtained from the ones explicitly written by exchanging left and right particles, which is equivalent to exchanging the labels  $\le$ and $\re$.

The  KMS vector for the Bogoliubov Liouvillean, in the Araki-Woods
representation,  is simply the vacuum $\Omega$. 
By general theory, we are guaranteed that the perturbed KMS state exists
and is an eigenvector of the full Liouvillean with eigenvalue $0$.
It can be obtained by perturbation theory. All terms in the
perturbation expansion for its energy will give $0$. In particular, at the order
$\kappa$ for a fixed $L$ we have
\begin{align}&
0=
  \big(\Omega|L_4^{L}
  \Omega\big)
  -
      \lim_{\epsilon\searrow0}
\Big(  \Omega|L_3^{L}
      \big(L_\Bog^L-\i\epsilon\big)^{-1}L_3^{L}
      \Omega\Big).\label{fermi1-}\end{align}

\subsection{Elementary quasiparticle from joint energy-momentum spectrum}

If one could make sense of the
Hamiltonian or Liouvillean with  positive $\nu$ in thermodynamic limit, then one  should
expect  that
the joint energy-momentum
spectrum is absolutely continuous everywhere, without singular
shells, except for the ground state at zero temperature and the KMS
vector at positive temperature. In particular, the thermodynamic limit
probably destroys the singular shells
\eqref{destroy1a} and \eqref{destroy2}.

Thus even if the Hamiltonian $H$ and the Liouvillean $L$ can be defined in thermodynamic
limit, rigorously it is not clear  what are quasiparticles beyond
the Bogoliubov approximation.
Still,
physicists seem to
believe that quasiparticles are a useful concept.


At zero temperature  the interacting Hamiltonian is bounded from below and
one can expect that it makes sense to speak about the infimum of the
energy-momentum spectrum  in thermodynamic limit, as described in
\cite{CDZ_2009}. Typically, in
the unperturbed Bogoliubov Hamiltonian this infimum  does
not coincide with the  quasiparticle spectrum
\eqref{destroy1a}, even for small momenta.  So in the interacting case this infimum is
not a good definition of the quasiparticle spectrum.
Anyway, for  zero temperature  at least   a certain 
interesting purely spectral information seems to survive thermodynamic limit.

For positive temperatures, the joint spectrum of the interacting Liouvilean and
the standard momentum in thermodynamic limit is almost certainly the full $\rr^{1+d}$, so from the spectrum we
cannot extract anything interesting.

    \subsection{Decay of one-quasiparticle states}
    \label{Decay of one-quasiparticle states}


As we mentioned above, for the pure
    Bogoliubov Liouvillean in the Araki-Woods representation the standard vector representative of the
    KMS state is the Fock vacuum $\Omega$. Switching on the interaction modifies
    the KMS vector, but it remains an eigenvector with eigenvalue $(0,\mathbf{0})$.

    For the pure Bogoliubov
    Liouvillian the one-quasiparticle state of momentum
    $\k$ has the standard
    vector representative,
    which in the Araki-Woods representation is given by
    \beq
    b_\le^*(\kk) b_\re^*(\kk)\Omega
.\eeq
    It is an eigenstate of $(L^L_\Bog,\mathbf{P}^{\mathrm{st,L}})$ with eigenvalue $(0,\mathbf{0})$.

After switching  the interaction in finite volume we expect that
perturbation theory will yield perturbed eigenstates of the
Liouvillean. However, these eigenstates will not survive thermodynamic limit.
In fact, we will find that already at the lowest nontrivial order of
perturbation theory we will obtain a nonzero imaginary contribution.

The computation is complicated by the fact that we are dealing with
two limits: small coupling constant $\kappa$ and large size $L$.
The value of the shift of the eigenvalue at the order $\kappa^1$ will be
computed from the formula
\begin{align}\label{formu1}
  \xi(\kk)
  :=		\lim_{\epsilon\searrow
  0}\lim_{L\rightarrow\infty}&\bigg{\{}
\bigg{\{} \left(b_\le^*(\k)b_\re^*(\k)\Omega|
		L_4^Lb_\le^*(\k)b_\re^*(\k)\Omega\right)
                             \\& - \left(L_{3}^Lb_\le^*(\k)b_\re^*(\k)\Omega|
		\big(L_\Bog^L-\i\epsilon\big)^{-1}L_3^Lb_\le^*(\k)b_\re^*(\k)\Omega\right) \label{formu2} \\ -& \left(\Omega\right.\left|L_4^L\Omega\right) +\left(L_3^L \Omega\right.\left| (L_\Bog^L  -\i\epsilon)^{-1}L_3^L\Omega\right) \bigg{\}} .
\label{eq: physical correction braket I..}\end{align}
We will see that the shift $\xi(\kk)$ is purely imaginary.

Above, \eqref{formu1}+\eqref{formu2} is the naive expression  for the
shift of the eigenvalue in thermodynamic limit.
In fact, for finite $L$
the shift is given by
\begin{align}\label{formua1}\lim_{\epsilon\searrow
  0}&
\bigg{\{} \left(b_\le^*(\k)b_\re^*(\k)\Omega|
		L_4^Lb_\le^*(\k)b_\re^*(\k)\Omega\right)
                             \\\label{formua2}& - \left(L_{3}^Lb_\le^*(\k)b_\re^*(\k)\Omega|
		\big(L_\Bog^L-\i\epsilon\big)^{-1}L_3^Lb_\le^*(\k)b_\re^*(\k)\Omega\right) \bigg{\}}\end{align}
However, we need to
subtract from this expression the term
\eqref {eq: physical correction braket I..}, which for fixed $L$ and
$\epsilon\searrow0$ goes to zero by
 \eqref{fermi1-}, is however extensive in volume.

Let us check  that the subtraction of
\eqref {eq: physical correction braket I..} is appropriate.
 Introduce  the  
shorthand $b_{\le\re}(\k):=b_\le(\k)b_\re(\k)$.
The expression from \eqref{formua1} and \eqref{formua2}
can be transformed as follows:
\begin{align}\label{expro}
  &
               \big(b_{\le\re}^*
               (\k)\Omega|L_4^L
  b_{\le\re}^*
               (\k)\Omega\big)-
\Big( b_{\le\re}^*
               (\k)\Omega|L_3^L
      \big(L_\Bog^L-\i\epsilon\big)^{-1}L_3^L
      b_{\le\re}^*
               (\k)\Omega\Big)
  \\\label{kmsterm}
  =&
  \big(\Omega|L_4^L
  \Omega\big)-
\Big(  \Omega|L_3^L
      \big(L_\Bog^L-\i\epsilon\big)^{-1}L_3^L
             \Omega\Big)\\\notag&
+\notag(\Omega|[b_{\le\re}(\k),
  [L_4^L, b_{\le\re}^*(\k)]]\Omega\big)\\\notag&
                                    -\big(\Omega|[b_{\le\re}(\k), L_{3}^L](L_\Bog^L-\i\epsilon)^{-1}[L_3^L, 
                       b_{\le\re}^*(\k)]\Omega\big)  \notag \\ &-\big(\Omega|[b_{\le\re}^*(\k), L_{3}^L](L^L_\Bog-\i\epsilon)^{-1}[L_3^L, 
                       b_{\le\re}(\k)]\Omega\big)\notag \\
&  -\big(\Omega| L_{3}^L(L_\Bog^L-\i\epsilon)^{-1}[b_{\le\re}(\k) ,
                                                           [L_3^L, b_{\le\re}^*(\k)]]\Omega\big)
                                                           \notag \\ &-\big(\Omega|[b_{\le\re}^*(\k),
  [L_3^L, b_{\le\re}(\k)]](L_\Bog^L-\i\epsilon)^{-1}L_3^L\Omega\big)
      .\notag\end{align}
    To derive this we used
    \begin{align}
    &  [b_{\le\re}(\k), b_{\le\re}^*(\k)]=b_\re^*(\k) b_\re(\k)+
      b_\le^*(\k) b_\le(\k)+1,\\&b_{\le\re}^*(\k)(L_\Bog^L-\i\epsilon)^{-1}=
      (L_\Bog^L-\i\epsilon)^{-1}b_{\le\re}^*(\k),\\
&      b_{\le\re}(\k)(L_\Bog^L-\i\epsilon)^{-1}=
      (L_\Bog^L-\i\epsilon)^{-1}b_{\le\re}(\k).
            \end{align}

    In the  formula for \eqref{expro},  by \eqref{fermi1-}, the term \eqref{kmsterm} goes to zero as
$\epsilon\to0$.
At the same time,  for a fixed $\epsilon>0$, it is large,
diverging as $L\to\infty$.
Therefore we need to subtract it.
All other terms  contain a commutator, so that one can
hope that they do not  blow up as
$L\to\infty$.

  \subsection{Two-point functions}

  Let us go back to the  abstract setting used in Subsection \ref{Positive temperatures}.
  Let $A_1,\dots,A_n\in B(\cH)$ be operators. Physicists often consider
 {\em $n$-point correlation functions} defined as
\beq  \i^{\frac{n}{2}}\omega_\beta\Big(\mathrm{T}\big(A_n(t_n)\cdots
A_1(t_1)\big)\Big),\eeq
where $\mathrm{T}(\cdots)$ denotes the time ordering \cite{FW_2003}.
(The coefficient
$\i^{\frac{n}2}$ follows one of the possible conventions, which is
convenient in  applications to the Bose gas).
Of special importance are two-point functions
\begin{align}  &\i\omega_\beta\Big(\mathrm{T}\big(A_2(t_2) 
  A_1(t_1)\big)\Big)\\=&\i\theta(t_2-t_1) \omega_\beta\big(A_2(t_2) 
  A_1(t_1)\big)+\i\theta(t_1-t_2) \omega_\beta\big(A_1(t_1) 
                         A_2(t_2)\big)\\
  =&\i\theta(t_2-t_1) \big(\Omega_\beta|A_2\e^{\i(t_1-t_2)L }
  A_1\Omega_\beta\big)+\i\theta(t_1-t_2) \big(\Omega_\beta|A_1\e^{\i(t_2-t_1)L}
                      A_2\Omega_\beta\big). \label{line3} \end{align}
             where in          \eqref{line3} we use the standard
             representation and we omitted $\pi_\le$. Note that
             \eqref{line3} depends only on $t:=t_2-t_1$, so that we can
             define
             \beq
             G_\beta(A_2,A_1,t):=\eqref{line3}.\eeq
             Let $E\in\rr$. The two-point functions in the energy representation are
             \begin{align}&\quad
               G_\beta(A_2,A_1;E)\;:=\;\int G_\beta(A_2,A_1,t)\e^{-\i tE}\d t\\
               &=\lim_{\epsilon\searrow0}\i\int_0^\infty\big(\Omega_\beta|A_2\e^{-\i 
                 tL-\epsilon t-\i tE}A_1\Omega_\beta\big)\d t\notag\\&
    +\lim_{\epsilon\searrow0}\i\int_0^\infty\big(\Omega_\beta|A_1\e^{-\i 
                 tL-\epsilon t+\i tE}A_2\Omega_\beta\big)\d t\notag\\
               &=\big(\Omega_\beta|A_2(L-\i0+E)^{-1}A_1\Omega_\beta\big)
    +\big(\Omega_\beta|A_1(L-\i0-E)^{-1}A_2\Omega_\beta\big).\label{line4}
             \end{align}
\eqref{line4} is sometimes called the {\em Lehmann representation}.

As a side remark let us mention that it is often possible to simplify the above formula if we have time
reversal invariance. More precisely, assume that there exists an
antilinear involution $\Phi\mapsto\bar\Phi$, such that $H=\bar H$, and
hence $\bar L=L$ and $\bar\Omega_\beta=\Omega_\beta$.
Let $A_1,A_2$ be also real, that is, $\bar{ A_1}=A_1$, $\bar{ A_2}=A_2$.
Then the second term in \eqref{line4} can be written as
\beq
\big(\Omega_\beta|A_2^*(L-\i0-E)^{-1}A_1^*\Omega_\beta\big).\eeq 
Thus, if $A_1$ and $A_2$ are self-adjoint, then
\beq
G_\beta(A_2,A_1;E)=\Big(\Omega_\beta|A_2 2L\big((L-\i0)^2-E^{2}\big)^{-1}A_1\Omega_\beta\Big).\eeq

\subsection{Elementary quasiparticle spectrum from two-point functions}

Physicists believe that quasiparticles can be seen in correlation
functions. The simplest choice for observables $A_1,A_2$ in these correlation
functions seems to be $a_\kk^*$, $a_\kk$. This choice is perhaps
not the most convenient experimentally, however it seems
mathematically the most natural and was adopted in classic physics
papers, e.g. by Beliaev and Hugenholtz-Pines.
Thus we define
\begin{align}&\begin{bmatrix}
  G_{11}(E,\kk)&G_{12}(E,\kk)\\
  G_{21}(E,\kk)&G_{22}(E,\kk)
  \end{bmatrix}\\:=&\i\int\d t\e^{-\i
    tE}\begin{bmatrix}\omega_\beta\big(\mathrm{T} (a_\kk^*(t)a_\kk(0))\big)&\omega_\beta\big(\mathrm{T}(a_\kk^*(t)a_{-\kk}^*(0))\big)\\
    \omega_\beta\big(\mathrm{T}(a_{-\kk}(t)a_{\kk}(0))\big)&
        \omega_\beta\big(\mathrm{T}(a_{-\kk}(t)a_{-\kk}^*(0))\big)\end{bmatrix}.\end{align}
  
A linear transformation connects
$a_\kk^*$, $a_\kk$ with $b_\kk^*$, $b_\kk$. Therefore, one
can instead consider
 \begin{align}\label{green2}
 	&\i\int\d t\e^{-\i
 		tE}\begin{bmatrix}\omega_\beta\big(\mathrm{T}\left(b_\kk^*(t)b_\kk(0)\right)\big)&\omega_\beta\big(\mathrm{T}\left(b_\kk^*(t)b_{-\kk}^*(0)\right)\big)\\
 		\omega_\beta\big(\mathrm{T}\left(b_{-\kk}(t)b_{\kk}(0)\right)\big)&
 		\omega_\beta\big(\mathrm{T}\left(b_{-\kk}(t)b_{-\kk}^*(0)\right)\big)\end{bmatrix}.\end{align}
 	
 	For the pure Bogoliubov Hamiltonian, the evaluation of \eqref{green2} can be carried out exactly via the \textit{Lehmann representation} \eqref{line4}. As an example, we show the computation of the first diagonal entry \begin{align}
 		&\i\int \d t \; \e^{-\i E t}\omega_\beta\left( \mathrm{T}\left( b_\k^*(t)b_\k(0)\right)\right) \notag \\
 		=&\label{eq: araki rep green}\left(\Omega | b_{\beta,\le}^*(\k) (L_\Bog -\i0+ E)^{-1} b_{\beta,\le}(\k)\Omega \right) +\left(\Omega | b_{\beta,\le}(\k) (L_\Bog-\i0 - E)^{-1} b_{\beta,\le}^*(\k)\Omega \right)\\
 		=& -\rho(\k) \frac{1}{\omega_\Bog(\k)-\i0-E} +(1+\rho(\k))\frac{1}{\omega_\Bog(\k)-\i0- E}\notag \\ =& \label{eq: time ordered green} \frac{1}{\omega_\Bog(\k)-\i0-E} ,
 	\end{align} where in the second line \eqref{eq: araki rep green} we have written  the state in the Araki-Woods representation. The final result is the following one 
 	\begin{equation}\label{bump3}
 		\begin{bmatrix}(\omega_\Bog(\k)-\i0-E)^{-1}&0\\0&(\omega_\Bog(\k)-\i0+ E)^{-1}\end{bmatrix} 
 	\end{equation}
In particular, for $E$ approaching the real line from above, the imaginary part of
\eqref{bump3} will be proportional to the delta function along the quasiparticle
dispersion relation:
 	\begin{equation}\label{bump3a}
\i\pi 		\begin{bmatrix}\delta(\omega_\Bog(\k)-E)&0\\0&\delta(\omega_\Bog(\k)+ E)\end{bmatrix} 
 	\end{equation}


 When we compute the dispersion relation in the interacting case
 in perturbation theory we expect to obtain a  complex 
 dispersion relation with a negative imaginary part $\kk\mapsto\omega(\kk)$. 
Then we can omit $-\i0$ and the two-point functions look like
\begin{align}&\begin{bmatrix}
  G_{11}(E,\kk)&G_{12}(E,\kk)\notag \\
  G_{21}(E,\kk)&G_{22}(E,\kk)
  \end{bmatrix}\\=& Z^{\pm}(E,\kk)\frac{1}{\omega(\kk)\mp E}+\text{regular},\quad 
                   \text{near}\quad \pm\omega_\Bog(\kk)=E.
                   \end{align}
  Above,  $Z^\pm(E,
  \kk)$ are some 
  slowly varying matrices.
Thus
\eqref{bump3a} should be replaced by a  ``ridge of 
mountains'' along the  quasiparticle 
dispertion relation, perhaps slightly shifted  in the real
direction. The width of this ``ridge'' is determined  by $\Im\omega(\kk)$.
In the lowest order of
perturbation theory, the shift of $\omega(\kk)$ includes the Feynman-Hellman term for
$L_4$ and the Fermi Golden Rule term for $L_3$.
The shift of the dispersion relation should be
computed from the formula
\begin{align}\label{formi1}
  \delta(\kk)
  :=		\lim_{\epsilon\searrow
  0}\lim_{L\rightarrow\infty}&\bigg{\{}
\bigg{\{} \left(b_\le^*(\k)\Omega|
		L_4^Lb_\le^*(\k)\Omega\right)
                             \\&- \left(L_{3}^Lb_\le^*(\k)\Omega|
		\big(L_\Bog^L-\omega_\Bog(\k)-\i\epsilon\big)^{-1}L_3^Lb_\le^*(\k)\Omega\right)
                           \label{formi2}     \\
  -& \left(\Omega\right.\left|L_4^L\Omega\right) +\left(L_3^L \Omega\right.\left| (L_\Bog^L  -\i\epsilon)^{-1}L_3^L\Omega\right) \bigg{\}} 
\label{eq: physical correction braket I.}\end{align}
Again, from the naive expression  \eqref{formi1}+\eqref{formi2} we
need to subtract the extensive quantity 
\eqref{eq: physical correction braket I.}.

    
The formula for      $\delta(\kk)$ is suggested by the
following computation:
\begin{align}\label{forri}&
               \big(b_{\le}^*
               (\k)\Omega|L_4^L
  b_{\le}^*
               (\k)\Omega\big)-
\Big( b_{\le}^*
               (\k)\Omega|L_3^L
      \big(L_\Bog^L-\omega_\Bog(\k)-\i\epsilon\big)^{-1}L_3^L
      b_{\le}^*
               (\k)\Omega\Big)
  \\\label{kmsterm+}
  =&
  \big(\Omega|L_4^L
  \Omega\big)-
\Big(  \Omega|L_3^L
      \big(L_\Bog^L-\i\epsilon\big)^{-1}L_3^L
             \Omega\Big)\\\notag&
+\notag(\Omega|[b_{\le}(\k),
  [L_4^L, b_{\le}^*(\k)]]\Omega\big)\\\notag&
                                    -\big(\Omega|[b_{\le}(\k),
                                                  L_{3}^L](L_\Bog^L-\omega_\Bog(\k)-\i\epsilon)^{-1}[L_3^L,  
                                                  b_{\le}^*(\k)]\Omega\big)  \\
  &-\big(\Omega|[b_{\le}^*(\k), L_{3}^L](L_\Bog^L+\omega_\Bog(\k)-\i\epsilon)^{-1}[L_3^L, 
                       b_{\le}(\k)]\Omega\big)\notag\\
&  -\big(\Omega| L_{3}^L(L_\Bog^L-\i\epsilon)^{-1}[b_{\le}(\k) ,
                                                           [L_3^L, b_{\le}^*(\k)]]\Omega\big)
                                                           \notag \\ &-\big(\Omega|[b_{\le}^*(\k),
  [L_3^L, b_{\le}(\k)]](L_\Bog^L-\i\epsilon)^{-1}L_3^L\Omega\big)
      .\notag\end{align}
    To derive this we used
    \begin{align}
    &  [b_{\le}(\k), b_{\le}^*(\k)]=1,\\&b_{\le}^*(\k)(L_\Bog^L+\omega_\Bog(\k)-\i\epsilon)^{-1}=
      (L_\Bog^L-\i\epsilon)^{-1}b_{\le}^*(\k),\\
&      b_{\le}(\k)(L_\Bog^L-\omega_\Bog(\k)-\i\epsilon)^{-1}=
      (L_\Bog^L-\i\epsilon)^{-1}b_{\le}(\k).
            \end{align}
            Now \eqref{forri} consists of an extensive quantity
            \eqref{kmsterm+},
            and a few terms  involving commutators, which should have
a            better behavior as $L\to\infty$.  
            \eqref{kmsterm+} goes to zero as $\epsilon\searrow0$ by
            \eqref{fermi1-}. It needs to be subtracted.

    Using the above formulas, one can verify that in the thermodynamic
    limit $$\lim_{|\k|\to 0}|\Re\delta(\kk)| = +\infty.$$ This is a
    consequence of the  naive \textit{c}-number, which misses relevant
    terms in the construction of the effective Liouvillean. This
    problem is solved in \cite{DP2026} in the zero temperature
    setting.
    Nonetheless,
    $\Im\delta(\kk)$ is finite
    and it is consistent with the imaginary shift of Subsection
    \ref{Decay  of one-quasiparticle states}: \beq 
\Im    \xi(\kk)=2\Im \delta(\kk).\eeq 

\subsection{Computation of matrix elements for the Liouvillean}  As we are interested in finite temperature computations, we will work in the Araki-Woods representation, where the interaction is implemented by the Liouvillean $L^L_3$ and $L^L_4$. 

We begin by computing the relevant matrix elements for $\xi(\k)$. The final result   is summarized in the following \begin{prop}\label{prop: xi computation}
Fix  $\nu_0\in ]0,+\infty[$ and suppose the potential satisfies Assumptions \eqref{eq: continuity of Fourier}--\eqref{eq: no plateau} for all $\nu\in ]0, \nu_0]$. Then, we have the following \begin{align}
	\xi(\k) = -\i\pi \int \frac{\d \p}{(2\pi)^d}j(\k;\p,\p-\k)^2\rho(\p)\rho(\p-\k)\rho(\k)\notag \\  \times  \left(1 -\e^{\beta\frac{\omega_\Bog(\p)+\omega_\Bog(\p-\k)+\omega_\Bog(\k)}{2}}\right)^2 \delta(\omega_\Bog(\p)+\omega_\Bog(\p-\k)-\omega_\Bog(\k))\label{eq: int1}\\ -\i2\pi\int \frac{\d \p}{(2\pi)^d}j(\p-\k;\p,\k)^2\rho(\p)\rho(\p-\k)\rho(\k)\notag \\  \times \left(\e^{\beta\frac{\omega_\Bog(\k)+\omega_\Bog(\p-\k)}{2}} -\e^{\frac{\beta}{2}\omega_\Bog(\p)}\right)^2 \delta(\omega_\Bog(\p)+\omega_\Bog(\k)-\omega_\Bog(\p-\k)),\label{eq: int2}
\end{align} where $j(\k;\p,\q)$ is the function defined in \eqref{eq: jj}.
\end{prop}

\begin{remark}
	Note that Integrals \eqref{eq: int1} and \eqref{eq: int2}
         coincide with the ones in the introduction, \eqref{gamma_B} and \eqref{gamma_L} respectively, apart from a factor of $2$. This additional factor appears in \eqref{eq: int1}, \eqref{eq: int2} because the vector state $b_\re^*(\k)b_{\le}^*(\k)\Omega$ describes processes with two incoming quasi-particles, a left one and a right one. 
\end{remark}

\begin{proof}
	 First of all, we note the following: $L^L_4$ and $L^L_3$ are antisymmetric under the exchange $\l\leftrightarrow \re$ while the vectors $\Omega$ and $b_{\le\re}^*\Omega$ are symmetric with respect to the same exchange. Thus, we deduce immediately that \begin{equation}
		(\Omega| L^L_4 \Omega ) = (b_{\le\re}^*\Omega| L^L_4 b_{\le\re}^*\Omega ) = 0, 
	\end{equation} and the same holds for $L^L_3$ matrix elements. Thus, we only need to consider three vectors \begin{equation}\label{eq: matrix for L3}
		[b^*_{\le\re}(\k), L^L_3]\Omega, \qquad [b_{\le\re}(\k), L^L_3]\Omega,\qquad [  b_{\le\re}(\k),[b^*_{\le\re}(\k), L^L_3]]\Omega,
	\end{equation}  and the matrix elements of $(L_\Bog -\i\epsilon)^{-1}$ with respect to these. With the constraints $\p,\q,\p+\q\neq \0$, the only non-zero vectors which might appear are all of the following  form \begin{align}
		&\label{eq: vector1}[b_{\le\re}^*(\k), b_\le^*(\p)b_\le^*(\q)b_\le(\p+\q)]\Omega = -\delta_{\k,\p+\q} b_\re^*(\k)b_\le^*(\k-\p)b_\le^*(\p)\Omega\\
		&\label{eq: 2vector}[b_{\le\re}^*(\k), b_\le^*(\p+\q)b_\re^*(\q)b_\le(\p)]\Omega = -\delta_{\k,\p} b_\le^*(\q+\k)b_\re^*(\q)b_\re^*(\k)\Omega\\
		&\label{eq: 4vector} [b_{\le\re}^*(\k), b_\re^*(\q+\p)b_\re^*(-\q)b_\le(-\p)]\Omega = -\delta_{-\p,\k} b_\re(\k)^*b_\re^*(\q-\k)b_\re^*(-\q)\Omega 
	\end{align} plus their $\le\leftrightarrow \re$ counterparts. All of \eqref{eq: vector1}---\eqref{eq: 4vector} are eigenvectors of $(L_\Bog-\i\epsilon)^{-1}$ and they are orthogonal to each other. In particular, we have\begin{equation}
	[b_{\le\re}(\k),[b_{\le\re}^*(\k),L_{3}^L]]\Omega=  [b_{\le\re}(\k), L^L_3]\Omega = 0.
	\end{equation}
	
	Thus, we see that the only matrix element which will contribute to $\xi(\k)$ is given by  \begin{equation}
	 -\big([L^L_{3},b_{\le\re}^*(\k)]\Omega|(L_\Bog^L-\i\epsilon)^{-1}[L_{3}^L, 
	b_{\le\re}^*(\k)]\Omega\big).
	\end{equation} Expanding the operators $b_{\beta,\le},\;b_{\beta,\re}$ in terms of $b_\le,\; b_\re$  and using Equations \eqref{eq: vector1}--\eqref{eq: 4vector}, we see that \begin{align}
	&[L^L_{3,1},b_{\le\re}^*(\k)]\Omega \notag \\ =& \sum_{\p\neq \0,\p\neq \k} V_{\p,\k-\p} 
	\left(\sqrt{(1+\rho(\p))(1+\rho(\p-\k) )(1+\rho(\k)) }  \notag \right. \\&\left.-\sqrt{\rho(\p)\rho(\p-\k)\rho(\k)}\right)b_\re^*(\k)b_\le^*(\p)b_\le^*(\k-\p)\Omega - (\le\leftrightarrow\re)\label{eq: beliaev1}  \\
	& + \sum_{\p,\p-\k\in\Xi_L^>}(V_{\k,\p} + V_{\p,\k}) \left(\right. \sqrt{(1+\rho(\p- \k))(1+\rho(\k))\rho(\p )} \notag \\ 
	 & \left. -\sqrt{\rho(\k)\rho(\p-\k)(1+\rho(\p)} \right.) b_\le^*(\k-\p)b_\re^*(\k) b_\re^*(-\p)\Omega  - (\le\leftrightarrow\re)\label{eq: landau1}
	\end{align}  and \begin{align}
	&[L^L_{3,0},b_{\le\re}^*(\k)]\Omega \notag \\ =& \sum_{\p,\p-\k\in\Xi^>_L} \left(U_{\k,-\p} + U_{-\p,\k} + U_{\p,\k-\p}\right) 	\left(\sqrt{(1+\rho(\k))\rho(\p)\rho(\p-\k) }  \notag \right. \\&\left.-\sqrt{\rho(\k)(1+\rho(\p-\k))(1+\rho(\p))}\right)b_\re^*(\k)b_\re^*(\p-\k)b_\re^*(-\p)\Omega - (\le\leftrightarrow\re)\label{eq: null1}.
	\end{align}
	Now, the vectors appearing in the sum in \eqref{eq: null1} are eigenvectors of 
	\((L_{\Bog} - \i\epsilon)^{-1}\) with eigenvalues
	\[
	-\bigl(\omega_{\Bog}(\k) + \omega_{\Bog}(\p) + \omega_{\Bog}(\p-\k) + \i\epsilon \bigr)^{-1},  
	\]
	while the terms labeled by \((\le \leftrightarrow \re)\) have eigenvalues
	\[
	\bigl(\omega_{\Bog}(\k) + \omega_{\Bog}(\p) + \omega_{\Bog}(\p-\k) - \i\epsilon \bigr)^{-1}.
	\]
	Thus, we obtain
	\begin{align}
	& \big([L^L_{3,0},b_{\le\re}^*(\k)]\Omega|(L_\Bog^L-\i\epsilon)^{-1}[L_{3,0}^L, 
	b_{\le\re}^*(\k)]\Omega\big) \notag \\
	=&\frac{1}{2}\sum_{\p,\p-\k\in\Xi^>_L}\left(U_{\k,-\p} + U_{-\p,\k} + U_{\p,\k-\p} + U_{\p-\k,\k} + U_{\k,\p-\k}+U_{\k-\p,\p}\right)^2\label{eq: lineacomplicata1} \\ 	&\times\left(\sqrt{(1+\rho(\k))\rho(\p)\rho(\p-\k) }  -\sqrt{\rho(\k)(1+\rho(\p-\k))(1+\rho(\p))}\right)^2\notag \\&\times\frac{2\i\epsilon}{(\omega_\Bog(\k)+\omega_\Bog(\p)+\omega_\Bog(\p-\k))^2 +\epsilon^2} \label{eq: null2},
	\end{align} where in the first line \eqref{eq: lineacomplicata1} we have symmetrized the effective potential on account of the symmetry $\p\leftrightarrow\k-\p$ of the whole expression. 
As the effective potential $U_{\q,\p}$ contains a factor $\frac{1}{L^\frac{d}{2}}$, the thermodynamic limit of \eqref{eq: null2} does converge as $L\to+\infty$ to a finite integral. After the thermodynamic limit, we can consider the $\epsilon\searrow0$ one. In this limit, the last factor in \eqref{eq: null2} converges in distribution to $2\pi \mathrm{i}\delta(\omega_\Bog(\k)+\omega_\Bog(\p)+\omega_\Bog(\p-\k))=0$.  Moreover, from the expressions \eqref{eq: beliaev1},\eqref{eq: landau1} and \eqref{eq: null1} and a simple application of Wick's theorem, we see that \begin{equation}
	\big([L^L_{3,1},b_{\le\re}^*(\k)]\Omega|(L_\Bog^L-\i\epsilon)^{-1}[L_{3,0}^L, 
	b_{\le\re}^*(\k)]\Omega\big) = 0.
	\end{equation} Hence, $L^L_{3,0}$ does not contribute to the damping. All that is left is to compute \begin{align}
		&\big([L^L_{3,1},b_{\le\re}^*(\k)]\Omega|(L_\Bog^L-\i\epsilon)^{-1}[L_{3,1}^L, 
	b_{\le\re}^*(\k)]\Omega\big) \notag \\
	=&\label{eq: symm2} \frac{1}{2}\sum_{\p,\p-\k\in\Xi^>_L} (V_{\p,\k-\p} + V_{\k-\p,\p})^2\left(\sqrt{(1+\rho(\p))(1+\rho(\p-\k) )(1+\rho(\k)) }   \right. \\&\left.-\sqrt{\rho(\p)\rho(\p-\k)\rho(\k)}\right)^2\frac{2\i\epsilon}{(\omega_\Bog(\k)-\omega_\Bog(\p-\k)-\omega_\Bog(\p))^2+\epsilon^2}\label{eq: symm3} \\ 
	&+ \sum_{\p,\k-\p\in\Xi^>_L} (V_{-\p,\k} + V_{\k,-\p})^2\left( \sqrt{(1+\rho(\p-\k))(1+\rho(\k))\rho(\p )} \right.\notag \\ 
	& \left. -\sqrt{\rho(\k)\rho(\p-\k)(1+\rho(\p)} \right)^2 \frac{2\i\epsilon}{(\omega_\Bog(\k)+\omega_\Bog(\p)-\omega_\Bog(\p-\k))^2+\epsilon^2}\label{eq: symm4} ,
	\end{align} where in the second line \eqref{eq: symm2} we have symmetrized the potential \begin{equation}
	V_{\p,\k-\p}^2 + V_{\p,\k-\p}V_{\k-\p,\p} \to \frac{1}{2}(V_{\k-\p,\p} + V_{\p,\k-\p})^2,
	\end{equation} exploiting the symmetry under $\p\leftrightarrow\k-\p$ of the other factors multiplying the effective potential and of the sums $\sum_{\p,\k-\p\in\Xi_L^>}$. The thermal factors in \eqref{eq: symm2},\eqref{eq: symm4} can be written in the following form \begin{align}
\left(\sqrt{(1+\rho(\p))(1+\rho(\k))(1+\rho(\p-\k))} - \sqrt{\rho(\p)\rho(\p-\k)\rho(\k)}\right)^2 \notag \\
= \rho(\p)\rho(\k)\rho(\p-\k)\left( 1- \e^{\beta\frac{\omega_\Bog(\k)+\omega_\Bog(\p)+\omega_\Bog(\p-\k)}{2}}\right)^2\label{eq: thermfactor1}\\
\left(\sqrt{(1+\rho(\p-\k))(1+\rho(\k))\rho(\p)} - \sqrt{(1+\rho(\p))\rho(\p-\k)\rho(\k)}\right)^2 \notag \\
= \rho(\p)\rho(\k)\rho(\p-\k)\left( e^{\beta\frac{\omega_\Bog(\p)}{2}}- \e^{\beta\frac{\omega_\Bog(\k)+\omega_\Bog(\p-\k)}{2}}\right)^2.\label{eq: thermfactor2}
	\end{align} Hence, the final expression before taking the thermodynamic limit $L\to+\infty$ reads as \begin{align}
-\frac{1}{2L^d}\sum_{\p,\p-\k\in\Xi^>_L} j(\k;\p,\p-\k)^2 \rho(\p)\rho(\p-\k)\rho(\k) \notag \\
\times \left( 1- \e^{\beta\frac{\omega_\Bog(\k)+\omega_\Bog(\p)+\omega_\Bog(\p-\k)}{2}}\right)^2\frac{2\i\epsilon}{(\omega_\Bog(\k)-\omega_\Bog(\p-\k)-\omega_\Bog(\p))^2+\epsilon^2} \\
-\frac{1}{L^d}\sum_{\p,\p-\k\in\Xi^>_L} j(\p-\k;\p,\k)^2\rho(\p)\rho(\p-\k)\rho(\k) \notag \\
\times \left( \e^{\beta\frac{\omega_\Bog(\p)}{2}}- \e^{\beta\frac{\omega_\Bog(\k)+\omega_\Bog(\p-\k)}{2}}\right)^2 \frac{2\i\epsilon}{(\omega_\Bog(\k)+\omega_\Bog(\p)-\omega_\Bog(\p-\k))^2+\epsilon^2}.
	\end{align} Since all the functions involved can be extended
        to continuous functions of the momenta $\p\in\mathbb{R}^d$, we
        can compute the limit $L\to +\infty$ by
        replacing $$\frac{1}{L^d}\sum_{\p,\p-\k\in\Xi^>_L}\to
        \int\frac{\d \p}{(2\pi)^d}.$$ The limit  $\epsilon\searrow 0$
        can then be evaluated using the limit of approximate delta functions
        $$ \lim_{\epsilon\searrow0}\frac{\epsilon}{x^2 +\epsilon^2} =
        \pi  \delta(x).$$
        The well-definedness and finiteness of the integrals in \eqref{eq: int1} and \eqref{eq: int2} under our hypotheses \eqref{eq: continuity of Fourier}--\eqref{eq: no plateau} will be discussed in the following section (See  Theorem \ref{thm: beliaev damping} and Theorem \ref{thm: Landau damping}). This concludes the proof.
\end{proof} 

The computation of $\delta(\k)$ in \eqref{formi1} is analogous to that of $\xi(\k)$, the only major difference being the absence of the $\le \leftrightarrow \re$ symmetry for the vector $b_\le^*(\k)\Omega$. This lack of symmetry produces additional terms contributing to $\Re\,\delta(\k)$. 

\begin{prop}
Fix  $\nu_0\in ]0,+\infty[$ and suppose the potential satisfies Assumptions \eqref{eq: continuity of Fourier}--\eqref{eq: no plateau} for all $\nu\in ]0, \nu_0]$. Then, we have  for the imaginary correction \begin{align}
	\i\Im{\delta(\k)} = \frac{1}{2}\xi(\k),\label{eq: immdelta}
\end{align} This contribution is only due to the Fermi Golden Rule term. The real correction is the sum of two contributions. One coming from the Fermi Golden Rule: \begin{align}
&-\frac{1}{2}\mathrm{PV}\int \frac{\d \p}{(2\pi)^d}\frac{j(\k; \p, \p-\k)^2\rho(\p)\rho(\p-\k)\rho(\k)}{\omega_\Bog(\p-\k)+\omega_\Bog(\p)-\omega_\Bog(\k)}\notag \\ 
&\times \left[ \left( 1- \e^{\beta\frac{\omega_\Bog(\k)+\omega_\Bog(\p)+\omega_\Bog(\p-\k)}{2}}\right)^2 -\left(\e^{\beta\frac{\omega_\Bog(\k)}{2}}- \e^{\beta \frac{\omega_\Bog(\p)+\omega_\Bog(\p-\k)}{2}}\right)^2  \right] \label{eq: redelta1}\\ 
&-\mathrm{PV}\int \frac{\d \p}{(2\pi)^d}\frac{j(\p-\k; \p,\k)^2\rho(\p)\rho(\p-\k)\rho(\k)}{\omega_\Bog(\p-\k)-\omega_\Bog(\p)-\omega_\Bog(\k)}\notag \\ 
&\times \left[\left( \e^{\beta\frac{\omega_\Bog(\p)}{2}}- \e^{\beta\frac{\omega_\Bog(\k)+\omega_\Bog(\p-\k)}{2}}\right)^2 -\left( \e^{\beta\frac{\omega_\Bog(\p-\k)}{2}}- \e^{\beta\frac{\omega_\Bog(\p)+\omega_\Bog(\k)}{2}}\right)^2 \right]\label{eq: redelta2}\\
&-\frac{1}{2}\int \frac{\d \p}{(2\pi)^d}\frac{\kappa(\k, \p,\p-\k)^2\rho(\p)\rho(\p-\k)\rho(\k)}{\omega_\Bog(\k)+\omega_\Bog(\p)+\omega_\Bog(\p-\k)}\notag \\ 
&\times \left[ \left( 1- \e^{\beta\frac{\omega_\Bog(\k)+\omega_\Bog(\p)+\omega_\Bog(\p-\k)}{2}}\right)^2 -\left(\e^{\beta\frac{\omega_\Bog(\k)}{2}}- \e^{\beta \frac{\omega_\Bog(\p)+\omega_\Bog(\p-\k)}{2}}\right)^2 \right],\label{eq: redelta3}
\end{align}where $j(\k;\p,\q)$ and $\kappa(\k,\p,\q)$ have been defined in Equations \eqref{eq: jj} and \eqref{eq: kappa}, and one from the Hellmann-Feynman term: \begin{align}
&\int \frac{\d \p}{(2\pi)^d}  \left[ \hat{v}(\0)(c_\k^2+s_\k^2)s_\p^2 + \hat{v}(\k-\p)\left((c_\k^2+s_\k^2)s_\p^2 +2c_\k s_\k c_\p s_\p\right)\right] \label{eq: int zeroHF}\\
&+ \int \frac{\d \p}{(2\pi)^d}(\hat{v}(\0)+ \hat{v}(\k-\p))(s_\k^2+c_\k^2)(s_\p^2+c_\p^2)\rho(\p)  \notag \\
&+\int\frac{\d \p}{(2\pi)^d} 4\hat{v}(\k-\p)c_\k s_\k c_\p s_\p \rho(\p).\label{eq: int thermaHF}
\end{align}

\end{prop}
\begin{proof}
For the Fermi Golden Rule, the only nonzero contribution to $\delta(\k)$ comes from vectors of the following form \begin{align}
&\label{eq: vectorL1}	[b_\le^*(\k), b_\le^*(\p+\q)b_\re^*(\p)b_\le(\p)]\Omega =  -\delta_{\k,\p} b_\le^*(\q+\k)b_\re^*(\p)\Omega,\\
&\label{eq: vectorL2}	[b_\le^*(\k), b_\le(\q+\p)b_\le^*(\p)b_\le^*(\q)]\Omega = -\delta_{\q+\p,\k}b_\le^*(\k-\p)b_\le^*(\p)\Omega , \\
&\label{eq: vectorL3}	[b_\le^*(\k), b_\re^*(\q+\p)b_\re^*(-\p)b_\le(-\q) ]\Omega =  -\delta_{-\q,\k} b_\re^*(\p -\k)b_\re^*(-\p)\Omega,\\
&\label{eq: vectorL4}	[b_\le(\k), b_\le^*(\p+\q)b_\re^*(\p)b_\re^*(\p)]\Omega = \delta_{\k,\p+\q} b_\re^*(\p)b_\re^*(\k-\p) \Omega ,\\
&\label{eq: vectorL5}	[b_\le(\k), b_\re^*(\q+\p)b_\le^*(\p)b_\le^*(\q)]\Omega =  \delta_{\k,\p}b_\re^*(\q+\k)b_\le^*(\q)\Omega \notag\\ &+\delta_{\k,\q}b_\re^*(\p+\k)b_\le^*(\p)\Omega ,\\
&\label{eq: vectorL6}	[b_\le(\k), b_\le^*(\q+\p)b_\le^*(-\p)b_\le^*(-\q) ]\Omega =  \delta_{\k,-\p}b_\le^*(\q-\k)b_\le^*(-\q)\Omega \notag \\ &+\delta_{\k,-\q}b_\le^*(\p-\k)b_\le^*(-\p)\Omega +\delta_{\q+\p,\k}b_\le^*(\p-\k)b_\le^*(-\p)\Omega.
\end{align} From Equations \eqref{eq: vectorL1}---\eqref{eq: vectorL6} we can deduce that \begin{equation}
[b_\le(\k),[b_\le^*(\k),L_3^L]]\Omega = 0.
\end{equation} Thus, we only need to compute \begin{equation}
[b_\le^*(\k),L_3^L]\Omega, \qquad[b_\le(\k),L_3^L]\Omega.
\end{equation} Writing $L_3^L$ as $L_{3,1}^L+ L^L_{3,0}$ we obtain 
\begin{align}
&[L^L_{3,1},b_{\le}^*(\k)]\Omega \notag \\ =& \sum_{\p,\p-\k\in\Xi^>_L} V_{\p,\k-\p} 
\left(\sqrt{(1+\rho(\p))(1+\rho(\p-\k) )(1+\rho(\k)) }  \notag \right. \\&\left.-\sqrt{\rho(\p)\rho(\p-\k)\rho(\k)}\right)b_\le^*(\p)b_\le^*(\k-\p)\Omega \label{eq: beliaev2}  \\
& + \sum_{\p,\p-\k\in\Xi^>_L}(V_{\k,-\p} + V_{-\p,\k}) \left(\right. \sqrt{(1+\rho(\p- \k))(1+\rho(\k))\rho(\p )} \notag \\ 
& \left. -\sqrt{\rho(\k)\rho(\p-\k)(1+\rho(\p)} \right.) b_\le^*(\k-\p) b_\re^*(-\p)\Omega  \label{eq: landau2},
\end{align}  and
\begin{align}
	&[L^L_{3,0},b_{\le}^*(\k)]\Omega \notag \\ =& \sum_{\p,\p-\k\in\Xi^>_L} \left(U_{\k,-\p} + U_{-\p,\k} + U_{\p,\k-\p}\right) 	\left(\sqrt{(1+\rho(\k))\rho(\p)\rho(\p-\k) }  \notag \right. \\&\left.-\sqrt{(1+\rho(\p))(1+\rho(\p-\k))\rho(\k)}\right)b_\re^*(\p-\k)b_\re^*(-\p)\Omega \label{eq: null3},
\end{align} where again we have exploited the reflection symmetry of $\rho,\; V$ and $U$. An analogous computation reveals that the vectors $[b_\le(\k),L^L_{3,1}]\Omega$ and $[b_\le(\k),L^L_{3,0}]\Omega$ can be obtained respectively from $[b_\le^*(\k),L^L_{3,1}]\Omega$ and $[b_\le(\k)^*,L^L_{3,0}]\Omega$ after exchanging the left and right labels $\le\leftrightarrow\re$ and  permuting the thermal factors
\begin{align}
	&[b_{\le}(\k),L^L_{3,1}]\Omega \notag \\ =& \sum_{\p,\p-\k\in\Xi^>_L} V_{\p,\k-\p} 
	\left(\sqrt{(1+\rho(\k))\rho(\p)\rho(\p-\k) }  \notag \right. \\&\left.-\sqrt{(1+\rho(\p))(1+\rho(\p-\k))\rho(\k)}\right)b_\re^*(\p)b_\re^*(\k-\p)\Omega \label{eq: beliaev3}  \\
	& + \sum_{\p,\p-\k\in\Xi^>_L}(V_{\k,-\p} + V_{-\p,\k}) \left(\right. \sqrt{(1+\rho(\p))(1+\rho(\k))\rho(\p -\k)} \notag \\ 
	& \left. -\sqrt{\rho(\k)\rho(\p)(1+\rho(\p-\k))} \right.) b_\re^*(\k-\p) b_\le^*(-\p)\Omega  \label{eq: landau4},
\end{align}

\begin{align}
		&[b_{\le}(\k),L^L_{3,0}]\Omega \notag \\ =& \sum_{\p,\p-\k\in\Xi^>_L} \left(U_{\k,-\p} + U_{-\p,\k} + U_{\p,\k-\p}\right) 	\left(\sqrt{(1+\rho(\k))(1+\rho(\p))(1+\rho(\p-\k)) }  \notag \right. \\&\left.-\sqrt{\rho(\k)\rho(\p-\k)\rho(\p)}\right)b_\le^*(\p-\k)b_\le^*(-\p)\Omega \label{eq: null4}.
\end{align}

From \eqref{eq: beliaev2}--\eqref{eq: null4} we see immediately that \begin{align}
	\left([b_\le^*(\k), L^L_{3,1}]\Omega\right| (L_\Bog - \omega_\Bog(\k) -\i\epsilon)^{-1}\left. [b_\le^*(\k), L^L_{3,0}]\Omega\right) \\= 	\left([b_\le(\k), L^L_{3,1}]\Omega\right| (L_\Bog + \omega_\Bog(\k) -\i\epsilon)^{-1}\left. [b_\le(\k), L^L_{3,0}]\Omega\right) = 0.
\end{align} Writing the thermal factor as in \eqref{eq: thermfactor1},\eqref{eq: thermfactor2} the final expression  before taking the thermodynamic limit will be given by \begin{align}
&-\frac{1}{2L^d}\sum_{\p,\p-\k\in\Xi^>_L} \frac{j(\k;\p,\p-\k)^2 \rho(\p)\rho(\p-\k)\rho(\k)}{\omega_\Bog(\p-\k)+\omega_\Bog(\p)-\omega_\Bog(\k)-\i\epsilon} \notag \\
&\times\left[ \left( 1- \e^{\beta\frac{\omega_\Bog(\k)+\omega_\Bog(\p)+\omega_\Bog(\p-\k)}{2}}\right)^2 -\left(\e^{\beta\frac{\omega_\Bog(\k)}{2}}- \e^{\beta \frac{\omega_\Bog(\p)+\omega_\Bog(\p-\k)}{2}}\right)^2  \right]  \label{eq: final L31 beliaev}\\
&-\frac{1}{L^d}\sum_{\p,\p-\k\in\Xi^>_L} \frac{j(\p-\k;\p,\k)^2\rho(\p)\rho(\p-\k)\rho(\k)}{\omega_\Bog(\p-\k)-\omega_\Bog(\k)-\omega_\Bog(\p)-\i\epsilon} \notag \\
&\times \left[\left( \e^{\beta\frac{\omega_\Bog(\p)}{2}}- \e^{\beta\frac{\omega_\Bog(\k)+\omega_\Bog(\p-\k)}{2}}\right)^2 -\left( \e^{\beta\frac{\omega_\Bog(\p-\k)}{2}}- \e^{\beta\frac{\omega_\Bog(\p)+\omega_\Bog(\k)}{2}}\right)^2 \right]\label{eq: final L31 landau} \\
&-\frac{1}{2L^d}\sum_{\p,\p-\k\in\Xi^>_L}\frac{\kappa(\k,\p,\p-\k)^2\rho(\p)\rho(\k)\rho(\p-\k)}{\omega_\Bog(\k)+\omega_\Bog(\p)+\omega_\Bog(\p-\k)-\i\epsilon}\notag\\ 
&\times\left[ \left( 1- \e^{\beta\frac{\omega_\Bog(\k)+\omega_\Bog(\p)+\omega_\Bog(\p-\k)}{2}}\right)^2 -\left(\e^{\beta\frac{\omega_\Bog(\k)}{2}}- \e^{\beta \frac{\omega_\Bog(\p)+\omega_\Bog(\p-\k)}{2}}\right)^2 \right].\label{eq: final L32}
\end{align} As $L\to+\infty$, the summation in Equations \eqref{eq: final L31 beliaev},\eqref{eq: final L31 landau} and \eqref{eq: final L32} will converge to the respective integrals. In the limit $\epsilon\searrow 0$ we use the \textit{Sochocki–Plemelj} distributional identity on the real line \begin{equation*}
\lim_{\epsilon\searrow 0} \frac{1}{x -\i\epsilon} = \pi\i \delta(x) + \mathcal{P}\frac{1}{x}.
\end{equation*} The resulting imaginary component of \eqref{eq: final L32} is $0$, as $\delta(\omega_\Bog(\k)+\omega_\Bog(\p)+\omega_\Bog(\p-\k))=0$. In Equation \eqref{eq: final L31 beliaev}, the thermal factor $$ \left( \e^{\beta\frac{\omega_\Bog(\k)}{2}}- \e^{\beta \frac{\omega_\Bog(\p)+\omega_\Bog(\p-\k)}{2}}\right)^2$$
is in the kernel of the Dirac delta distribution $\delta(\omega_\Bog(\k)-\omega_\Bog(\p-\k)-\omega_\Bog(\p))$ and thus, it will not contribute to the imaginary part. The same holds for the rightmost thermal factor inside the square brackets of Equation \eqref{eq: final L31 landau}. Hence, we see that the resulting contributions give exactly \eqref{eq: immdelta} for the imaginary part, and \eqref{eq: redelta1},\eqref{eq: redelta2},\eqref{eq: redelta3} for the real component. 

For the Hellmann-Feynman correction, we have to compute the vector \begin{equation}
 [b_\le(\k),[L_4^L, b_\le^*(\k)]]\Omega
\end{equation} and evaluate its projection on $\Omega$. The only non zero contributions come from \begin{align}\label{eq: HF1}
&\left( \Omega|[b_\le(\k),[b_{\beta,\le}^*(\k'+\p)b_{\beta,\le}^*(\q-\k')b_{\beta,\le}(\q)b_{\beta,\le}(\p) -(\le\leftrightarrow\re),b_\le^*(\k)]]\Omega\right)\notag \\
=& \rho(\p)(\delta_{\q,\k}\delta_{\k',\0} +\delta_{\q,\k}\delta_{\k'+\p,\k})  + \rho(\q)(\delta_{\p,\k}\delta_{\k',\0}  +\delta_{\p,\k}\delta_{\q,\k+\k'}),
\end{align} and \begin{align}\label{eq: HF2}
&\left( \Omega|[b_\le(\k),[b_{\beta,\le}^*(\p)b_{\beta,\le}(\p) -(\le\leftrightarrow\re),b_\le^*(\k)]]\Omega\right)
= \delta_{\k,\p}.
\end{align} This shows that to the first non-trivial order only $L^L_{2,1}$ and $L^L_{4,2}$ contribute to the Hellman-Feynmann energy correction. Using \eqref{eq: HF1} and \eqref{eq: HF2} and the full expression for $L^L_{2,1}+L^{L}_{4,2}$, Equations \eqref{eq: L41},\eqref{eq: L42}, we obtain \begin{align}
&\left(\Omega | [b_\le(\k),[L_4^L, b_\le^*(\k)]\Omega\right) = \notag \\
=&\frac{1}{L^d}\sum_{\p\in\Xi^>_L} \left[ \hat{v}(\0)(c_\k^2+s_\k^2)s_\p^2 + \hat{v}(\k-\p)\left((c_\k^2+s_\k^2)s_\p^2 +2c_\k s_\k c_\p s_\p\right)\right] \label{eq: zeroHF}\\
&+ \frac{1}{L^d}\sum_{\p\in\Xi^>_L}(\hat{v}(\0)+ \hat{v}(\k-\p))(s_\k^2+c_\k^2)(s_\p^2+c_\p^2)\rho(\p) \notag \\
&+ \frac{4}{L^d}\sum_{\p\in\Xi^>_L} \hat{v}(\k-\p)c_\k s_\k c_\p s_\p \rho(\p).\label{eq: thermaHF}
\end{align} Taking the limit $L\to+\infty$ Equations \eqref{eq: zeroHF} and \eqref{eq: thermaHF}  converge to the corresponding integral over $\p\in\mathbb{R}^d$. We note that \eqref{eq: zeroHF} does not depend on the temperature, so that this term survives the zero temperature limit $\frac{1}{\beta}\to 0$, while \eqref{eq: thermaHF} is proportional to $\rho(\p)$, which goes to zero as $\e^{-\beta\omega_\Bog(\p)}$ in the aforementioned limit. We also remark how the thermal factor $\rho(\p)$ makes \eqref{eq: thermaHF} convergent, while \eqref{eq: zeroHF} necessitates a cut-off over momenta, which is naturally introduced by the Fourier transform of the potential, \textit{cf.} Assumption \eqref{eq: cutoff}. 
\end{proof}

\section{Computations of decay  rates} \label{sec: decayrates}

In this section we fix the spatial dimension $d=3$. We will compute the integrals \eqref{gamma_B}, \eqref{gamma_L} in the small momenta and low temperature limit, i.e.  $\frac{|\k|}{\sqrt{\nu}}, \frac{1}{\beta
	\nu}\rightarrow 0$. The different limits of the damping will depend also on the ratio between the temperature and the momentum $\beta\sqrt{\nu}|\k|$.

We will obtain integrals which can be computed explicitely in terms of \textit{polylogarithm functions} $\operatorname{Li}_n$ \cite{LW_81} of positive  integer order $n$. The latter are defined for $|z|< 1$ by the power series \begin{equation}\label{eq: def on poly}
	\operatorname{Li}_n(z):= \sum_{k=1}^{+\infty}\frac{z^k}{k^n},
\end{equation} which can be extended by analytic continuation with the principal branch having a cut along $[1,+\infty[$ and agreeing with Equation \eqref{eq: def on poly} when $n\geq 2$ and $|z|\leq 1$. For $n\geq 3$, we have the following special values of the polylogarithm: \begin{equation}
	\operatorname{Li}_n(1) = \zeta(n),\quad \left(\frac{\d}{\d z}\operatorname{Li}_n\right)(1) = \zeta(n-1),
\end{equation} where $\zeta(n)$ is the \textit{Riemann zeta function} evaluated at $n$. 

Our results are collected in two theorems dedicated to the analysis of Beliaev and Landau damping's integrals. In the following. we will often consider, without any further mention of this fact, functions of a generic momentum  vector $\k$ which only depends on its euclidean norm as functions of the one dimensional variable $k:= |\k|$.

\subsection{Assumptions}
\label{Assumptions}
Let us describe in detail the assumptions on the potential.
We fix the dimension  $d=3$. Some of our hypotheses involve the effective chemical potential $\nu>0$. These hypotheses are to be understood in the following way: fix $\nu_0\in ]0,+\infty[$, then the potential satisfies the assumptions for  every $\nu \in ] 0, \nu_0] $. 

The first set of hypotheses was stated and
their meaning was discussed in Subsect. \ref{The ``first principles Hamiltonian''}:
\begin{subequations}\label{eq: assumptions on potential1.}
	\begin{align}
		&\label{eq: continuity of Fourier.} v\in L^1(\mathbb{R}^d,\d \x),\; v(\x)\in\mathbb{R};\\
		&\label{eq: cutoff.} \hat{v}\in L^2(\mathbb{R}^d,\d \k);\\
		&\label{eq: positivity of Fourier.} \hat{v}(\0)>0,\; \hat{v}(\k) > -\hat{v}(\0)\frac{|\k|^2}{4\nu}; \qquad \nu\in ]0, \nu_0];\\
		&v(\x) = v(-\x),\quad \text{which implies $\hat{v}(\k) = \hat{v}(-\k)$}.\label{eq: symmetry requirement.}
	\end{align}
      \end{subequations}

To carry out the computation of the Beliaev damping in Theorem \ref{thm: beliaev damping.} we will require in addition the rotational invariance of  $v$. Moreover, we will need additional regularity assumptions for  $\hat{v}(\k)$ close to $\k=\0$, both for making sense of the Dirac delta distribution $\delta(\omega_\Bog(\k)-\omega_\Bog(\p)-\omega_\Bog(\p-\k))$ that requires a differentiability hypothesis for small momenta, and for the estimates of the remainder terms in Equations \eqref{eq: reduced expression for Beliaev.}, \eqref{eq: Beliaev for high temperatures.}. The additional hypotheses are the following,

\begin{subequations}\label{eq: assumptions on potential2.}
	 \begin{align}
		& v \text{ is rotationally invariant}; \label{eq: rot invariace} \\
		&\hat{v} \text{ is } C^5 \text{ in a neighborhood of } \k = \0;\label{eq: regularity at 0}\\
		& \frac{\nu}{\hat{v}(0)}\frac{\d^2\hat{v}}{\d k^2}(0) > -\frac{1}{2},\qquad \nu\in ]0, \nu_0] .\label{eq: convexity for v}
	\end{align} 
\end{subequations}
  Hypotheses \eqref{eq: rot invariace} and \eqref{eq: regularity at 0} ensure that we can refard its Fourier transform as a function of $k:= |\k|$ and write it in the following form\begin{equation*}
\hat{v}(k) = \hat{v}(0)(1+ r(k)),\qquad r(k) = O(k^2) \text{ as }k\to 0.
\end{equation*} Hypothesis \eqref{eq: convexity for v} is sufficient to guarantee the existence of an interval $[0,K]$, with $K>0$, such that $p\to \omega_\Bog(p)$ is strictly convex for $p\in [0,K]$.

To evaluate the Landau damping as in Theorem \ref{thm: Landau damping.}, we need to make sense of the Dirac delta distribution $\delta( 
\omega_\Bog(\p-\k)-\omega_\Bog(\p)-\omega_\Bog(\k))$. The support of this  distribution could contain regions of arbitrarily high momenta. Thus, we will need to control the derivatives of $\omega_\Bog(\p),\;\omega_\Bog(\p-\k)$ for all values of $\p\in\mathbb{R}^d$. For this reason, in addition to \eqref{eq: rot invariace} and \eqref{eq: convexity for v}, we need to require
\begin{subequations}\label{eq: Assumption3}
 \begin{align}
&	\hat{v} \text{ is }C^1 \label{eq: regularity everywhere}; \\
&\frac{k^2}{2\nu} + \frac{\hat{v}(k)}{\hat{v}(0)} + \frac{k}{2\hat{v}(0)}\frac{\d\hat{v}(k)}{\d k}= 0\text{ has only a finite number of solutions}\notag \\
& \text{and }\liminf_{k\to+\infty}\left(\frac{k^2}{2\nu} + \frac{\hat{v}(k)}{\hat{v}(0)} + \frac{k}{2\hat{v}(0)}\frac{\d\hat{v}(k)}{\d k}\right) > 0,\qquad \nu\in ]0, \nu_0]  \label{eq: no plateau}.
\end{align}
\end{subequations}

We remark that these hypotheses allow for $\omega_\Bog(k)$ to have a finite number of stationary points. Such stationary points typically arise in realistic dispersion relations, for instance in superfluid $^4\mathrm{He}$ \cite{VXRAK_2002, GBSKDFO_2021, BS_2018, DOSK_2002, HM_65}, where the excitation spectrum develops local maxima and minima. The quasiparticles associated with these extrema are known as \emph{maxons} and \emph{rotons}, respectively. As we will see, the contribution from second order perturbation theory is not influenced by the presence of such excitations in the small momentum regime.

\begin{example}
	To show that all Assumptions \eqref{eq: continuity of Fourier}--\eqref{eq: no plateau} can be satisfied we present here an example of a potential satisfying them all: \begin{equation}
	\hat{v}(\k) := v e^{-\frac{|\k|^2}{8\nu_0}},\qquad 1>\frac{2v}{\nu} >0.
	\end{equation} Hypotheses \eqref{eq: continuity of Fourier} through \eqref{eq: regularity at 0} are easy to verify. The last one can be proved by noting that \begin{equation}
	\inf_{k>0}\left( \frac{k^2}{2\nu} +e^{-\frac{k^2}{8\nu_0}}\left(1-\frac{k^2}{8\nu_0}\right)\right) \geq 1.
	\end{equation}
\end{example}

\subsection{Changes of variables}

We begin by proving the following lemma \begin{lemma}\label{lemma: regularity of the Bogoliubov lemma}
Fix  $\nu_0\in ]0,+\infty[$ and suppose the potential satisfies \eqref{eq: continuity of Fourier}--\eqref{eq: symmetry requirement} and \eqref{eq: rot invariace}--\eqref{eq: convexity for v} for all $\nu\in ]0, \nu_0]$. Then, there exists a sufficiently small neighborhood $U$ of 0 such that $k\to \omega_\Bog(k)$ is invertible and its inverse, $\omega\to p(\omega)$, is $C^5$.  In particular, $U\ni \omega\to \hat{v}(p(\omega))$ is in $C^5(U)$. 
\end{lemma}
\begin{proof}
	The Fourier transform of the potential $\k\to \hat{v}(\k)$ can be regarded as a one-dimensional function of $k:= |\k|$, which is $C^5$ in the latter variable in some neighborhood of $0$.  Thus, we have that $$\omega_\Bog(k) = \sqrt{\frac{k^4}{4} +  \frac{\nu\hat{v}(k)}{\hat{v}(0)}k^2}$$
	is again a $C^5$ function of the variable $k$ for some $K>0$ and $0<k<K$. Then, we can consider the  function of two variables \begin{equation}
		\mathbb{R}^2\ni(\omega,k) \to g(\omega, k) := \omega- k\sqrt{\frac{k^2}{4} + \frac{\nu\hat{v}(k)}{\hat{v}(0)}}.
	\end{equation} It is immediate to verify that \begin{equation}
		g(0,0) = 0,\qquad (\partial_kg)(0,0) = -\sqrt{\nu}\neq 0,
	\end{equation} so that by the implicit function theorem we can find a neighborhood $U$ of $\omega = 0$ and a (unique) $C^5$ function $U\ni \omega\to p(\omega)$ such that $g(\omega, p(\omega))=0$. Hence, for $\omega_\Bog(\k)\in U$, $\hat{v}(\k)$ can be implicitly represented as a $C^5$ function of the Bogoliubov energy $\omega_\Bog(\k)$. 
\end{proof} 
The previous lemma allows us to use as integration variables, at least for sufficiently small momenta, the Bogoliubov energies $u := \omega_\Bog(\p),\text{ and } w:=\omega_\Bog(\q)$. For arbitrary values of the momenta, the relation $k\to\omega_\Bog(k)$ need not be globally invertible. However, if $\hat{v}$ is at least $C^1$,  one has\begin{equation}
	\frac{\d \omega_\Bog(k)}{\d k} =\frac{k^3}{2\omega_\Bog(k)} \left( 1+ \frac{2\nu\hat{v}(k)}{\hat{v}(0)k^2} + \frac{\nu}{\hat{v}(0)k}\frac{\d \hat{v}(k)}{\d k} \right),\qquad k >0,
\end{equation} which by Assumptions \eqref{eq: regularity everywhere}--\eqref{eq: no plateau} can only be zero at a finite  number of points and is positive for sufficiently large $k$. This implies that there are always a finite number of intervals where $k\to \omega_\Bog(k)$ can be inverted.

For completeness, we illustrate a list of relevant changes of variables that one can perform to simplify the integrals. A similar presentation can be found in \cite{DLN_2023}. First of all, it is convenient to employ spherical coordinates to integrate out the azimutal coordinate. This is done by  choosing the $z$-axis  parallel to $\k$. We can represent this change of variables as $\p = (p\sin\theta\cos\phi, p\sin\theta\sin\phi, p\cos\theta)$, so that the integration measure is changed as \begin{equation}
	\d \p  = p^2 \d p \d (\cos\theta)\d \phi.
\end{equation} Due to the symmetry of the problem we can always perform the $\phi$  integration, which results in a global factor of $2\pi$. Now, we can move to more natural coordinates for the problem given by $p:= |\p|$ and $q:=|\p-\k|$. The Jacobian of this transformation is computed to be \begin{equation}
	p^2 \d p \d (\cos\theta) = \frac{pq}{k}\d p \d q
\end{equation} and the integration range can be read from the constraints \begin{align}
	&|p-q|\leq k,\\ & k\leq p+q,
\end{align} that follow from the triangle inequality. 

If we now restrict to intervals such that both $p\to\omega_\Bog(p)$ and $q\to \omega_\Bog(q)$ are differentiable and can be inverted, we can use the variables $u := \omega_\Bog(p) $ and $w := \omega_\Bog(q)$. After this, the measure changes as follows \begin{equation}\label{eq: change in u e w}
	\frac{pq}{k}\d p \d q = \frac{uw}{k}|f(u)f(w)| \d u \d w,
\end{equation} with \begin{align}
	f(u) &:= \frac{\d p(u)^2}{\d( u^2)}\\ &= \left(\frac{\nu\hat{v}(p(u))}{\hat{v}(0)} +\left(\frac{1}{2}+\frac{\nu}{\hat{v}(0)} \frac{\d \hat{v}}{\d (p^2)}(p(u))\right) p(u)^2 \right)^{-1} 
\end{align} where $u\to p(u)$ is the inverse of $p\to\omega_\Bog(p)$. Note that by Assumptions \eqref{eq: positivity of Fourier}, \eqref{eq: regularity at 0} $f$ is always a bounded function for $u$ in a neighborhood of $0$. Furthermore, if we assume \eqref{eq: regularity everywhere} -- \eqref{eq: no plateau} we see that $f$ is well defined and bounded for sufficiently large $u$.

 Another useful change of variables that we will employ is
\begin{align}\label{eq: change in x and y 1}
x &= u + w, \quad y = u - w,\\
\label{eq: change in x and y 2}
u &= \frac{x + y}{2}, \quad w = \frac{x - y}{2},
\end{align}
which transforms the measure as
\begin{equation}
\frac{uw}{k}| f(u) f(w)| \d u \d w
= \frac{x^2 - y^2}{8k}\left| f\left(\frac{x + y}{2}\right)
f\left(\frac{x - y}{2}\right)\right| \d x \d y.
\end{equation}
In addition, it will be convenient to introduce the following rescalings and shifts, which allow us to factor out the dependence on the physical parameters from the integrals:
\begin{equation}
t := \frac{y}{\omega_{\Bog}(k)},
\end{equation}
for the Beliaev damping, and
\begin{equation}
t := \beta \nu \frac{x - \omega_{\Bog}(k)}{2\nu},
\end{equation}
for the Landau damping.

\subsection{Estimates for remainders}
In this subsection we will prove some preliminary results that will help us characterize the asymptotic properties of the function $j(\k;\p,\q)$ appearing in the Beliaev and Landau integrals. We begin by proving a lemma that will allow us to carefully estimate the remainder terms arising in some of the subsequent expansions. \begin{lemma}\label{lemma: lagrange}
	Consider a function of two variables $U\ni(u,w)\to R(u,w)$ defined in a neighborhood $U$ of $(0,0)$. Suppose that \begin{enumerate}
		\item [i)] $R\in C^3(U)$.
		\item [ii)] $R(0,w) = R(u,0)= 0$.
	\end{enumerate} Then, the following holds \begin{equation}
		R(u,w) = uw \int^{1}_0\int^1_0 \d s_1 \d s_2  (\partial_u\partial_w R)(s_1 u, s_2 w)
	\end{equation}
\end{lemma}
\begin{proof}
	The proof is two applications of the fundamental theorem of calculus.  We have \begin{align}
		R(u,w) &= R(u,w) - R(0,w) \\
		& = u\int^1_0 \d s_1(\partial_u R)(s_1u,w) \label{eq: nice second line} \\
		&\label{eq: nice properties} = uw \int^1_0 \int^1_0 \d s_1 \d s_2(\partial_{u}\partial_{w}R)(s_1 u, s_2w),
	\end{align} where in the last line \eqref{eq: nice properties} we have used that
	\begin{equation} 0=R(u,0) = u\int^1_{0}\d s_1 (\partial_u R)(s_1 u, 0),\end{equation} by the equality in the second line \eqref{eq: nice second line} and the second hypothesis on $R$.
      \end{proof} Another result we will use is the following

 \begin{lemma}\label{lemma: bootstrap}
	Take integers $l,m\geq 0$ and $n\geq l+m+1$. Consider a $C^{n-l-m}$ function $(u,w)\to f(u,w)$ satisfying \begin{equation}
		u^lw^m f(u,w) = O\left( (u^2+w^2)^{\frac{n}{2}}\right),\qquad \text{as } (u,w)\to (0,0).\end{equation} Then \begin{equation}
		f(u,w) \in O\left( (u^2 +w^2)^{\frac{n-l-m}{2}}\right).
	\end{equation}
\end{lemma}
\begin{proof}
	We will prove the Lemma  by induction on $n$. If $f$ satisfies the hypotheses, then we have \begin{equation}
		|u^{l+m}f(u,u)| \leq C|u^n|,\qquad \text{as } u\to 0
	\end{equation} for some constant $C>0$. As $n> l+m$, this implies that $f(0,0) = 0$. Then, we can write $f$ as \begin{align}\label{eq: expansion}
		f(u,w) =&  \int^{1}_0 \d s \left[u(\partial_uf)(su,sw) +w(\partial_wf)(su,sw)  \right]
	\end{align} If $n= l+m+1$, the proof is complete. 
	We proceed now by induction. Let us denote $n-l-m =:k$ to ease the notation. Suppose that the statement has been proved for $n-1 \geq l+m+1$. Then, if $f$ satisfies the hypotheses of the Lemma with integer $n$, we get that  \begin{equation}
		u^lw^mf(u,w) \in O\left((u^2+w^2)^{\frac{n}{2}}\right) \subset O\left((u^2+w^2)^{\frac{n-1}{2}}\right). 
	\end{equation} Thus, by induction $f$ satisfies $f(u,w)\in O\left((u^2+w^2)^{\frac{k-1}{2}}\right)$. Hence, by the differentiability properties of $f$, we know that \begin{equation}
		(\partial^i_u\partial^j_wf)(0,0) = 0, \qquad \text{for all } i,j\in\mathbb{Z}_+,\; i+j \leq k-2.
	\end{equation}  Iterating the expansion in \eqref{eq: expansion} for the derivatives of $f$, we see that $f$ can be written as \begin{align}
		f(u,w) 
		= &\int^{1}_0\d s_1 \dots \int^1_0 \d s_{k-1}\notag \\
		&\times\sum_{i=0}^{k-1} \binom{k-1}{i}u^i w^{k-1-i}(\partial_u^i \partial^{k-1-i}_w f)(s_1\dots s_{k-1}u,s_1\dots s_{k-1}w).
	\end{align} Let us now take $w = \alpha u$, for some arbitrary $\alpha\in\mathbb{R}\backslash\{0\}$. we have that \begin{equation}
		|\alpha|^mf(u,\alpha u)| \leq C |u|^{k},
	\end{equation} with $k>0$. This last relation implies that \begin{equation}
		P(\alpha):= \sum_{i=0}^{k-1}  \alpha^{k-1-i}\binom{k-1}{i}(\partial_u^i \partial^{k-1-i}_w f)(0,0) = 0,
	\end{equation} for every $\alpha\in \mathbb{R}\backslash\{0\}$, and thus, by continuity, for every $\alpha\in\mathbb{R}$. Now, $\alpha\to P(\alpha)$ is equal to the zero polynomial in $\alpha$. By linear independence of the monomials in $\alpha$,  all of $P$'s coefficients must be zero, i.e. \begin{equation}
		(\partial_u^i \partial^{k-1-i}_w f)(0,0) = 0.
	\end{equation} recalling that $f\in C^{k} = C^{n-l-m}$, we can iterate once again \eqref{eq: expansion} to obtain
	\begin{align}
		f(u,w) 
		= &\int^{1}_0\d s_1 \dots \int^1_0 \d s_{k}\notag \\
		&\times\sum_{i=0}^{k} \binom{k}{i}u^i w^{k-i}(\partial_u^i \partial^{k-i}_w f)(s_1\dots s_{k}u,s_1\dots s_{k}w),
	\end{align}which concludes the proof.
\end{proof}
The previous lemma can be easily generalized to any finite number of variables. Let us now consider the function  $(\k,\p,\q)\to j(\k;\p,\q)$, which appears   in the Beliaev and Landau damping integrals \eqref{eq: int1} and \eqref{eq: int2}. The next proposition contains the main properties of some auxiliary functions which will appear inside the integrals after changing variables. We will denote by $u\to p(u)$ the inverse function of $p\to \omega_\Bog(p)$ in a sufficiently small neighborhood of $0$, see Lemma \ref{lemma: regularity of the Bogoliubov lemma}. \begin{prop}\label{prop: estimates for the remainder}Fix a $\nu_0 \in ] 0,+\infty[$ and  and suppose the potential satisfies \eqref{eq: continuity of Fourier}--\eqref{eq: symmetry requirement} and \eqref{eq: rot invariace}--\eqref{eq: convexity for v} for any $\nu\in]0, \nu_0]$.
	Consider the function \begin{equation}
		F(\omega;u,w) := \sqrt{8\omega uw}\sqrt{\frac{\nu}{\hat{v}(0)}}j(p(\omega);p(u),p(w)),
	\end{equation} which is well defined in a sufficiently small neighborhood of $(0,0,0)$. Then, $F$  satisfies the following properties\begin{enumerate}
		\item [i)] It is invariant under the exchange of the last two variables $F(\omega;u,w) = F(\omega;w,u)$; 
		\item [ii)]$F$ is at least $C^5$ in a sufficiently small neighborhood of $(0,0,0)$ and it satisfies \begin{align}\label{eq: null identities}
			F(\omega; 0 ,\omega) = F(\omega; \omega,0) = 0
		\end{align} for every sufficiently small  $\omega \geq 0$;
	\end{enumerate}  
	Moreover, $G(u,w):= F(u+w;u,w)^2$ satisfies the following expansion \begin{enumerate}
		\item[iii)]\begin{align}
			\frac{1}{\nu^5}G(u,w) &\label{eq: first form}= \frac{9}{\nu^6}u^2w^2(u+w)^2 +\frac{u^2w^2}{\nu^4}  S(u,w)
		\end{align} where $(u,w)\to S(u,w)$ is a continuous functions of order $O(\frac{u^4+w^4}{\nu^4})$ as $(u,w)\to (0,0)$.
	\end{enumerate}
\end{prop} 
\begin{proof}
	\textit{i)} We start by defining the following  auxiliary functions of the energy\begin{align}
		& \nu(\omega) :=\nu \frac{\hat{v}(p(\omega))}{\hat{v}(0)}\\
		&\label{eq: c regular}	c(\omega) := \sqrt{\sqrt{\omega^2+ \nu(\omega)^2 } + \omega}\\
		&\label{eq: s regular} s(\omega) :=  \sqrt{\sqrt{\omega^2+ \nu(\omega)^2 } - \omega}.
	\end{align} Equations \eqref{eq: c regular} and \eqref{eq: s regular} are just regularized versions of $c_\k$ and $s_\k$ obtained by multiplying by $\sqrt{2\omega_\Bog(\k)}$ and expressing everything in terms of the energies. Then, we have $\nu(0) = \nu$ and $c(0) = s(0)=\sqrt{\nu}$. We can write $F$ in terms of these as \begin{align}\label{eq: explicit expression of F}
		F(\omega; u, w) = &  \nu(\omega)(s(\omega)-c(\omega))(c(u)s(w) + c(w)s(u)) \notag\\ + & \nu(u)(c(u)-s(u))(c(\omega)c(w)+s(\omega)s(w))\notag \\
		+& \nu(w)(c(w)-s(w))(c(\omega)c(u)+s(\omega)s(u)).
	\end{align}
	From Equation \eqref{eq: explicit expression of F} we read immediately the symmetry of $F$ under $u\leftrightarrow w$. \textit{ii)} As a consequence of Lemma \ref{lemma: regularity of the Bogoliubov lemma}, $\omega\to\nu(\omega)$ and $\omega\to p(\omega)$ are $C^5$ functions in a neighborhood of $0$. This immediately implies that $s$ and $c$ are $C^5$ and thus, also $F$ is of the same class of differentiability close to the origin. To verify the identity \eqref{eq: null identities} it is sufficient to observe that $c(0)-s(0)= 0$ which implies $F(\omega; u,0) = -F(u;\omega,0)$.
	\textit{iii)} As $G$ is the square of $F$, it is of class at least $C^5$. Hence, we can take derivatives of $G$. Its first-order derivatives will be of the form $F$ times first order derivatives of $F$ and thus, we have by $ii)$\begin{equation}\label{eq: first order derivative}
		(\partial_{u}G)(u,0) = \partial_{u}G(0,w) = 	(\partial_{w}G)(u,0) = \partial_{w}G(0,w) = 0.
	\end{equation}  We proceed now to perform an  expansion of $G$ as $(u,w)\to (0,0)$. Let us write $\hat{v}$ as \begin{equation}
\hat{v}(p) = 	\hat{v}(0)(1+ r(p)),\qquad r(p) = O(p^2)\text{ as }p\to0. 
	\end{equation}Then, we will have \begin{equation}
		\nu(\omega)^2 = \nu^2+2\nu^2r(p(\omega)) + \nu^2r(p(\omega))^2.
	\end{equation} With a slight abuse of notation we will denote $r(\omega)\equiv r(p(\omega))$. Thus, we have the expansions \begin{align}
		&c(\omega)- s(\omega) = \sqrt{\nu}\left(\frac{\omega}{\nu} - \frac{\omega r(\omega)}{2\nu} -\frac{\omega^3}{8\nu^3} + R_1\left(\omega\right)\right),\\
		&c(u)s(w) + c(w)s(u) \notag\\ &= \nu\left( 2 + r(u)+r(w) +\frac{1}{4\nu^2}(u-w)^2 + R_2\left(u,w\right)\right),\\
		& c(u)c(w)+s(u)s(w) \notag \\&= \nu\left(2+r(u)+r(w) +\frac{1}{4\nu^2}(u+w)^2 + R_3\left(u,w\right)\right),
	\end{align} where $R_1$ and $R_2$, $R_3$ are suitable remainders of order $O\left(\frac{\omega^5}{\nu^5}\right)$, $O\left(\frac{u^4+w^4}{\nu^4}\right)$ respectively, obtained from the Taylor expansion in $\omega,u,w$ and $r$. Inserting these expansions in Equation \eqref{eq: explicit expression of F}, we can expand $F$ as \begin{align}\label{eq:first expansion for F}
		F(u+w;u,w) = \frac{3}{\sqrt{\nu}}uw(u+w) + \nu^{5/2}R\left(u,w\right)
	\end{align} where $R$ is some remainder of order $O\left(\frac{u^5+w^5}{\nu^{5}}\right)$.  We can also write $R$ in terms of $F$ as \begin{equation}
		R\left(u,w\right) = \frac{1}{\nu^{5/2}}\left( F(u+w;u,w) - \frac{3}{\sqrt{\nu}}uw(u+w)\right).
	\end{equation} Thanks to the properties of the function appearing on the right, we deduce that $R$ is $C^5$, symmetric under the exchange $u\leftrightarrow w$ and it satisfies  \begin{equation}\label{eq: first symmetry for R}
		R( u,0) = R( 0, w) = 0.
	\end{equation}  By  Lemma \ref{lemma: lagrange} we know that $R(u,w)$ can be written as \begin{equation}\label{integral for remainder}
		R(u,w) = uw\int^1_0\int^1_0 \d s_1 \d s_2 (\partial_u \partial_w R)(s_1 u, s_2 w) = O\left(\frac{u^5+w^5}{\nu^5}\right),
	\end{equation} and by Lemma \ref{lemma: bootstrap} we know that the integral function in \eqref{integral for remainder} satisfies \begin{equation}
		\int^1_0\int^1_0 \d s_1 \d s_2 (\partial_u \partial_w R)(s_1 u, s_2 w) = O\left(\frac{u^3+w^3}{\nu^5}\right)
	\end{equation}

	After taking the square of $F(u+w;u,w)$ we obtain \begin{align}
		\frac{1}{\nu^5}G(u,w) = & \frac{9}{\nu^6}u^2w^2(u+w)^2 \notag \\
		+& \label{eq: first remainder R}\frac{6}{\nu^3}uw(u+w)R(u,w) \\ +& \label{eq: second remainder R} R^2(u,w). 
	\end{align} 
	After writing $R$ as in \eqref{integral for remainder} we get that both the term in \eqref{eq: first remainder R} and in \eqref{eq: second remainder R} are of the correct form \eqref{eq: first form}. 
\end{proof}

\subsection{Beliaev damping rate}
We can now start with the computation of the Beliaev damping rate: \begin{thm}\label{thm: beliaev damping}
	Fix a $\nu_0 \in ] 0,+\infty[$ and  and suppose the potential satisfies \eqref{eq: continuity of Fourier}--\eqref{eq: symmetry requirement} and \eqref{eq: rot invariace}--\eqref{eq: convexity for v} for any $\nu\in]0, \nu_0]$.  Then: \begin{enumerate}
		\item [(a)] the Beliaev damping rate is computed to be \begin{align}\label{eq: full expressio for Beliaev}
			\gamma_\BB(\k;\beta,\nu) =&\frac{9\hat{v}(\0)\nu^{3/2}}{2048\pi}\frac{|\k|^4}{\nu^2}\frac{1}{\beta\nu}\mathcal{I}(\beta\omega_\Bog(\k))\left(1+ O\left(\frac{|\k|^2}{\nu}\right)\right)\quad \text{as }\frac{|\k|}{\sqrt{\nu}}\to 0,
		\end{align} where
		\begin{equation}
			\mathcal{I}(\theta) = \theta(1-\e^{-\theta} )\int^1_{-1} \d t \; \frac{(1-t^2)^2}{(1-\e^{-\theta\frac{1+t}{2}})(1-\e^{-\theta\frac{1-t}{2}})}
		\end{equation} is a continuous function of $\theta\in[0,+\infty)$ satisfying \begin{align*}
			& \mathcal{I}(\theta) = \frac{16}{3} + O(\theta) \quad \text{as }\theta \to 	0\\
			& \mathcal{I}(\theta) = \frac{16}{15}\theta + O\left(\frac{1}{\theta^2}\right) \quad \text{as }\theta\to +\infty
		\end{align*}
		\item [(b)] In the small temperature regime, $\frac{1}{\beta\sqrt{\nu}|\k|}\to 0$, the Beliaev damping rate can be extimated as \begin{align}\label{eq: reduced expression for Beliaev}
			\gamma_\BB(\k,\beta,\nu)=  \frac{3\hat{v}(\0)\nu^{3/2}}{640\pi}\frac{|\k|^5}{\nu^{5/2}}\left(1 + O\left( 
			\frac{1}{(\beta\sqrt{\nu}|\k|)^3} +  \frac{|\k|^2}{\nu}\right)\right) \notag \\ \text{as }\frac{|\k|}{\sqrt{\nu}},\; \;\frac{1}{\beta\sqrt{\nu}|\k|}\to 0 .
		\end{align}
		\item[(c)] In the opposite limit $\beta\sqrt{\nu}|\k|\to 0$,  the Beliaev damping rate can be extimated as	\begin{multline}\label{eq: Beliaev for high temperatures}
			\gamma_\BB(\k; \beta,\nu) = \frac{3\hat{v}(\0)\nu^{3/2}}{128\pi } \frac{|\k|^4}{\nu^2}\frac{1}{\beta\nu}\left( 1  + O\left(\beta\sqrt{\nu}|\k| +\frac{|\k|^2}{\nu}\right)\right)\\ \text{as } \frac{|\k|}{\sqrt{\nu}},  \beta\sqrt{\nu}|\k|\to 0.
		\end{multline} 
	\end{enumerate}
\end{thm}
\begin{remark}
	We observe that Equations \eqref{eq: reduced expression for Beliaev} and \eqref{eq: Beliaev for high temperatures} are independent of  the momentum dependent term $r(\k)$  of the interaction potential. This implies that, to the leading order in the small-momentum limit, the process depends on the interaction solely through the value of its Fourier transform at zero, $\hat{v}(\0)$.
\end{remark}

\begin{proof}
	
	\textit{(a)} For sufficiently small $\k$, the support of the Dirac delta distribution \begin{equation*}
		\delta(\omega_\Bog(\k)-\omega_\Bog(\p)-\omega_\Bog(\p-\k)),
	\end{equation*} will be restricted to a neighborhoodof $p:=|\p| = 0$ $$U_\k\subset \{\p\in\mathbb{R}^3 | p\in [0,k]\},$$ which can be taken as small as we like for $k\to 0$. Thus, by Lemma \ref{lemma: regularity of the Bogoliubov lemma} we can always suppose that the  set $U_\k$ is contained in a neighborhood of $0$, where the Bogoliubov dispersion relations $p\to \omega_{\Bog}(p)$ and $|\p-\k|=q\to \omega_\Bog(q)$ can be inverted. Thus, we can change variables as in Equation \eqref{eq: change in u e w} taking $u = \omega_\Bog(p)$ and $w=\omega_\Bog(q)$.  Hence, \eqref{gamma_B} takes the form  
	\begin{align}\label{eq: Beliaev uw}
		\frac{\hat{v}(0)}{64\pi \nu k \omega_\Bog(k)} \int \d u \d w G(u,w) f(u)f(w) \notag \\
		\times 
		\frac{(1-\e^{\frac{\beta}{2}(\omega_\Bog(k)+u+w)})^2}{(\e^{\beta\omega_\Bog(\k)}-1)(\e^{\beta u}-1)(\e^{\beta w}-1)}\delta(\omega_\Bog(k)-u-w)
	\end{align} where $G(u,w)$ has been defined in Proposistion \ref{prop: estimates for the remainder}. Following Proposistion \ref{prop: estimates for the remainder} we can now write \begin{equation}\label{eq: G in u w}
		G(u,w) = \frac{9}{\nu}u^2w^2(u+w)^2 + \nu^5 \frac{u^2 w^2}{\nu^4}S(u,w),
	\end{equation} where $S(u,w)$ is a remainder of order $O\left(\frac{u^4+w^4}{\nu^4}\right)$ as in Equations \eqref{eq: first remainder R},\eqref{eq: second remainder R}. In particular, since $u\to f(u)$ is a bounded, continuous function in the neighborhood $U_\k$, the expression in \eqref{eq: first form} guarantees tha the integral of the remainder is convergent also when the energies $u$ and $w$ approach $0$ and the thermal factor  in \eqref{eq: Beliaev uw}  becomes divergent. To evaluate the Dirac delta distribution we change now variables as in Equations \eqref{eq: change in x and y 1} and \eqref{eq: change in x and y 2}, so that $\delta(\omega_\Bog(k)-u-w) = \delta(\omega_\Bog(k)-x)$. If we carry out the $x$ integration, we see that the remaining integral in $y$ has an integration domain given by $-\omega_\Bog(k)\leq y \leq \omega_\Bog(k)$. This is a consequence of the convexity Assumption \eqref{eq: convexity for v} and of \cite[Lemma 4]{DLN_2023}.

	Let us first focus on the leading term on the right-hand side of \eqref{eq: G in u w}; the remainder $S$ will be treated afterward.
	
	We have to compute \begin{align}
		\frac{9\hat{v}(0)\omega_\Bog(k)}{2048 \pi \nu^2 k }(\e^{\beta\omega_\Bog(k)}-1) \int_{-\omega_\Bog(k)}^{\omega_\Bog(k)} \d y (\omega_\Bog(k)^2-y^2)^2\notag \\ \times f\left(\frac{\omega_\Bog(k)+y}{2}\right)f\left(\frac{\omega_\Bog(k)-y}{2}\right)\frac{1}{(\e^{\beta \frac{\omega_\Bog(k)+y}{2}}-1)(\e^{\beta \frac{\omega_\Bog(k)-y}{2}}-1)}.
	\end{align}Now, $f$ is a bounded, sufficinetly regular function which can be expanded around $u = 0$ as $f(u) = \nu^{-1}\left(+  O\left(\frac{u^2}{\nu^2}\right)\right)$. Thus,  after further rescaling the integration variable as $y := t \omega_\Bog(k)$ and denoting $\theta:= \beta\omega_\Bog(k)$ we obtain \begin{align}
		\frac{9\hat{v}(0)\omega_\Bog(k)^5}{2048 \pi \nu^4 k }\frac{1}{\beta\nu}\theta(1-\e^{-\theta})\int_{-1}^{1} \d t (1-t^2)^2\notag \\ \times \frac{1}{(1-\e^{\frac{\theta(1-t)}{2}})(1-\e^{-\frac{\theta(1+t)}{2}})}\left(1 + O\left(\frac{\omega_\Bog(k)^2}{\nu^2}\right)\right),
	\end{align} which gives exaclty Equation \eqref{eq: full expressio for Beliaev} after expanding the Bogoliubov dispersion relation $\omega_\Bog(k) = \sqrt{\nu}k + O\left(\frac{k^3}{\sqrt{\nu}}\right)$.
	
	It only remains to examine the remainder $S(u,w)$ in \eqref{eq: G in u w}. After changing variables as we already did for the main term, we are left with \begin{align}\label{eq: nice form for remainder}
		\frac{\hat{v}(0)\nu^2}{128 \pi  k }\frac{\omega_\Bog(k)^4}{\nu^4}(1-\e^{-\theta}) \int_{-1}^{1} \d t (1-t^2)^2S\left( \omega_\Bog(k)\frac{1+t}{2},\omega_\Bog(k)\frac{1-t}{2}\right)\notag \\ \times \frac{1}{(1-\e^{\frac{-\theta(1+t)}{2}})(1-\e^{\frac{-\theta(1-t)}{2} }) }\left(1 + O\left(\frac{\omega_\Bog(k)^2}{\nu^2}\right)\right),
	\end{align} where we have already expanded the functions $f$. Now, as $t\in[-1,1]$ we can estimate from above \begin{equation}
		\left|S\left( \omega_\Bog(k)\frac{1+t}{2},\omega_\Bog(k)\frac{1-t}{2}\right)\right| \leq K \frac{\omega_\Bog(k)^4}{\nu^4}
	\end{equation} for some suitable constant $K>0$. Thus, we estimate the integral of the remainder as \begin{align}
		|\eqref{eq: nice form for remainder}|\leq M \hat{v}(0)\nu^{3/2}\frac{\omega_\Bog(k)^7}{\nu^7} \frac{1}{\beta\nu}\theta (1-\e^{-\theta})\int^{1}_1 \d t \frac{(1-t^2)^2}{(1-\e^{-\theta\frac{1+t}{2}})(1-\e^{-\theta\frac{1-t}{2}})},
	\end{align} where  $M$ is some other positive constant. This concludes point \textit{(a)}.

	To prove \textit{(b)}, we have to compute  the limit of $\mathcal{I}(\theta)$ as $\theta\to +\infty$. This function can be computed explicitly in terms of polylogarithm functions of integer order. The computation goes as follows:
	\begin{align}
		&\frac{1}{32\theta (1-e^{-\theta})}\mathcal{I}(\theta)=\int^{1}_0 \d t\; \frac{t^4 + t^2 -2t^3}{(1-\e^{-\theta t})(1-\e^{-\theta(1-t)})} \notag \\ =& \sum_{m=0}\sum_{n=0}\e^{-m\theta}\int^1_0 \d t\; (t^4 +t^2- t^3)\e^{-(n-m)\theta t} \notag \\ 
  =& \frac{1}{30}+ 4\left(\frac{\zeta(3)}{\theta^3}-6\frac{\zeta(4)}{\theta^4}\right. \notag\\  &\left.+12\frac{\zeta(5)}{\theta^5}  -\frac{\operatorname{Li}_3(\e^{-\theta})}{\theta^3} -6\frac{\operatorname{Li}_4(\e^{-\theta})}{\theta^4} -12\frac{\operatorname{Li}_5(\e^{-\theta})}{\theta^5}\right).
	\end{align} 
	 To conclude one has
	
	\begin{align}
		\mathcal{I}(\theta):=  &32\theta(1-\e^{-\theta})\left[ \frac{1}{30} +4\left(\frac{\zeta(3)}{\theta^3} -6\frac{\zeta(4)}{\theta^4} + 12\frac{\zeta(5)}{\theta^5} - \frac{\operatorname{Li}_3(\e^{-\theta})}{(\theta)^3} \notag \right.\right.\\ & \left.\left. -6\frac{\operatorname{Li}_4(\e^{-\theta})}{\theta^4}-12\frac{\operatorname{Li}_5(\e^{-\theta)}}{\theta^5}\right)\right],\quad \theta \geq0.
	\end{align}
	
	Since $\operatorname{Li}_n(\e^{-\theta})$ decays exponentially as $\theta \to +\infty$, we have the asymptotic value \begin{equation}
		\mathcal{I}(\theta) = \frac{16}{15}\theta + O\left(\frac{1}{\theta^2}\right),\quad \text{as }\theta\to+\infty
	\end{equation}
	which gives the correct estimate in \textit{(b)}.
	Similarly for \textit{(c)}  we just have to estimate
        $\mathcal{I}(\theta)$ for $\theta\to 0$, that is \begin{align}
		\mathcal{I}(\theta)=&32\theta(1-\e^{-\theta})\int^{1}_{0}\d t \frac{t^2(1-t)^2}{(1-\e^{-t\theta} )(1-\e^{-(1-t)\theta} ) }\notag \\ =& 32\int^{1}_{0} \d t\; t(1-t) + O(\theta)\notag \\ =& \frac{16}{3} + O(\theta).
	\end{align} 
	
\end{proof}

\subsection{Landau damping rate}In order to carry out the analysis of Landau damping, we will rely on the following Lemma \begin{lemma}\label{lemma: applicationf for the theorem}
	Let \begin{equation}[0,+\infty[\times[0,+\infty[\ni (s,\delta)\rightarrow f(s,\delta)\end{equation} be continuous  and polinomially bounded, such that there exists a positive constant $C$ and an integer $n\geq 0$ for which \begin{equation}
		|f(s,\delta)| \leq C(s+\delta)^n,\qquad \text{as }(s,\delta)\to (0,0).
	\end{equation} Then, there exist strictly positive constants $A,B,D$, such that \begin{align}
	\left|	\int^{+\infty}_{0} \d t \frac{t^2 f\left(\frac{t}{x},\delta\right)}{\sinh(\frac{t}{2})\sinh(\frac{t+\theta}{2})} \right|\leq D \int^{+\infty}_{0}\d t \frac{t^2\left(\frac{t}{x}+\delta\right)^n}{\sinh(\frac{t}{2})\sinh(\frac{t+\theta}{2})} 
	\end{align} uniformily for $x\geq B >0$, $\delta\leq A$ and $\theta\geq 0$.
\end{lemma}\begin{proof}  

	We separate the integral into two regions
	
\begin{equation}
		\int^{+\infty}_{0} \d t  = \int^{x}_0 \d t+ \int_{x}^{+\infty}\d t.
	\end{equation} 
	
		\textbf{Integral for $t\geq x$}. We estimate the second integral first. Note that for $t\geq x>0$ , we can always find a positive constant $M$ such that \begin{equation}
	\frac{1}{\sinh\left(\frac{t}{2}\right)} \leq M\frac{1}{\sinh\left(\frac{t}{4}\right)^2},
	\end{equation} where $M$ can be chosen indipendently from $x$ as $t\geq x\geq B$. By hypotheses, $f(t,\delta)$ is polynomially bounded, thus \begin{equation}
	\sup_{t\geq x,\delta\in[0,A]} \frac{|f\left(\frac{t}{x},\delta\right)|}{\sinh\left(\frac{t}{4}\right)} =: L< +\infty
	\end{equation} Hence, we can estimate the integral as 
	\begin{align}
		\left|	\int^{+\infty}_{x} \d t \frac{t^2 f\left(\frac{t}{x},\delta\right)}{\sinh(\frac{t}{2})\sinh(\frac{t+\theta}{2})} \right| \leq ML \int^{+\infty}_{x} \d t \frac{t^2 \left(\frac{t}{x}+\delta\right)^n }{\sinh(\frac{t}{4})\sinh(\frac{t+\theta}{2} ) }\\
		\leq 2ML \int^{+\infty}_{\frac{x}{2}} \d t \frac{t^2 \left(\frac{t}{x}+\delta\right)^n }{\sinh(\frac{t}{2})\sinh(\frac{t+\theta}{2} ) }\\
		\leq 2ML \int^{+\infty}_{0} \d t \frac{t^2 \left(\frac{t}{x}+\delta\right)^n }{\sinh(\frac{t}{2})\sinh(\frac{t+\theta}{2} ) },
		\end{align} where we also used that $\frac{t}{x}+\delta \geq 1$ for $\delta\geq 0,\; t\geq x$.

	\textbf{Integral for $0\leq t\leq x$}. By hypotheses, there exists a neighborhood $U$ of $(0,0)$ in $[0,+\infty[^2$ such that \begin{equation}
		\left|f\left(\frac{t}{x},\delta\right)\right|\leq C\left(\frac{t}{x}+\delta\right)^n,\qquad \text{for }\left(\frac{t}{x},\delta\right)\in U.
	\end{equation} We fix now \begin{equation}
	0< \frac{1}{C'}:= \min_{(s,\delta)\in U^c \cap [0,1]\times[0,A]}(s+\delta)^n,
	\end{equation}and \begin{equation}
	F:=\max_{(s,\delta)\in [0,1]\times [0,A]}|f(s,\delta)|<+\infty.
	\end{equation}Then, we will have \begin{equation}
	\left|f\left(\frac{t}{x},\delta\right)\right| \leq \max\{C, C'F\}\left(\frac{t}{x}+\delta\right)^n.
	\end{equation}  Inserting the latter inequality in the integral for $t\leq x$ concludes the proof.
	
\end{proof} 

An important step in estimating \eqref{gamma_L} is to determine the
integration domain remaining after evaluating the Dirac delta distribution.
Writing $p=|\p|$, $q=|\p-\k|$, and $k=|\k|$, the relevant resonance condition is
\begin{equation*}
	y=\omega_\Bog(p)-\omega_\Bog(q)=-\omega_\Bog(k).
\end{equation*}
The following two lemmas describe the admissible radii for sufficiently
small $k$ and exclude resonances near stationary points of the dispersion
relation.

Throughout the argument, $\nu\in]0,\nu_0]$ is fixed. The constants may
depend on $\nu$ and the potential, but are independent of $k$ and $\beta$.
We require the dispersion relation to be real and strictly positive for
$p>0$; equivalently,
\begin{equation}\label{eq: Landau stability condition}
	\hat v(p)>-\frac{\hat v(0)p^2}{4\nu},
	\qquad p>0.
\end{equation}
In particular, this condition holds if $\hat v\geq0$.

\begin{lemma}\label{lemma: integration domain Landau}
	Fix $\nu_0>0$ and suppose that the potential satisfies
	\eqref{eq: continuity of Fourier}--\eqref{eq: symmetry requirement} and
	\eqref{eq: rot invariace}--\eqref{eq: convexity for v}
	for every $\nu\in]0,\nu_0]$.
	For each fixed $\nu\in]0,\nu_0]$, there exist constants $K,\delta>0$
	such that, for every $\k\in\mathbb R^3$ with $0<k=|\k|\leq\delta$
	and every $p\in[0,K]$, there is a vector $\p\in\mathbb R^3$ satisfying
	$|\p|=p$ and
	\begin{equation}
		\omega_\Bog(\p)-\omega_\Bog(\p-\k)=-\omega_\Bog(\k).
	\end{equation}
\end{lemma}

\begin{proof}
	The assumptions on the potential imply that $\omega_\Bog$ is $C^5$,
	strictly increasing, and strictly convex on $[0,\rho]$, for some $\rho>0$.
	Choose $K,\delta>0$ so that $\delta<K$ and $K+\delta<\rho$.
	The case $p=0$ is immediate. Suppose therefore that $p>0$.
	
	If $\p$ is parallel to $\k$ and points in the same direction, then
	$q=|p-k|$ and
	\begin{equation*}
		\omega_\Bog(p)-\omega_\Bog(|p-k|)>-\omega_\Bog(k).
	\end{equation*}
	Indeed, the left-hand side is nonnegative if $p\geq k$; if $0<p<k$,
	the inequality follows from strict monotonicity and $\omega_\Bog(p)>0$.
	
	If $\p$ points in the opposite direction, then $q=p+k$. Set
	\begin{equation*}
		f_k(p):=\omega_\Bog(p)-\omega_\Bog(p+k).
	\end{equation*}
	Strict convexity gives
	\begin{equation*}
		f_k'(p)=\omega_\Bog'(p)-\omega_\Bog'(p+k)<0,
	\end{equation*}
	and hence
	\begin{equation*}
		f_k(p)<f_k(0)=-\omega_\Bog(k),
		\qquad p>0.
	\end{equation*}
	The quantity $\omega_\Bog(p)-\omega_\Bog(|\p-\k|)$ depends
	continuously on the angle between $\p$ and $\k$. The intermediate value
	theorem therefore yields the required vector. For $p>0$, the resonant
	angle lies strictly between $0$ and $\pi$.
\end{proof}

The dispersion relation need not be globally invertible, so changes to
energy variables must be performed separately on its monotonicity
intervals. Lemma~\ref{lemma: regularity of the Bogoliubov lemma}
provides a local inverse near $p=0$, while
Lemma~\ref{lemma: integration domain Landau} shows that every sufficiently
small radius admits a resonant direction. The next lemma excludes
resonances whenever the derivative of the dispersion relation becomes
too small between the two radii.

\begin{lemma}\label{lemma: rotons and maxons}
	Assume the hypotheses of Lemma~\ref{lemma: integration domain Landau},
	\eqref{eq: regularity everywhere}, and
	\eqref{eq: Landau stability condition}.
	For every compact set $P\subset[0,+\infty[$, there exist constants
	$\delta,\epsilon>0$ such that, whenever
	\begin{equation*}
		0<k\leq\delta,
		\qquad p\in P,
		\qquad q\geq0,
		\qquad |p-q|\leq k,
	\end{equation*}
	the condition
	\begin{equation*}
		\min_{r\in[\min\{p,q\},\max\{p,q\}]}
		|\omega_\Bog'(r)|\leq\epsilon\sqrt\nu
	\end{equation*}
	implies
	\begin{equation}
		|\omega_\Bog(p)-\omega_\Bog(q)|<\omega_\Bog(k).
	\end{equation}
	In particular, the resonance condition cannot hold if either endpoint
	has derivative of absolute value at most $\epsilon\sqrt\nu$, or if a
	stationary point lies between $p$ and $q$.
\end{lemma}

\begin{proof}
	Choose $R>0$ such that $P\subset[0,R]$ and restrict to $\delta\leq1$.
	Then $p$, $q$, and the interval between them belong to $[0,R+1]$.
	The assumptions imply that $\omega_\Bog'$ is continuous on this compact
	interval, with the right derivative $\omega_\Bog'(0)=\sqrt\nu$ at zero.
	By uniform continuity, we can choose $\delta>0$ sufficiently small that
	\begin{equation*}
		|\omega_\Bog'(r)-\omega_\Bog'(s)|
		<\frac{\sqrt\nu}{4}
		\quad\text{whenever }
		r,s\in[0,R+1],
		\quad |r-s|\leq\delta.
	\end{equation*}
	Moreover, $\frac{\omega_\Bog(k)}{k}\to\sqrt\nu$ as $k\to0$.
	After decreasing $\delta$ if necessary, we therefore have
	\begin{equation*}
		\omega_\Bog(k)\geq\frac{3\sqrt\nu}{4}k,
		\qquad 0<k\leq\delta.
	\end{equation*}
	Take $\epsilon=\frac{1}{4}$. If the stated minimum is at most
	$\epsilon\sqrt\nu$, then $|\omega_\Bog'(r)|<\frac{\sqrt\nu}{2}$ throughout
	the interval between $p$ and $q$. Thus
	\begin{align*}
		|\omega_\Bog(p)-\omega_\Bog(q)|
		&=\left|(p-q)\int_0^1\d s\,
		\omega_\Bog'\bigl(q+s(p-q)\bigr)\right|\\
		&\leq\frac{\sqrt\nu}{2}|p-q|
		\leq\frac{\sqrt\nu}{2}k
		<\omega_\Bog(k).
	\end{align*}
\end{proof}

These two lemmas yield the following representation of the Landau
damping integral.

\begin{prop}\label{prop: invertibility cor}
	Assume the hypotheses of Lemma~\ref{lemma: rotons and maxons} and
	\eqref{eq: no plateau}.
	For each fixed $\nu\in]0,\nu_0]$, there exist $\delta>0$ and a finite
	integer $N\geq0$, independent of $k$ and $\beta$, with the following
	properties.
	
	Let $I_0,\ldots,I_N$ be the intervals obtained by removing the stationary
	points of $\omega_\Bog$ from $[0,+\infty[$, ordered by increasing
	momentum, with $0\in I_0$. Write $p_n$ for the inverse of
	$\omega_\Bog|_{I_n}$ and set
	\begin{equation*}
		f_n(u):=\frac{p_n(u)}{u}p_n'(u),
		\qquad u>0,
	\end{equation*}
	with $f_0(0):=\frac{1}{\nu}$ by continuous extension. For $s>0$ sufficiently
	small and $u>0$ with $u,u+s\in\omega_\Bog(I_n)$, define
	\begin{equation}\label{eq: Landau branch vertex}
		G_n(u,s):=\frac{8\nu}{\hat v(0)}\,us(u+s)
		j\bigl(p_n(u+s);p_0(s),p_n(u)\bigr)^2.
	\end{equation}
	Here $j$ is regarded as a function of the three momentum magnitudes,
	using rotational invariance. On the branch containing zero, $G_0$
	agrees near the origin with the function $G$ introduced above and
	extends continuously to $u=0$.
	
	For $0<k\leq\delta$ and $\beta>0$, put $a:=\omega_\Bog(k)$ and define
	\begin{align*}
		J_n=J_n(k,\beta):=\biggl\{t\geq0:\;&
		\frac{t}{\beta},\frac{t}{\beta}+a\in\omega_\Bog(I_n),\\
		&\left|p_n\left(\frac{t}{\beta}\right)
		-p_n\left(\frac{t}{\beta}+a\right)\right|\leq k
		\leq p_n\left(\frac{t}{\beta}\right)
		+p_n\left(\frac{t}{\beta}+a\right)\biggr\}.
	\end{align*}
	Let $\lambda_n(t)$ be $1$ when both triangle inequalities are strict
	and $\frac{1}{2}$ when one is an equality. The latter value is the boundary
	weight for when the zero of the Dirac delta distribution lies exactly on the boubdary of the integration interval.
	
	Then
	\begin{align}\label{eq: t integral}
		\gamma_L(\k;\beta,\nu)
		&=\frac{\hat v(0)}{32\pi k\omega_\Bog(k)}
		\frac{1}{\beta\nu}
		\sum_{n=0}^N\int_{J_n}\d t\,\lambda_n(t)
		G_n\left(\frac{t}{\beta},\omega_\Bog(k)\right)\notag\\
		&\quad\times f_n\left(\frac{t}{\beta}\right)
		f_n\left(\frac{t}{\beta}+\omega_\Bog(k)\right)
		\frac{\bigl(\e^{\beta\omega_\Bog(k)}-1\bigr)\e^t}
		{(\e^t-1)(\e^{t+\beta\omega_\Bog(k)}-1)}.
	\end{align}
	Both factors $f_n$ are bounded in absolute value on these domains,
	uniformly in $0<k\leq\delta$ and $\beta>0$.
	Furthermore, there are constants $b_0>0$ and $b_n\geq b_0$ for
	$1\leq n\leq N$, independent of $k$ and $\beta$, such that
	\begin{equation*}
		[0,\beta\nu b_0]\subset J_0\subset[0,+\infty[,
		\qquad
		\lambda_0(t)=1
		\quad\text{for }0<t\leq\beta\nu b_0,
	\end{equation*}
	and
	\begin{equation*}
		J_n\subset[\beta\nu b_n,+\infty[,
		\qquad 1\leq n\leq N.
	\end{equation*}
\end{prop}

\begin{proof}
	It is convenient to work with radial coordinates rather than the resonance
	set in $\mathbb R^3$. Set
	\begin{equation*}
		\mathcal R_k:=\bigl\{(p,q)\in[0,+\infty[^2:
		|p-q|\leq k\leq p+q,\quad
		\omega_\Bog(q)=\omega_\Bog(p)+\omega_\Bog(k)\bigr\}.
	\end{equation*}
	The triangle inequalities are precisely the conditions for the
	existence of a vector $\p$ with $|\p|=p$ and $|\p-\k|=q$.
	Thus Lemma~\ref{lemma: integration domain Landau} gives
	\begin{equation}\label{eq: subinterval}
		[0,K]\subset
		\bigl\{p\geq0:\text{there exists }q\geq0
		\text{ such that }(p,q)\in\mathcal R_k\bigr\}.
	\end{equation}

	By \eqref{eq: no plateau}, $\omega_\Bog'$ has only finitely many
	zeros, and there are $R,\epsilon_\infty>0$ such that
	\begin{equation}\label{eq: good containment at infinity}
		\frac{\omega_\Bog(r)\omega_\Bog'(r)}{r}
		=\frac{r^2}{2}
		+\frac{\nu}{\hat v(0)}
		\left(\hat v(r)+\frac r2\hat v'(r)\right)
		\geq\epsilon_\infty,
		\qquad r\geq R.
	\end{equation}
	Thus, all stationary points lie in $[0,R]$.
	Apply Lemma~\ref{lemma: rotons and maxons} to the compact set $[0,R+1]$ and choose
	$\delta\leq1$ accordingly. For every resonant pair $(p,q)$ for which $p\in [0,R+1]$, it follows that
	\begin{equation}\label{eq: good containment}
		\min_{r\in[\min\{p,q\},\max\{p,q\}]}
		|\omega_\Bog'(r)|>\epsilon\sqrt\nu.
	\end{equation}
	If $p>R+1$, then $q>R$ and the interval between $p$ and $q$ lies in
	the increasing tail. Consequently, every resonant pair $(p,q)\in \mathcal{R}_k$ belongs to
	$I_n\times I_n$ for a single $n$.
	
	Since $\hat v$ is bounded, $\omega_\Bog(r) \geq C\frac{r^2}{2}$ as
	$r\to+\infty$ for some $C>0$. Together with 
	$\frac{\omega_\Bog(r)}{r}\to\sqrt\nu$ at zero and $\omega_\Bog(r)>0$ for $r>0$, this implies that
	\begin{equation*}
		\frac{r}{\omega_\Bog(r)}
		\leq\frac C{\sqrt\nu},
		\qquad r>0,
	\end{equation*}
	for some $C>0$. Therefore,
	\eqref{eq: good containment at infinity} and
	\eqref{eq: good containment} yield
	\begin{equation}\label{eq: precise estimate for f}
		|f_n(u)|
		=\frac{p_n(u)}{u\,|\omega_\Bog'(p_n(u))|}
		\leq\max\left\{
		\frac{C}{\epsilon\nu},
		\frac1{\epsilon_\infty}
		\right\},\qquad p_n(u)\in I_n.
	\end{equation}
	The same estimate holds at $u+\omega_\Bog(k)$.
	
	Integrating out the azimuthal angle and changing from the polar angle
	to $q$ gives the measure $$\frac{2\pi pq}{k}\,\d p\,\d q.$$
	On $I_n\times I_n$, set $u=\omega_\Bog(p)$ and
	$w=\omega_\Bog(q)$. The absolute value of the Jacobian is
	\begin{equation*}
		pq\,\d p\,\d q
		=uw\,|f_n(u)f_n(w)|\,\d u\,\d w
		=uw\,f_n(u)f_n(w)\,\d u\,\d w,
	\end{equation*}
	because the two inverse derivatives have the same sign on a single
interval. The delta distribution is $\delta(w-u-a)$, where
	$a=\omega_\Bog(k)>0$. On its support, the thermal factor simplifies to
	\begin{equation*}
		\frac{
			\bigl(\e^{\frac\beta2(a+w)}
			-\e^{\frac\beta2u}\bigr)^2
		}{
			(\e^{\beta a}-1)
			(\e^{\beta u}-1)
			(\e^{\beta w}-1)
		}
		=
		\frac{
			(\e^{\beta a}-1)\e^{\beta u}
		}{
			(\e^{\beta u}-1)
			(\e^{\beta(u+a)}-1)
		}.
	\end{equation*}
	Using \eqref{eq: Landau branch vertex}, evaluating the $w$ integral,
	and setting $t=\beta u$ gives \eqref{eq: t integral}, with the
	boundary weights $\lambda_n(t)$ specified in the statement. In the coordinates
	$x,y$ of
	\eqref{eq: change in x and y 1}--\eqref{eq: change in x and y 2},
	these steps amount to setting $y=-a$ and
	\begin{equation*}
		t=\frac\beta2(x-a)
		=\beta\nu\frac{x-\omega_\Bog(k)}{2\nu},
	\end{equation*}
	in agreement with the rescaling introduced above.
	
	Finally, choose $K$ as in
	Lemma~\ref{lemma: integration domain Landau} and decrease $\delta$
	so that $[0,K+\delta]\subset I_0$.
	For every $p\in[0,K]$, a resonant $q$ exists and satisfies
	$q\leq p+k\leq K+\delta$. Thus, $(p,q)\in I_0\times I_0$ and
	\begin{equation*}
		[0,\beta\omega_\Bog(K)]\subset J_0.
	\end{equation*}
	The strict inequalities in the proof
	of lemma \ref{lemma: integration domain Landau} also imply $\lambda_0=1$ on this interval away from zero.
	If $N\geq1$, each $I_n$, $n\geq1$, is bounded away from zero. Hence
	\begin{equation*}
		m_n:=\inf_{r\in I_n}\omega_\Bog(r)>0.
	\end{equation*}
	 Set
	$b_n=\frac{m_n}{\nu}$ for $n\geq1$ and choose
	\begin{equation*}
		0<b_0\leq\min\left\{
		\frac{\omega_\Bog(K)}\nu,
		b_1,\ldots,b_N
		\right\},
	\end{equation*}
	omitting the $b_n$ when $N=0$. These constants give the asserted
	inclusions and are independent of $k$ and $\beta$.
\end{proof}

We can now state and prove the theorem estimating the Landau damping rate.

\begin{thm}\label{thm: Landau damping}
	Fix a $\nu_0 \in ] 0,+\infty[$ and  and suppose the potential satisfies \eqref{eq: continuity of Fourier}--\eqref{eq: no plateau} for any $\nu\in]0, \nu_0]$. Then, in the limit of small temperature $\frac{1}{\beta\nu} \to 0$ and small momenta $\frac{|\k|}{\sqrt{\nu}}\to 0$: 
	\begin{enumerate}
		\item [(a)] the Landau damping rate is equal to \begin{align}\label{eq: full expression for Landau}
			\gamma_\LL(\k;\beta,\nu) =& \frac{9\hat{v}(\0)\nu^{3/2}}{64\pi}\frac{1}{(\beta\nu)^5}\Big{[} \mathcal{G}_4(\beta\omega_\Bog(\k))\notag \\ +&2\beta\sqrt{\nu}|\k|\mathcal{G}_3(\beta\omega_\Bog(\k))\notag \\ +& (\beta\sqrt{\nu}|\k|)^2 \mathcal{G}_2(\beta\omega_\Bog(\k)) \Big{]}\left(1+O\left(\frac{1}{(\beta\nu)^2}+ \frac{|\k|^2}{\nu}\right)\right) \\ &\text{as } \frac{|\k|}{\sqrt{\nu}},\;\frac{1}{\beta\nu}\to 0
		\end{align} where \begin{align}
			\mathcal{G}_k(\theta)&= \int^{+\infty}_0 \d t \frac{\sinh(\frac{\theta}{2})t^k}{\sinh(\frac{t+\theta}{2})\sinh(\frac{t}{2})}
		\end{align} is a continuous function of $\theta\in [0,+\infty)$ satisfying \begin{align}
			\mathcal{G}_{k}(\theta) &= 2\Gamma(k+1)\zeta(k) \theta+ O(\theta^2)\quad \text{as } \theta\to 0,\\
			\mathcal{G}_{k}(\theta) &\label{eq: asymptotic for Landau special II}= 2\Gamma(k+1)\zeta(k+1) + O(\e^{-\theta})\quad \text{as } \theta\to +\infty.
		\end{align}
		\item [(b)] Additionally, in the limit $\beta\sqrt{\nu}|\k|\to 0$, the Landau damping rate can be extimated as \begin{align}\label{eq: very high temperatures}
			\gamma_\LL(\k;\beta,\nu) =  \frac{3\pi^3\hat{v}(\0)\nu^{3/2}}{40}\frac{|\k|}{\sqrt{\nu}}\frac{1}{(\beta\nu)^4}	\left(1 +  O\left( \beta\sqrt{\nu}|\k|+  \frac{1}{(\beta\nu)^2}\right)\right)\notag \\  \text{as }\frac{|\k|}{\sqrt{\nu}},\;\frac{1}{\beta\nu},\;\beta\sqrt{\nu}|\k|\to 0 
		\end{align}
		\item[(c)] In the opposite limit, $\frac{1}{\beta\sqrt{\nu}|\k|}\to 0$, the Landau damping rate can be extimated as 
		\begin{align}\label{eq: basse temperature}
			\gamma_\LL(\k;\beta,\nu) = \frac{9\zeta(3)\hat{v}(\0)\nu^{3/2}}{16\pi}\frac{1}{(\beta\nu)^3}\frac{|\k|^2}{\nu}\left(1 + O\left(\frac{1}{\beta\sqrt{\nu}|\k|} +\frac{|\k|^2}{\nu}\right)\right) \notag \\ \text{as }\frac{|\k|}{\sqrt{\nu}},\;\frac{1}{\beta\nu},\;\frac{1}{\beta\sqrt{\nu}|\k|}\to 0
		\end{align}
	\end{enumerate}
\end{thm}
\begin{proof}
	
	\textit{(a)} We define \begin{equation}
	\delta:= \frac{\omega_\Bog(k)}{\nu},\qquad \theta = \beta\omega_\Bog(k).
	\end{equation}
	
	Then, using Proposistion \ref{prop: invertibility cor} we can write the Landau damping integral as
	\begin{align}
		\sum_{n=0}^N\frac{\hat{v}(0)}{32\pi k \omega_\Bog(k)}\frac{1}{\beta\nu}\int_0^{+\infty} \d t  \chi_{J_n}(t) G\left(\frac{\nu t}{\beta\nu}, \nu\delta\right)f_n\left(\frac{\nu t}{\beta\nu}\right)f_n\left(\frac{\nu t+\nu \theta}{\beta\nu} \right)\notag \\ \times \frac{\left(\e^{\theta} - 1\right)\e^{ t}}{(\e^{ t}-1)(\e^{t +\theta}-1)},
	\end{align} where $\chi_{J_n}$ is the characteristic function of the set $J_n$ introduced in the same Proposistion \eqref{prop: invertibility cor}. Without loss of generality, we can suppose that $ [0,\beta\nu]\subset J_0$. Now, we separate the integral into two regions of integration \begin{equation}
		\int^{+\infty}_0 = \int^{\beta\nu}_{0} +  \int^{+\infty}_{\beta\nu},
	\end{equation} \textbf{Integral for $t\leq \beta\nu$.} Let us write $G$ as in \eqref{eq: first form} and expand $f$ around $0$, so to obtain  \begin{align}\label{eq: low momentum region}
		\frac{\hat{v}(0)}{32\pi k \omega_\Bog(k)\nu^2}\frac{1}{\beta\nu}\int_0^{\beta\nu} \d t \; \left[ 9\nu^5\delta^2 \frac{t^2}{(\beta\nu)^2}(\delta+\frac{t}{\beta\nu})^2\notag  \right.\\ \left. + \frac{t^2}{(\beta\nu)^2}\delta^2S\left(\frac{\nu t}{\beta\nu},\nu\delta\right) \right] \frac{\left(\e^{\theta} - 1\right)\e^{ t}}{(\e^{ t}-1)(\e^{t +\theta}-1)}\left(1 + O\left(\frac{t^2}{(\beta\nu)^2}\right)\right).
	\end{align} Let us first focus on the leading term on in Equation \eqref{eq: low momentum region}. This can be written as  \begin{align}
		&\int_0^{\beta\nu} \d t \; \delta^2 \frac{t^2}{(\beta\nu)^2}(\delta+\frac{t}{\beta\nu})^2\frac{\left(\e^{\theta} - 1\right)\e^{ t}}{(\e^{ t}-1)(\e^{t +\theta}-1)} \notag \\
		=&	\int_0^{+\infty} \d t \; \delta^2 \frac{t^2}{(\beta\nu)^2}(\delta+\frac{t}{\beta\nu})^2\frac{\left(\e^{\theta} - 1\right)\e^{ t}}{(\e^{ t}-1)(\e^{t +\theta}-1)}(1+ O(\e^{-\frac{\beta\nu}{2}})),
	\end{align}  
	where  we have estimated the remainder as in  Lemma \ref{lemma: applicationf for the theorem}.
	We know that  \begin{equation}
		S\left(\frac{\nu t}{\beta\nu},\nu\delta\right) = O\left(\left(\frac{t}{\beta\nu} + \delta\right)^4\right),\qquad \text{as } \delta,\frac{t}{\beta\nu} \to 0,
	\end{equation}thus the remainder term in \eqref{eq: low momentum region} can be estimated directly using Lemma \ref{lemma: applicationf for the theorem}. This gives a term of order \begin{align}\label{eq: integral at high energy}
		\frac{\hat{v}(0)}{ k \omega_\Bog(k)\nu^2}\frac{1}{\beta\nu}\int_0^{+\infty} \d t \;  \nu^5\delta^2 \frac{t^2}{(\beta\nu)^2}(\delta+\frac{t}{\beta\nu})^4\frac{\left(\e^{\theta} - 1\right)\e^{ t}}{(\e^{ t}-1)(\e^{t +\theta}-1)}.
	\end{align} 
	
	\textbf{Integral for} $t\geq \beta\nu$. This integration region produces a contribution which is exponentially damped as $\frac{1}{\beta\nu }\to 0$. Indeed, following Lemma \ref{lemma: lagrange} we can write \begin{equation}G\left(\frac{\nu t}{\beta\nu},\nu\delta\right) = \nu^2\delta^2 \frac{t^2}{(\beta\nu)^2}K\left(\frac{\nu t}{\beta\nu},\nu\delta\right)\end{equation}
	where $K$ is continuous and polynomially bounded in its arguments.  In addition, the function $f$ is continuous and bounded by Proposistion \ref{prop: invertibility cor}, Equation \eqref{eq: precise estimate for f}.  Thus, following again Lemma \ref{lemma: applicationf for the theorem}, we have the bound \begin{align}
	\left| \int_{\beta\nu}^{+\infty} \d t \chi_{J_n}(t) G\left(\frac{\nu t}{\beta\nu}, \nu\delta\right)f_n\left(\frac{\nu t}{\beta\nu}\right)f_n\left(\frac{\nu t+\nu \theta}{\beta\nu} \right)\frac{\left(\e^{\theta} - 1\right)\e^{ t}}{(\e^{ t}-1)(\e^{t +\theta}-1)}\right| \notag  \\
		 \leq  \e^{-\frac{\beta\nu}{2}}\delta^2\nu^2\int^{+\infty}_0 \frac{t^2}{(\beta\nu)^2}P\left(\frac{t}{\beta\nu}\right) \frac{\left(\e^{\theta} - 1\right)\e^{ t}}{(\e^{ t}-1)(\e^{t +\theta}-1)},
	\end{align} where $t\to P(t)$ is a suitable polynomial whose coefficients do not depend on $\beta\nu$ nor on $\delta$.  
	
	All the integrals in the preceding discussion can be re-arranged as \begin{equation}
		\mathcal{G}_k(\theta):=\int^{+\infty}_{0}\d t \frac{\sinh(\frac{\theta}{2})t^k}{\sinh(\frac{t}{2})\sinh(\frac{t+\theta}{2})},
	\end{equation} for some integer $k\geq 0$. The integrals defining the special functions $\mathcal{G}_k$ can be computed explicitely in terms of the polylogarithm functions of integer order $\operatorname{Li}_n $ via  an expansion of the hyperbolic sine \begin{align}\label{eq: hyper computation}
		\mathcal{G}_k(\theta)=&\int^{+\infty}_{0}\d t \frac{\sinh(\frac{\theta}{2})t^k}{\sinh(\frac{t}{2})\sinh(\frac{t+\theta}{2})}\notag \\ =& 4\sinh(\frac{\theta}{2})\e^{-\frac{\theta}{2}} \int^{+\infty}_{0} \d t\sum_{n=0}^{+\infty}\sum_{m=0}^{+\infty}\e^{-m\theta}\e^{-t(n+m+1)} t^{k} \notag\\ 
		=& 2\Gamma(k+1)\sum_{p=1}^{+\infty}\frac{1}{p^{k+1}}\left(1- \e^{-p\theta}\right)\notag \\
		 =& 2\Gamma(k+1)\left(\operatorname{Li}_{k+1}(1)- \operatorname{Li}_{k+1}(\e^{-\theta})\right).
	\end{align}  In particular, Equation \eqref{eq: hyper computation} gives a  bounded function of $\theta$ for every $k\geq 2$ with asymptotic values \begin{align}\label{eq: asymptotic for Landau special I}
		\mathcal{G}_{k}(\theta) &= 2\Gamma(k+1)\zeta(k) \theta+ O(\theta^2)\quad \text{as } \theta\to 0,\\
		\mathcal{G}_{k}(\theta) &\label{eq: asymptotic for Landau special III}= 2\Gamma(k+1)\zeta(k+1) + O(\e^{-\theta})\quad \text{as } \theta\to +\infty.
	\end{align}
	
	Then, the Landau damping can be written for $\frac{|\k|}{\sqrt{\nu}},\frac{1}{\beta\nu}\to 0$ as 
	\begin{align}\label{eq: result for a}
		\gamma_\LL(\k;\beta,\nu) =& \frac{9\hat{v}(\0)\nu^{3/2}}{64\pi}\frac{1}{(\beta\nu)^5}\Big{[} \mathcal{G}_4(\beta\omega_\Bog(\k))\notag \\ +&2\beta\sqrt{\nu}|\k|\mathcal{G}_3(\beta\omega_\Bog(\k))\notag \\ +& (\beta\sqrt{\nu}|\k|)^2 \mathcal{G}_2(\beta\omega_\Bog(\k)) \Big{]}\left(1+O\left(\frac{1}{(\beta\nu)^2}+ \frac{|\k|^2}{\nu}\right)\right) 
	\end{align} 
	where we  have expanded $\delta = \frac{\omega_\Bog(\k)}{\nu} = \frac{|\k|}{\sqrt{\nu}} + O\left(\frac{|\k^3|}{\nu^{3/2}}\right)$.
	This concludes the proof of \textit{(a)}.
	
	Points \textit{(b)} and \textit{(c)} are  immediately verified by exploiting the asymptotic properties of $\mathcal{G}_k$ in Equations \eqref{eq: asymptotic for Landau special I},\eqref{eq: asymptotic for Landau special III}.
\end{proof}

\subsection{High temperature results}\label{sec: high temperature results}

The description of the high temperature limit for the integrals \eqref{gamma_B} and \eqref{gamma_L} requires a careful physical analysis of the scales at play. For the Beliaev damping, the only dependence on $\beta$ in Equation \eqref{eq: full expressio for Beliaev} is through the product $\beta\omega_\Bog(\k)$. Thus, while the limit $\beta\omega_\Bog(\k) \rightarrow +\infty$ necessarily corresponds to a small temperature limit, the opposite limit $\beta\omega_\Bog(\k )\rightarrow 0$ can be achieved  for $\frac{|\k|}{\sqrt{
		\nu}} \to 0$, compatibly with the hypotheses $T< T_c$; where $T_c$ is the critical temperature relative to  Bose-Einstein phase transition.  Similarly, for the Landau damping, the limit $\beta\omega_\Bog(\k)\to  0$ could in principle be tested.

Now, we analyze the limit $\beta\nu \to  0$. In this regime, the parameters have to be chosen so that the assumption $T<T_c$ can still be verified. If we substitue for the density of particles $\frac{n_0(\beta)}{L^3}$ its value for a free Bose gas, we can estimate the critical temperature as $$ T_c \sim \left(\frac{n_0}{L^d}\right)^{\frac{2}{3}}.$$ Then, since $\nu= \hat{v}(\0)n_0$, the requirement for $\beta\nu$ to be small is equivalent to $$ \frac{\beta^{\frac{3}{2}}}{\beta_c^{\frac{3}{2}}}\frac{\hat{v}(\0)\nu^{\frac{1}{2}}}{(\beta\nu)^{\frac{1}{2}}} \ll 1.$$ This condition could be valid in the limit of a very \textit{dilute} potential.

However, numerical analysis, \cite[Fig. 1]{PS_97} and \cite[Fig. 1]{G_98}, show that the convergence towards this asymptotic regime is rather slow, requiring very large values of the temperature.   Moreover, for high temperatures the Landau damping rate depends on the high momenta region of $\hat{v}$, so that its computation becomes highly potential-dependent.

To compare our results with the pre-existing literature we compute this rate assuming the specific form for the potential\begin{equation}
	\hat{v}(\k) = \begin{dcases}\label{eq: high temperature potential}
		\hat{v}(\0)> 0 & |\k|<\Lambda\\
		0 & |\k|\geq 2\Lambda
	\end{dcases},
\end{equation} and $\hat{v}(\k)$ smoothly interpolating between $\hat{v}(\0)$ and $0$ for $\Lambda<|\k|<2\Lambda$. In the  computation of the rate, we will remove  the cut-off by taking the limit $\Lambda\to +\infty$. The final expression for the Landau damping rate coincides with the one estimated in the literature  \cite{HG_98, PS_97,SK_74}.
\begin{prop}\label{prop: high temperature}
	Suppose the potential is of the form \eqref{eq: high temperature potential}. Then, after having performed the limit $\Lambda \to +\infty$, we find \begin{equation}\label{eq: high temp landau rate}
		\gamma_\LL(\k;\beta,\nu) = \frac{3\hat{v}(\0)\nu^{3/2}}{32}\frac{1}{\beta\nu}\frac{|\k|}{\sqrt{\nu}}\left(1+ O\left( \frac{|\k|^2}{\nu}+\beta\nu\right)\right) \quad \text{as }\frac{|\k|}{\sqrt{\nu}},\beta\nu \to 0.
	\end{equation} 
\end{prop}
\begin{proof}
	In the high temperature limit,	the contribution of momenta outside the phonon regime must
	be retained. We work in the flat Fourier-transform limit
	$\hat{v}(\k)=\hat{v}(\0)$, after removing the cutoff $\Lambda\to+\infty$ at fixed
	$k,\beta,\nu>0$. This can be justified by a simple dominated convergence argument, using the exponential decay of the thermal factor.
	
	In this limit, the Bogoliubov dispersion relation becomes
	\begin{equation}
		\omega_\Bog(k)
		=\sqrt{\frac{k^4}{4}+\nu k^2}.
	\end{equation}
	It is strictly increasing for $k\geq0$, with inverse
	\begin{equation}
		k(\omega)
		=\sqrt{-2\nu+2\sqrt{\nu^2+\omega^2}}.
	\end{equation}
	We may therefore introduce the variables $x$ and $y$
	as in \eqref{eq: change in x and y 1}
	and \eqref{eq: change in x and y 2}.
	Setting
	\begin{equation}
		z:=\frac{x-\omega_\Bog(k)}{2\nu},
		\qquad
		\delta:=\frac{\omega_\Bog(k)}{\nu},
		\qquad
		\theta:=\beta\omega_\Bog(k)=\beta\nu\delta,
	\end{equation}
	we obtain
	\begin{align}\label{eq: high temperature integral}
		\gamma_\LL(\k;\beta,\nu)
		&=
		\frac{\hat{v}(\0)}{32\pi k\omega_\Bog(k)}
		\int_0^{+\infty}\d z\,
		G(\nu z,\nu\delta)
		f(\nu z)f(\nu z+\nu\delta)
		\notag\\
		&\qquad\times
		\frac{(\e^\theta-1)\e^{\beta\nu z}}
		{(\e^{\beta\nu z}-1)
			(\e^{\beta\nu z+\theta}-1)}.
	\end{align}
	
	The flat Fourier transform also gives
	\begin{equation}
		f(\nu z)=\frac{1}{\nu\sqrt{1+z^2}},
		\qquad
		f(\nu z+\nu\delta)
		=\frac{1}{\nu\sqrt{1+(z+\delta)^2}},
	\end{equation}
	and
	\begin{align}
		G(\nu z,\nu\delta)
		=4\nu^2\Bigl[
		&c(\nu z+\nu\delta)c(\nu z)c(\nu\delta)
		-s(\nu z+\nu\delta)s(\nu z)s(\nu\delta)
		\notag\\
		&+s(\nu z+\nu\delta)s(\nu\delta)c(\nu z)
		-c(\nu z+\nu\delta)c(\nu\delta)s(\nu z)
		\notag\\
		&+s(\nu z+\nu\delta)s(\nu z)c(\nu\delta)
		-c(\nu z+\nu\delta)c(\nu z)s(\nu\delta)
		\Bigr]^2.
	\end{align}
	
	To obtain a low-momentum expansion uniform in $\beta\nu$,
	we first center the integral by replacing $z$ with
	$z-\frac{\delta}{2}$, so that the new integration interval is
	$[	\frac{\delta}{2},+\infty[$. A direct expansion as $\delta\to 0$ of the centered
	expression gives, for $z\geq\delta$,
	\begin{align}
		&\frac{
			G\bigl(\nu(z-\frac{\delta}{2}),\nu\delta\bigr)
			f\bigl(\nu(z-\frac{\delta}{2})\bigr)
			f\bigl(\nu(z+\frac{\delta}{2})\bigr)
		}{8\nu^3\delta^2}\label{normalized expression}
		\\
		=&
		\frac12
		\left(
		2-\frac{1}{\sqrt{1+z^2}}-\frac{1}{1+z^2}
		\right)^2
		+O\bigl(\delta^2\min\{z^2,1\}\bigr).
	\end{align}
	The remainder is uniform in $z\geq\delta$. The absence
	of a linear term in $\delta$ follows from the fact that the centered
	expression is even in $\delta$.
	
	The normalized thermal factor satisfies
	\begin{align}
		&\frac{\beta\nu}{\delta}
		\left(
		\frac{1}{\e^{\beta\nu(z-\frac{\delta}{2})}-1}
		-
		\frac{1}{\e^{\beta\nu(z+\frac{\delta}{2})}-1}
		\right)=
		\frac1\delta
		\int_{z-\frac{\delta}{2}}^{z+\frac{\delta}{2}}
		\d x\,\frac{(\beta\nu)^2}
		{4\sinh^2\frac{\beta\nu x}{2}} \notag \\
		=& 	\frac{(\beta\nu)^2}{4\sinh^2\frac{\beta\nu z}{2}}+O\left(\frac{\delta^2}{z^4}\right),
		\qquad z\geq\delta,
	\end{align} where we used that the linear term in $\delta$ vanishes due to the symmetry under $\delta\to -\delta$ and the uniform estimates
	\begin{equation}
		0<
		\frac{(\beta\nu)^2}{4\sinh^2\frac{\beta\nu z}{2}}
		\leq\frac1{z^2},
		\qquad
		\frac{\d^2}{\d z^2}
		\left(
		\frac{(\beta\nu)^2}
		{4\sinh^2\frac{\beta\nu z}{2}}
		\right)
		=O(\delta^0 z^{-4})\text{ as }\delta\to 0,
	\end{equation}
	with a constant independent of $\beta\nu$.
	
	Thus, on $[\delta,+\infty)$, the preceding estimates give an error, as $\delta\to 0$,  bounded by
	\begin{equation}
		O\left(
		\delta^2
		\int_\delta^{+\infty}\d z\,
		\left[
		\frac{\min\{z^2,1\}}{z^2}
		+
		\frac{\delta^2\min\{z^2,1\}}{z^4}
		\right]
		\right)
		=O(\delta^2),
	\end{equation}
	uniformly in $\beta\nu>0$.

	For
	$\frac{\delta}{2}\leq z\leq\delta$, the normalized expression \eqref{normalized expression}
	involving $G$ and $f$ is
	$O\bigl((z^2-\frac{\delta^2}{4})^2\bigr)$, while the 
	thermal factor is bounded above by
	$(z^2-\frac{\delta^2}{4})^{-1}$. Therefore, the total contribution to the
	integral is  $O(\delta^3)$.

	we conclude that
	\begin{align*}
		\gamma_\LL(\k;\beta,\nu)
		&=
		\frac{\hat{v}(\0)\nu^{3/2}}{4\pi}
		\frac{1}{\beta\nu}\frac{k}{\sqrt{\nu}}
		\Biggl[
		\int_0^{+\infty}\frac{\d z}{2z^2}
		\left(
		2-\frac{1}{\sqrt{1+z^2}}-\frac{1}{1+z^2}
		\right)^2
		\\
		&\hspace{35mm}\times
		\left(
		\frac{\beta\nu z}
		{2\sinh\frac{\beta\nu z}{2}}
		\right)^2
		+O(\delta^2)
		\Biggr],
	\end{align*}
	with a remainder uniform in $\beta\nu>0$.
	
	We now estimate the high-temperature error.
	Since
	\[
	0\leq
	\frac12
	\left(
	2-\frac{1}{\sqrt{1+z^2}}-\frac{1}{1+z^2}
	\right)^2
	\leq2,
	\]
	we have
	\begin{align*}
		0
		&\leq
		\int_0^{+\infty}\frac{\d z}{2z^2}
		\left(
		2-\frac{1}{\sqrt{1+z^2}}-\frac{1}{1+z^2}
		\right)^2
		\left[
		1-
		\left(
		\frac{\beta\nu z}
		{2\sinh(\beta\nu z/2)}
		\right)^2
		\right]
		\\
		&\leq
		2\int_0^{+\infty}\d z\,
		\left[
		\frac1{z^2}
		-
		\frac{(\beta\nu)^2}
		{4\sinh^2(\beta\nu z/2)}
		\right]
		=\beta\nu.
	\end{align*}
	Thus the temperature error in the normalized integral
	is $O(\beta\nu)$.
	
	Finally, the limiting integral is
	\begin{align}
		&\int_0^{+\infty}\frac{\d z}{2z^2}
		\left(
		2-\frac{1}{\sqrt{1+z^2}}-\frac{1}{1+z^2}
		\right)^2\notag
		\\
		=&
		\left[
		\frac34\arctan z
		+\frac1z
		\left(\frac1{\sqrt{1+z^2}}-1\right)
		-\frac{z}{4(1+z^2)}
		\right]_0^{+\infty}
		=\frac{3\pi}{8}.
	\end{align}
	Combining these estimates yields
	\begin{align}
		\gamma_\LL(\k;\beta,\nu)
		=
		\frac{3\hat{v}(\0)\nu^{3/2}}{32}
		\frac{1}{\beta\nu}\frac{|\k|}{\sqrt{\nu}}
		\left(
		1+O\left(\frac{|\k|^2}{\nu}+\beta\nu\right)
		\right),\notag
		\\
		\frac{|\k|}{\sqrt{\nu}},\,\beta\nu\to0,
	\end{align}
	uniformly, without any restriction on the relative
	sizes of $\frac{|\k|}{\sqrt{\nu}}$ and $\beta\nu$.
\end{proof}

\begin{remark}
	In addition to the previous consideration regarding the well-foundedness of the hypotheses $T<T_c$, we observe how Equation \eqref{eq: high temp landau rate} is divergent to the leading order as $\beta\nu\to 0$, which can be seen as a manifestation of the non validity of the high temperature limit.
\end{remark}

\paragraph{\textbf{Acknowledgements.}} L. P. is grateful for the support of the National Group of Mathematical Physics (GNFM-INdAM).
\vspace{5mm}

\paragraph{\textbf{Data availability statement.}}
Data sharing is not applicable to this article as no new data were created or analysed in this study.
\vspace{5mm}

\paragraph{\textbf{Conflict of interest statement.}} The authors certify that they have no affiliations with or involvement in any
organization or entity with any financial interest or non-financial interest in
the subject matter discussed in this manuscript.

\end{document}